\documentclass[12pt,a4paper]{article}
\usepackage[english]{babel}
\usepackage{amsmath,amsfonts, amsthm, amssymb}

\newtheorem{thm}{Theorem}[section]
\newtheorem{cor}[thm]{Corollary}
\newtheorem{lem}[thm]{Lemma}
\theoremstyle{definition}
\newtheorem{defin}{Definition}[section]
\theoremstyle{remark}
\newtheorem*{remark}{Remark}

\author{Andr\'es E. Ace\~na\\
{\small Max Planck Institute for Gravitational Physics,}\\
{\small Am M\"uhlenberg 1, D-14476 Golm, Germany}\\
{\small E-mail: acena@aei.mpg.de}}
\date{}
\title{Convergent null data expansions at space-like infinity of stationary vacuum solutions}

\begin{document}
\maketitle
\begin{abstract}
We present a characterization of the asymptotics of all asymptotically flat stationary vacuum solutions of Einstein's field equations. This characterization is given in terms of two sequences of symmetric trace free tensors (we call them  the `null data'), which determine a formal expansion of the solution, and which are in a one to one correspondence to Hansen's multipoles. We obtain necessary and sufficient growth estimates on the null data to define  an absolutely convergent series in a neighbourhood of spatial infinity. This provides a complete characterization of all asymptotically flat stationary vacuum solutions to the field equations.
\end{abstract}

\section{Introduction}
A stationary vacuum spacetime is given by $(\tilde{M},\tilde{g}_{\mu\nu},\xi^{\mu})$, where $\tilde{M}$ is a four-dimensional manifold, $\tilde{g}_{\mu\nu}$ is a Lorentzian metric with signature $(+---)$ that satisfy Einstein's vacuum equations (i.e. $Ric[\tilde{g}]=0$), and $\xi^{\mu}$ is a time-like Killing vector field with complete orbits. The metric can be written locally as
\begin{equation}\label{stationaryMetric}
\tilde{g}=V(dt+\gamma_ad\tilde{x}^a)^2+V^{-1}\tilde{h}_{ab}d\tilde{x}^ad\tilde{x}^b,\,\,\,\,\,a,b=1,2,3,
\end{equation}
where $V$, $\gamma_a$ and $\tilde{h}_{ab}$ depend only on the spatial coordinates $\tilde{x}^a$. As shown by Geroch \cite{Geroch71} the description of this spacetime can be done in terms of fields defined in an abstract three-dimensional manifold $\tilde{N}$ defined as  the quotient space of $\tilde{M}$ with respect to the trajectories of $\xi^{\mu}$. The fields $V$, $\gamma_a$, $\tilde{h}_{ab}$
on $\tilde{M}$ can be obtained as pull-backs of fields on $\tilde{N}$ under the projection map. The latter will be
denoted by the same symbols, being $\tilde{h}_{ab}$ the negative definite metric on $\tilde{N}$. In the following we shall only work
on $\tilde{N}$.\\
The vacuum Einstein's field equations in $\tilde{M}$ imply that on $\tilde{N}$ the quantity
\begin{equation*}
\omega_a=-V^2\tilde{\epsilon}_{abc}\tilde{D}^b\gamma^c
\end{equation*}
is curl-free, i.e.
\begin{equation*}
\tilde{D}_{[a}\omega_{b]}=0,
\end{equation*}
where $\tilde{D}$ is the covatiant derivative with respect to $\tilde{h}_{ab}$ and $\tilde{\epsilon}_{abc}=\tilde{\epsilon}_{[abc]}$, $\tilde{\epsilon}_{123}=|\det \tilde{h}_{ab}|^\frac{1}{2}$. We are interested in the asymptotics of the space-time at spatial infinity, so it will be assumed that $\tilde{N}$ is diffeomorphic to the complement of a closed ball $\bar{B}_R(0)$ in $\mathbb{R}^3$. Thus $\tilde{N}$ is simply connected and hence there exists a scalar field $\omega$ such that
\begin{equation*}
\tilde{D}_a\omega=\omega_a.
\end{equation*}
Instead of working with $V$ and $\omega$ it is convenient to use the combinations
\begin{equation*}
\tilde{\phi}_M=\frac{V^2+\omega^2-1}{4V},
\end{equation*}
\begin{equation*}
\tilde{\phi}_S=\frac{\omega}{2V},
\end{equation*}
introduced by Hansen \cite{Hansen74}. Einstein's vacuum field equations in this setting are equivalent to
\begin{equation}\label{stationaryEq1}
\Delta_{\tilde{h}}\tilde{\phi}_{A}=2R[\tilde{h}]\tilde{\phi}_{A},\;\;\;\;\mbox{\footnotesize{$A=M,S,$}}
\end{equation}
\begin{equation}\label{stationaryEq2}
R_{ab}[\tilde{h}]=2[(\tilde{D}_{a}\tilde{\phi}_{M})(\tilde{D}_{b}\tilde{\phi}_{M})+(\tilde{D}_{a}\tilde{\phi}_{S})(\tilde{D}_{b}\tilde{\phi}_{S})-(\tilde{D}_{a}\tilde{\phi}_{K})(\tilde{D}_{b}\tilde{\phi}_{K})],
\end{equation}
where $\tilde{\phi}_K=\left(\tfrac{1}{4}+\tilde{\phi}_M^2+\tilde{\phi}_S^2\right)^\frac{1}{2}$. Equations \eqref{stationaryEq1}, \eqref{stationaryEq2} will be referred to as the \emph{stationary vacuum field equations}. Having $(\tilde{M},\tilde{g}_{\mu\nu},\xi^{\mu})$ is equivalent to having $(\tilde{N},\tilde{h}_{ab},\tilde{\phi}_M,\tilde{\phi}_S)$. We are looking for solutions of \eqref{stationaryEq1} and \eqref{stationaryEq2}.\\
The asymptotic flatness condition is usually stated by assuming $(\tilde{N},\tilde{h}_{ab})$ to admit a smooth conformal extension in the following way: there exist a smooth Riemannian manifold $(N,h_{ab})$ and a function $\Omega\in C^2(N)\cap C^{\infty}(\tilde{N})$ such that $N=\tilde{N}\cup\{i\}$, where $i$ is a single point,
\begin{equation*}
\Omega>0\mbox{  on  }\tilde{N},
\end{equation*}
\begin{equation*}
h_{ab}=\Omega^2\tilde{h}_{ab}\mbox{ on }\tilde{N},
\end{equation*}
\begin{equation}\label{condOmega}
\Omega|_i=0,\;\;D_{a}\Omega|_i=0,\;\;D_{a}D_{b}\Omega|_i=2h_{ab}|_i,
\end{equation}
where $D$ is the covariant derivative operator defined by $h$. This makes $N$ diffeomorphic to an open ball in $\mathbb{R}^3$, with center at the point $i$, which represents space-like infinity. From now on we assume $\tilde{N}$ to be asymptotically flat in the stated sense.\\
Considering that $\tilde{N}$ is diffeomorphic to the complement of a closed ball $\bar{B}_R(0)$ in $\mathbb{R}^3$ is natural in the present context. It corresponds to the idea of an isolated system, where the material sources are confined to a bounded region outside of which is vacuum. Lichnerowicz \cite{Lichnerowicz52} has shown that if $\tilde{N}$ is diffeomorphic to $\mathbb{R}^3$ then $\tilde{N}$ is flat.\\
Reula \cite{Reula89} has shown existence and uniqueness of asymptotically flat solutions to \eqref{stationaryEq1}, \eqref{stationaryEq2}, in terms of a boundary value problem, when data are prescribed on the sphere $\partial\tilde{N}$.\\
In order to be able to control the precise asymptotic behaviour of the spacetime it would be convenient to have a complete description of the asymptotically flat stationary vacuum solutions in terms of asymptotic quantities. Candidates for this task are Hansen's multipoles \cite{Hansen74}. With the previous assumptions Hansen proposes a definition of multipoles, which extends Geroch's definition of multipoles for asymptotically flat static space-times \cite{Geroch70} to the stationary case. He defines the conformal potentials
\begin{equation}\label{conformalPotentials}
\phi_{A}=\Omega^{-\frac{1}{2}}\tilde{\phi}_{A},\,\,\,\mbox{\footnotesize{$A=M,S,$}}
\end{equation}
and two sequences of tensor fields near $i$ through
\begin{eqnarray}
\label{multipoleField1} P^{A}=\phi_{A},\,\,\,P^{A}_a=D_aP^A,\,\,\,P^{A}_{a_2a_1}={\cal C}\left(D_{a_2}P^A_{a_1}-\tfrac{1}{2}P^AR_{a_2a_1}\right),\\
\label{multipoleField2} P^{A}_{a_{s+1}...a_{1}}={\cal C}\left[D_{a_{s+1}}P^{A}_{a_{s}...a_{1}}-\tfrac{1}{2}s(2s-1)P^{A}_{a_{s+1}...a_{3}}R_{a_{2}a_{1}}\right],\,\,\,\mbox{\footnotesize{$A=M,S,$}}
\end{eqnarray}
where $R_{ab}$ is the Ricci tensor of $h_{ab}$ and ${\cal C}$ is the projector onto the symmetric trace free part of the respective tensor fields. The multipole moments are then defined as the tensors
\begin{equation}\label{multipoles}
\nu^A=P^A(i),\,\,\,\nu^A_{a_p...a_1}=P^A_{a_p...a_1}(i),\,\,\,\mbox{\footnotesize{$A=M,S,$}}\,\,p=1,2,3,...
\end{equation}
Keeping aside the monopoles, $\nu^A$, we will denote the two sequences of remaining multipoles by
\begin{equation*}
{\cal D}^A_{mp}=\{\nu^A_{a_1},\nu^A_{a_2a_1},\nu^A_{a_3a_2a_1},...\},\,\,\,\mbox{\footnotesize{$A=M,S$}}.
\end{equation*}
The multipole moments are proposed as a way to characterize solutions of \eqref{stationaryEq1}, \eqref{stationaryEq2}. So a natural question is to what extent do the multipoles determine the metric $h$ and the potentials $\phi_M$, $\phi_S$. For this to be the case the metric and the potentials should be real analytic even at $i$ in suitable coordinates and conformal rescaling. Beig and Simon \cite{BeigSimon81} and Kundu \cite{Kundu81} have shown that the metric and the potentials do extend in a suitable gauge as real analytic fields to $i$ if it is assumed that
\begin{equation*}
(\nu^M)^2+(\nu^S)^2\neq 0.
\end{equation*}
As explained in \cite{Simon80} (cf. also \cite{BeigSimon83}), in order for a solution of \eqref{stationaryEq1}, \eqref{stationaryEq2} to lead to an asymptotically flat space-time $\tilde{M}$ it is necessary that $\nu^S=0$. So, we assume from now on that
\begin{equation}\label{notZeroCondition}
\nu^M\neq 0,\,\,\,\,\,\nu^S=0.
\end{equation}
In \cite{BeigSimon81} and \cite{Kundu81} it is also shown that for given multipoles there is a unique formal expansion of a `formal solution' to the stationary field equations, but it is not touched upon the convergence of the expansion.\\
B\"ackdahl and Herberthson \cite{BackdahlHerberthson06} have found, assuming a given asymptotically flat solution of the stationary field equations, necessary bounds on the multipoles.\\
The question that remains open is under which conditions a pair of sequences, taken as the multipoles, do indeed determine a convergent expansion of a stationary solution. This question has been studied for the axisymmetric case by B\"ackdahl \cite{Backdahl07}. In the static case, where there is only one sequence of multipoles, Friedrich \cite{Friedrich07} has used as data a sequence of trace-free symmetric tensors, different but related to the multipoles. He has shown that imposing certain types of estimates on the data he prescribes is necessary and sufficient for the existence of asymptotically flat static space-times. However, so far the question for the general case has never been answered.\\
The purpose of this work is to derive, under the assumption \eqref{notZeroCondition}, necessary and sufficient conditions for certain minimal sets of asymptotic data, different to the multipoles, denoted collectively by ${\cal D}_n^\phi$, ${\cal D}_n^S$, and referred to as \emph{null data}, to determine (unique) real analytic solutions of \eqref{stationaryEq1} and \eqref{stationaryEq2} and thus to provide a complete characterization of all possible asymptotically flat solutions to the stationary vacuum field equations.\\
In the following we shall work in terms of the conformally rescalled fields, the conformal factor will be specified in more detail later
on.\\
For the same reasons that justify $\tilde{N}$ to be considered diffeomorphic to the complement of a closed ball in $\mathbb{R}^3$, we shall treat the case in which $N$ may comprise a small neighbourhood of the point $i$, without worring about the behaviour of the solution in the large (note that in terms of $\tilde{h}$ a neighbourhood of $i$ cover an infinite domain extending to space-like infinity). This work generalizes the work by Friedrich \cite{Friedrich07} from the static to the stationary case in a way discussed later on this section.\\
The multipoles are defined for any conformal gauge, but for our analysis it is convenient to remove the conformal gauge freedom and use, following Beig and Simon \cite{BeigSimon81},
\begin{equation}\label{conformalFactor}
\Omega=\frac{1}{2}m^{-2}\left[\left(1+4\tilde{\phi}_M^2+4\tilde{\phi}_S^2\right)^\frac{1}{2}-1\right].
\end{equation}
With this conformal factor they derive fall-off conditions and then show that under some assumptions the rescaled metric can be extended on a suitable neighbourhood of space-like infinity and in suitable coordinates as a real analytic metric at $i$. The potentials $\phi_M$ and $\phi_S$ are then also real analytic at $i$, so that the multipoles are well defined. Using this gauge, and taking into account that the angular momentum monopole vanish, we get
\begin{equation*}
\nu^M=m,\,\,\,\nu^M_a=0.
\end{equation*}
Instead of using the multipoles sequences, it will be convenient for our analysis to use, in the given gauge, the following two sequences
\begin{equation*}
{\cal D}^{\phi}_n=\{{\cal C}(D_{a_1}\phi)(i),{\cal C}(D_{a_2}D_{a_1}\phi)(i),{\cal C}(D_{a_3}D_{a_2}D_{a_1}\phi)(i),...\},
\end{equation*}
\begin{equation}\label{nullDataSgen}
{\cal D}^S_n=\{S_{a_2a_1}(i),{\cal C}(D_{a_3}S_{a_2a_1})(i),{\cal C}(D_{a_4}D_{a_3}S_{a_2a_1})(i),...\},
\end{equation}
where $\phi=\phi_S$ and $S_{ab}$ is the trace free part of the Ricci tensor of $h$.\\
We express now the tensors in ${\cal D}^{\phi}_n,{\cal D}^S_n$ in terms of an $h$-orthonormal frame field $c_{\bf a},\,{\bf a}=1,2,3$, near $i$, which is $h$-parallelly propagated along the geodesics through $i$, denoting by $D_{\bf a}$ the covariant derivative in the direction of $c_{\bf a}$, and write
\begin{equation}\label{nullDataPhi}
{\cal D}^{\phi*}_n=\{{\cal C}(D_{{\bf a}_1}\phi)(i),{\cal C}(D_{{\bf a}_2}D_{{\bf a}_1}\phi)(i),{\cal C}(D_{{\bf a}_3}D_{{\bf a}_2}D_{{\bf a}_1}\phi)(i),...\},
\end{equation}
\begin{equation}\label{nullDataS}
{\cal D}^{S*}_n=\{S_{{\bf a}_2{\bf a}_1}(i),{\cal C}(D_{{\bf a}_3}S_{{\bf a}_2{\bf a}_1})(i),{\cal C}(D_{{\bf a}_4}D_{{\bf a}_3}S_{{\bf a}_2{\bf a}_1})(i),...\}.
\end{equation}
These tensors are defined uniquely up to rigid rotations in $\mathbb{R}^3$. These two series will be referrred to as the \emph{null data of $h$ in the frame $c_{\bf a}$}.\\
For a real analytic metric $h$ near $i$ there exist constants $M,r>0$ such that the components of these tensors satisfy the estimates
\begin{equation*}
|{\cal C}(D_{{\bf a}_p}...D_{{\bf a}_1}\phi)(i)|\leq \frac{Mp!}{r^p},\,\,\,{\bf a}_p,...,{\bf a}_1=1,2,3,\,\,\,p=0,1,2,...,
\end{equation*}
\begin{equation}\label{estimatesDerS}
|{\cal C}(D_{{\bf a}_p}...D_{{\bf a}_1}S_{{\bf b}{\bf c}})(i)|\leq \frac{Mp!}{r^p},\,\,\,{\bf a}_p,...,{\bf a}_1,{\bf b},{\bf c}=1,2,3,\,\,\,p=0,1,2,....
\end{equation}
Althoug these estimates have similar form to Cauchy estimates they are not the same, the difference being that here the estimates are on the symmetric trace free part of the derivatives instead of being directly on the derivatives. These estimates are derived from Cauchy estimates in Section \ref{exactSets}. Remarkably, the converse is also true. This consitutes our main result, given in the following theorem.\\
\begin{thm}\label{mainThm}
Suppose $m\neq 0$ and
\begin{equation}\label{abstractNullDataPhi}
\hat{{\cal D}}^{\phi}_n=\{\psi_{{\bf a}_1},\psi_{{\bf a}_2{\bf a}_1},\psi_{{\bf a}_3{\bf a}_2{\bf a}_1},...\},
\end{equation}
\begin{equation}\label{abstractNullDataS}
\hat{{\cal D}}^S_n=\{\Psi_{{\bf a}_2{\bf a}_1},\Psi_{{\bf a}_3{\bf a}_2{\bf a}_1},\Psi_{{\bf a}_4{\bf a}_3{\bf a}_2{\bf a}_1},...\},
\end{equation}
are two infinite sequences of symmetric, trace free tensors given in an orthonormal frame at the origin of a 3-dimensional Euclidean space. If there exist constants $M,r>0$ such that the components of these tensors satisfy the estimates
\begin{equation*}
|\psi_{{\bf a}_p...{\bf a}_2{\bf a}_1}|\leq \frac{Mp!}{r^p},\,\,\,{\bf a}_p,...,{\bf a}_1=1,2,3,\,\,\,p=1,2,...,
\end{equation*}
\begin{equation*}
|\Psi_{{\bf a}_p...{\bf a}_2{\bf a}_1{\bf b}{\bf c}}|\leq \frac{Mp!}{r^p},\,\,\,{\bf a}_p,...,{\bf a}_1,{\bf b},{\bf c}=1,2,3,\,\,\,p=0,1,2,...,
\end{equation*}
then there exists an analytic, asymptotically flat, stationary vacuum solution $(\tilde{h},\tilde{\phi}_M,\tilde{\phi}_S)$ with mass monopole $m$ and zero angular momentum monopole, unique up to isometries, so that the null data implied by $h=\tfrac{1}{4}m^{-4}[(1+4\tilde{\phi}_M^2+4\tilde{\phi}_S^2)^\frac{1}{2}-1]^2\tilde{h}$ and $\phi_{S}=2^\frac{1}{2}m[(1+4\tilde{\phi}_M^2+4\tilde{\phi}_S^2)^\frac{1}{2}-1]^{-\frac{1}{2}}\tilde{\phi}_S$ in a suitable frame $c_{\bf a}$ as described above satisfy
\begin{eqnarray*}
{\cal C}(D_{{\bf a}_q}...D_{{\bf a}_1}\phi_S)(i)=\psi_{{\bf a}_q...{\bf a}_1},\,\,\,{\bf a}_q,...,{\bf a}_1=1,2,3,\,\,\,q=1,2,...,
\end{eqnarray*}
\begin{eqnarray*}
{\cal C}(D_{{\bf a}_q}...D_{{\bf a}_3}S_{{\bf a}_2{\bf a}_1})(i)=\Psi_{{\bf a}_q...{\bf a}_1},\,\,\,{\bf a}_q,...,{\bf a}_1=1,2,3,\,\,\,q=2,3,... .
\end{eqnarray*}
\end{thm}
Two sequences of data of the form \eqref{abstractNullDataPhi}, \eqref{abstractNullDataS}, not necessarily satisfying any estimates, will be referred to as \emph{abstract null data}.\\
The type of estimates imposed here on the abstract null data does not depend on the orthonormal frame in which they are given. Since these estimates are necessary as well as sufficient, all possible asymptotically flat solutions of the stationary vacuum field equations are characterized by the null data.\\
In relation with the works of Corvino and Schoen \cite{CorvinoSchoen06} and Chru\'sciel and Delay \cite{ChruscielDelay03}, as they need a family of asymptotically flat stationary solutions to perform the gluing procedure, this result gives a complete survey of the possible stationary asymptotics that can be attained, beyond the known exact solutions.\\
As both the multipoles and the null data determine the metric and the potentials then there is a bijective map between them.  Thus the sets ${\cal D}_n^\phi$, ${\cal D}_n^S$ and ${\cal D}_{mp}^M$, ${\cal D}_{mp}^S$ contain the same information. We prefer to work with the null data because the expressions are linear in $\phi$ and $S_{ab}$.\\
This work contains the static case as a special case. Starting from \eqref{stationaryMetric} the static case can be attained by making $\gamma_a=0$, which gives $\omega=0$, $\tilde{\phi}_S=0$ and $\phi_S=0$. This implies that all tensors in ${\cal D}_n^\phi$ are zero. Conversely, if all tensors in ${\cal D}_n^\phi$ are zero then all tensors in ${\cal D}_{mp}^S$ are zero and by Xanthopoulos' work \cite{Xanthopoulos79} the space-time is static. So we are left with ${\cal D}_{n}^S$ as the free data in the static case.\\
Friedrich \cite{Friedrich07} has given the same result for the static case using a different conformal metric. Let us assume for now that we are in the static case, then Friedrich uses a metric $\breve{h}$, wich is conformally related to our metric $h$ by
\begin{equation}\label{metricFriedrich}
\breve{h}=\breve{\Omega}^2h,
\end{equation}
where
\begin{equation*}
\breve{\Omega}=\frac{4\left[(1+m^2\Omega)^\frac{1}{2}+m\Omega^\frac{1}{2}\right]}{\left[(1+m^2\Omega)^\frac{1}{2}+m\Omega^\frac{1}{2}+1\right]^2}.
\end{equation*}
Using $\breve{h}$ he defines a sequence of symmetric trace-free tensors $\breve{\cal D}_n$ in the same way as we defined ${\cal D}_{n}^S$ \eqref{nullDataSgen}. He shows that impossing estimates of the type \eqref{estimatesDerS} on the tensors in $\breve{\cal D}_n$ is necessary and sufficient for the existence of an asymptotically flat static vacuum solution of the Einstein's equations. To see that this result is equivalent to our result in the static case, we have to show that having estimates of the type \eqref{estimatesDerS} on the tensors in ${\cal D}_{n}^S$ imply estimates of the same type on the tensors in $\breve{\cal D}_n$ and vice versa. This is done through the theorem \ref{mainThm} and the relation \eqref{metricFriedrich}. If the tensors in ${\cal D}_{n}^S$ satisfy estimates of the type \eqref{estimatesDerS} then there exist $h$ and $\Omega$ analytic, and then $\breve{h}$ given by \eqref{metricFriedrich} is also analytic, thus the tensors in $\breve{\cal D}_n$ satisfy estimates of the type \eqref{estimatesDerS}, the converse is shown in the same way using Friedrich's result. Hence this work generalizes the work by Friedrich \cite{Friedrich07} from the static to the stationary case. The procedure that we use in the present work follows similar steps and several of the technics in \cite{Friedrich07} will be used.  For completeness we include them.\\
Theorem \ref{mainThm} will be proven in terms of the conformal metric $h$. Thus we shall express in Section \ref{conformalSetting} the stationary vacuum field equations as \emph{`conformal stationary vacuum field equations'}. In Section \ref{exactSets} we show, by going to space-spinor formalism, that the abstract null data indeed determine the expansion coefficients of a formal expansion of a solution to the conformal stationary vacuum field equations. Showing convergence in this way appears difficult, and for this reason the problem is cast in a certain setting, which is necessarily singular at a certain subset of the manifold, as a characteristic initial value problem in Section \ref{characteristicProblem}. In Section \ref{sectionConformalEquations} it is shown how to determine a formal solution to a subset of the confomal field equations from a given set of abstract null data. Then, in Section \ref{convergence}, the convegence of the obtained series is shown. In Section \ref{completeSet} it is shown that the obtained solution satisfy the full set of conformal field equations. Finally, in Section \ref{analyticity}, the convergence result is translated into a gauge which is regular near $i$, allowing us to prove Theorem \ref{mainThm}.

\section{The stationary field equations in the conformal setting}\label{conformalSetting}
The existence problem will be analyzed completely in terms of the conformally rescaled metric $h$, so we need to express the stationary field equations in terms of the conformal fields. If we directly transform the fields in \eqref{stationaryEq1} and \eqref{stationaryEq2} we arrive at a system of equations that is singular at $i$. To overcome this problem we follow the work of Beig and Simon \cite{BeigSimon81}. Using \eqref{conformalFactor} as the conformal factor, where by a constant conformal rescaling it can allways be achieved $m=1$ and for simplicity we use this scale from now, and standard formulae for conformal transformations they manipulate the stationary field equations, arriving at the following equivalent system of equations:
\begin{equation}\label{Omega}
\Omega:=\phi_{M}^2+\phi_{S}^2-1,
\end{equation}
\begin{equation}\label{pi}
\pi_{ab}:=D_{a}\phi_{M}D_{b}\phi_{M}+D_{a}\phi_{S}D_{b}\phi_{S},
\end{equation}
\begin{equation*}
\Delta\phi_{A}=-\frac{1}{2}\left[R-\frac{5}{2}D_a\Omega D^a\Omega+10(1+\Omega)\pi_a\,^a\right]\phi_{A},\,\,\,\mbox{\scriptsize{A=M,S}},
\end{equation*}
\begin{eqnarray*}
D_{a}D_{b}\Omega & = & -\Omega R_{ab}-\frac{1}{3}h_{ab}R+\left(\Omega+\frac{2}{3}\right)h_{ab}D_c\Omega D^c\Omega\\
&&-4\left(\Omega+\frac{2}{3}\right)(\Omega+1)h_{ab}\pi_c\,^c-\frac{1}{2}(\Omega-1)D_{a}\Omega D_{b}\Omega+2\Omega^{2}\pi_{ab},
\end{eqnarray*}
\begin{eqnarray*}
D_{a}R & = & 7D^{b}\Omega D_{a}D_{b}\Omega+3R_{ab}D^{b}\Omega+4(3\Omega-2)\pi_b\,^bD_{a}\Omega\\
&&-\frac{3}{2}D_b\Omega D^b\Omega D_{a}\Omega-6\Omega\pi_{ab}D^{b}\Omega-2(7\Omega+4)D_{a}\pi_b\,^b,
\end{eqnarray*}
\begin{eqnarray*}
D_{[c}R_{b]a} & = & 2(3\Omega-1)\pi_d\,^d h_{a[b}D_{c]}\Omega-h_{a[b}D_{c]}\Omega D_d\Omega D^d\Omega\\
&&-2(\Omega-1)h_{a[b}\pi_{c]d}D^{d}\Omega-2(2\Omega+1)h_{a[b}D_{c]}\pi_d\,^d\\
&&+2h_{a[b}D_{c]}D_{d}\Omega D^{d}\Omega+\frac{1}{2}D_{[c}\Omega D_{b]}D_{a}\Omega-(\Omega-4)\pi_{a[b}D_{c]}\Omega\\
&&+2\Omega D_{[c}\pi_{b]a}+\frac{1}{2}R_{a[b} D_{c]}\Omega+h_{a[b}R_{c]d} D^{d}\Omega.
\end{eqnarray*}
These equations are regular even at $i$. They form a quasi-linear, overdetermined system of PDE's which implies, by applying formal derivatives to some of the equations, elliptic equations for all unknowns in a suitable gauge. Considering the fall-off conditions on the fields Beig and Simon \cite{BeigSimon81} deduced a certain smoothness of the conformal fields at $i$. Invoking a general theorem of Morrey on elliptic systems of this type they concluded that the solutions are in fact real analytic at $i$. Later Kennefick and O'Murchadha \cite{KennefickO'Murchadha95} showed that the fall-off conditions are reasonable, as they are implied by the space-time being asymptotically flat. To avoid introducing additional constraints by taking derivatives, we shall deal with the system as it is.\\
For our pourposes it is convenient to make some changes to this system. We separate the Ricci tensor into its trace free part and the Ricci scalar,
\begin{equation*}
R_{ab}=S_{ab}+\frac{1}{3}h_{ab}R.
\end{equation*}
We also get rid of $\pi_{ab}$ by using \eqref{pi} in the other equations. From \eqref{Omega} we see that $\Omega$, $\phi_M$ and $\phi_S$ are not independent, we use this equation to get rid of $\phi_M$ in the other equations. With these changes and the change of notation $\phi_S\rightarrow\phi$ the system of equations takes the form
\begin{eqnarray}\label{D2phiT}
\Delta\phi & = & -\phi\left\{\frac{1}{2}R+\frac{5}{1+\Omega-\phi^2}\left[\frac{1}{4}\phi^2D^a\Omega D_a\Omega\right.\right.\\
\nonumber&& \left.\left.\frac{}{}-(1+\Omega)\phi D^a\Omega D_a\phi+(1+\Omega)^2D^a\phi D_a\phi\right]\right\},
\end{eqnarray}
\begin{eqnarray}\label{DDOmegaT}
D_{a}D_{b}\Omega &= & -\Omega S_{ab}-\frac{1}{3}(1+\Omega)h_{ab}R\\
\nonumber&& +\frac{1}{1+\Omega-\phi^2}\left\{\frac{1}{2}\left[1+(-1+\Omega)\phi^2\right]D_{a}\Omega D_{b}\Omega\right.\\
\nonumber&& -\frac{1}{3}(2+3\Omega)\phi^2 h_{ab}D^c\Omega D_c\Omega-2\Omega^2 \phi D_{(a}\Omega D_{b)}\phi\\
\nonumber&& +\frac{4}{3}(1+\Omega)(2+3\Omega)\phi h_{ab}D^c\Omega D_c\phi+2\Omega^2 (1+\Omega)D_{a}\phi D_{b}\phi\\
\nonumber&& \left.-\frac{4}{3}(1+\Omega)^2(2+3\Omega)h_{ab}D^c\phi D_c\phi\right\},
\end{eqnarray}
\begin{eqnarray}\label{DRT}
&& D_{a}R\\
\nonumber&& =\frac{1}{1+\Omega-\phi^2}\left\{\frac{}{}2(4+7\Omega)\phi D^{b}\Omega D_{b}D_{a}\phi-4(1+\Omega)(4+7\Omega)D^{b}\phi D_{b}D_{a}\phi\right.\\
\nonumber&& +\left[3+(-3+7\Omega)\phi^2\right]D^{b}\Omega S_{ba}-2\Omega(4+7\Omega)\phi D^{b}\phi S_{ba}\\
\nonumber&& \left.+\frac{1}{3}(4+7\Omega)\phi^2 R D_{a}\Omega-\frac{2}{3}(1+\Omega)(4+7\Omega)\phi R D_{a}\phi \right\}\\
\nonumber&& +\frac{1}{3(1+\Omega-\phi^2)^2}\left\{\frac{1}{2}\phi^2 \left[-12+(40+21\Omega)\phi^2\right]D^b\Omega D_b\Omega D_{a}\Omega\right.\\
\nonumber&& -2\phi \left[-18(1+\Omega)+(46+61\Omega+21\Omega^2)\phi^2\right]D^b\Omega D_b\phi D_{a}\Omega\\
\nonumber&& +2(1+\Omega)\left[-24(1+\Omega)+(52+61\Omega+21\Omega^2)\phi^2\right]D^b\phi D_b\phi D_{a}\Omega\\
\nonumber&& -\phi \left[12(1+\Omega)+(16+61\Omega+21\Omega^2)\phi^2\right]D^b\Omega D_b\Omega D_{a}\phi\\
\nonumber&& +4(1+\Omega)\left[6(1+\Omega)+(22+61\Omega+21\Omega^2)\phi^2\right]D^b\Omega D_b\phi D_{a}\phi\\
\nonumber&& \left.\frac{}{}-4(1+\Omega)^2 (28+61\Omega+21\Omega^2)\phi D^b\phi D_b\phi D_{a}\phi \right\},
\end{eqnarray}
\begin{eqnarray}\label{DST}
D_{[c}S_{b]a} & = & \frac{1}{1+\Omega-\phi^2}\left\{\frac{}{}\Omega\phi D_{a}D_{[b}\phi D_{c]}\Omega-2\Omega(1+\Omega)D_{a}D_{[b}\phi D_{c]}\phi\right.\\
\nonumber&& -\frac{2}{3}(1+\Omega)\phi h_{a[b}D_{c]}D_{d}\phi D^{d}\Omega+\frac{4}{3}(1+\Omega)^2 h_{a[b}D_{c]}D_{d}\phi D^{d}\phi\\
\nonumber&& +\frac{1}{2}[1+(-1+\Omega)\phi^2]S_{a[b}D_{c]}\Omega-\Omega^2 \phi S_{a[b}D_{c]}\phi\\
\nonumber&& -\frac{1}{3}\Omega\phi^2 h_{a[b}S_{c]d}D^{d}\Omega+\frac{2}{3}\Omega(1+\Omega)\phi h_{a[b}S_{c]d}D^{d}\phi\\
\nonumber&& +\frac{1}{18}(-2+\Omega)\phi^2 R h_{a[b}D_{c]}\Omega-\frac{1}{9}(-2+\Omega)(1+\Omega)\phi R h_{a[b}D_{c]}\phi\\
\nonumber&& \left.+2\phi D_{a}\Omega D_{[b}\Omega D_{c]}\phi-4(1+\Omega)D_{a}\phi D_{[b}\Omega D_{c]}\phi \right\}\\
\nonumber&& +\frac{1}{9(1+\Omega-\phi^2)^2}\left\{\frac{1}{2}\phi^2 [3+2(-5+3\Omega)\phi^2] h_{a[b}D_{c]}\Omega D^d\Omega D_d\Omega\right.\\
\nonumber&& -\phi [6(1+\Omega)+(-13-4\Omega+6\Omega^2)\phi^2] h_{a[b}D_{c]}\phi D^d\Omega D_d\Omega\\
\nonumber&& -2(-7-4\Omega+6\Omega^2)\phi^3 h_{a[b}D_{c]}\Omega D^d\Omega D_d\phi\\
\nonumber&& +4(1+\Omega)^2 [3+2(-5+3\Omega)\phi^2] h_{a[b}D_{c]}\phi D^d\Omega D_d\phi\\
\nonumber&& +2(1+\Omega)[-3(1+\Omega)+(-4-4\Omega+6\Omega^2)\phi^2] h_{a[b}D_{c]}\Omega D^d\phi D_d\phi\\
\nonumber&& \left.\frac{}{}-4(1+\Omega)^2 (-7-4\Omega+6\Omega^2)\phi h_{a[b}D_{c]}\phi D^d\phi D_d\phi\right\}.
\end{eqnarray}
Besides \eqref{D2phiT}, \eqref{DDOmegaT}, \eqref{DRT}, \eqref{DST} we need an equation for the metric or for the frame field and the connection coefficients. This equation is
\begin{equation}\label{RicH}
R_{ab}[h]=S_{ab}+\frac{1}{3}h_{ab}R,
\end{equation}
where the expression on the left hand side is understood as the Ricci operator acting on the metric $h$.\\
The system of equations \eqref{RicH},\eqref{D2phiT},\eqref{DDOmegaT},\eqref{DRT},\eqref{DST}, together with conditions \eqref{condOmega}, which imply
\begin{equation}\label{initialCondR}
R|_i=-\left(6+8D^a\phi D_a\phi\right)|_i,
\end{equation}
will be referred to as the \emph{conformal stationary vacuum field equations} for the unknown fields
\begin{equation}\label{unknownFields}
h_{ab},\,\phi,\,\Omega,\,R,\,S_{ab}.
\end{equation}

\section{The exact sets of equations argument}\label{exactSets}
To see that it is possible to construct solutions to the conformal stationary vacuum field equations from the null data we study expansions of the conformal fields \eqref{unknownFields} in normal coordinates.\\
We assume from now on $N$ to be small enough to coincide with a convex $h$-normal neighbourhood of $i$. Let $c_{\bf a}$, ${\bf a}=1,2,3$, be an $h$-orthonormal frame field on $N$ which is parallelly transported along the $h$-geodesics through $i$ and let $x^a$ denote normal coordinates centered at $i$ so that $c^b\,_{\bf a}\equiv \langle dx^b,c_{\bf a}\rangle=\delta^b\,_{\bf a}$ at $i$. We refer to such a frame as \emph{normal frame centered at} $i$. Its dual frame will be denoted by $\chi^{\bf c}=\chi^{\bf c}\,_b dx^b$. In the following all tensor fields, except the frame field $c_{\bf a}$ and the coframe field $\chi^{\bf c}$, will be expressed in terms of this frame field, so that the metric is given by $h_{\bf ab}\equiv h(c_{\bf a},c_{\bf b})=-\delta_{\bf ab}$. With $D_{\bf a}\equiv D_{c_{\bf a}}$ denoting the covariant derivative in the $c_{\bf a}$ direction, the connection coefficients with respect to $c_{\bf a}$ are defined by $D_{\bf a}c_{\bf c}=\Gamma_{\bf a}\,^{\bf b}\,_{\bf c}c_{\bf b}$.\\
An analytic tensor field $T_{{\bf a}_1...{\bf a}_k}$ on $N$ has in the normal coordinates $x^a$ a \emph{normal expansion} at $i$, which can be written
\begin{equation}\label{normalExpansion}
T_{{\bf a}_1...{\bf a}_k}(x)=\sum_{p\geq 0}\frac{1}{p!}x^{c_p}...x^{c_1}D_{{\bf c}_p}...D_{{\bf c}_1}T_{{\bf a}_1...{\bf a}_k}(i),
\end{equation}
where we assume from now on that the summation convention does not distinguish between bold face and other indices.\\
Since $h_{\bf ab}=-\delta_{\bf ab}$ it remains to be seen how to obtain normal expansions for
\begin{equation}\label{unknownFieldsNormalCoord}
\phi,\,\Omega,\,R,\,S_{\bf ab},
\end{equation}
using the field equations and the null data.\\
The algebra simplifies considerably in the space-spinor formalism. To do the transition we introduce the constant van der Waerden symbols $\alpha^{AB}\,_a$, $\alpha^a\,_{AB}$, $a=1,2,3$, $A,B=0,1$, which are symmetric in $AB$ and whose components, if readed as matrices, are
\begin{eqnarray*}
& \alpha^{AB}\,_1=\frac{1}{\sqrt{2}}\left(\begin{array}{cc}-1 & 0 \\ 0 & 1 \end{array}\right),\,\alpha^{AB}\,_2=\frac{1}{\sqrt{2}}\left(\begin{array}{cc}-i & 0 \\ 0 & -i \end{array}\right),\,\alpha^{AB}\,_3=\frac{1}{\sqrt{2}}\left(\begin{array}{cc} 0 & 1 \\ 1 & 0 \end{array}\right),\\
& \alpha^1\,_{AB}=\frac{1}{\sqrt{2}}\left(\begin{array}{cc}-1 & 0 \\ 0 & 1 \end{array}\right),\,\alpha^2\,_{AB}=\frac{1}{\sqrt{2}}\left(\begin{array}{cc} i & 0 \\ 0 & i \end{array}\right),\,\alpha^3\,_{AB}=\frac{1}{\sqrt{2}}\left(\begin{array}{cc} 0 & 1 \\ 1 & 0 \end{array}\right).
\end{eqnarray*}
The relation between tensors \emph{given in the frame} $c_{\bf a}$ and space-spinors is made by $T^{{\bf a}_1...{\bf a}_p}\,_{{\bf b}_1...{\bf b}_q}\rightarrow T^{A_1B_1...A_pB_p}\,_{C_1D_1...C_qD_q}$, where
\begin{equation*}
T^{A_1B_1...A_pB_p}\,_{C_1D_1...C_qD_q}=T^{{\bf a}_1...{\bf a}_p}\,_{{\bf b}_1...{\bf b}_q}\alpha^{A_1B_1}\,_{a_1}...\alpha^{b_q}\,_{C_qD_q}.
\end{equation*}
With the summation rule also applying to capital indices we get
\begin{eqnarray*}
\delta^b\,_a=\alpha^b\,_{AB}\alpha^{AB}\,_a,\,\,\,-\delta_{ab}\alpha^a\,_{AB}\alpha^b\,_{CD}=-\epsilon_{A(C}\epsilon_{D)B}\equiv h_{ABCD},\\a,b=1,2,3,\,\,A,B,C,D=0,1,
\end{eqnarray*}
where the constant $\epsilon$-spinor satisfies $\epsilon_{AB}=-\epsilon_{BA}$, $\epsilon_{01}=1$. It is used to move indices according to the rules $\iota_B=\iota^A\epsilon_{AB}$, $\iota^A=\epsilon^{AB}\iota_B$, so that $\epsilon_A\,^B$ corresponds to the Kronecker delta.\\
As the spinors are in general complex, we need a way to sort out those that arise from real tensors. For this we define
\begin{equation*}
\tau^{AA'}=\epsilon_0\,^A\epsilon_0\,^{A'}+\epsilon_1\,^A\epsilon_1\,^{A'}.
\end{equation*}
Primed indices take values 0, 1 and the summation rule also applies to them. A bar denotes complex conjugation and indices acquire a prime under complex conjugation, an exception being $\epsilon_{A'B'}$, the complex conjugate of $\epsilon_{AB}$. We define
\begin{equation*}
\xi^+_{A...H}=\tau_A\,^{A'}..\tau_H\,^{H'}\bar{\xi}_{A'...H'}.
\end{equation*}
Then a space spinor field $T_{A_1B_1...A_pB_p}=T_{(A_1B_1)...(A_pB_p)}$ arises from a real tensor field $T_{{\bf a}_1...{\bf a}_p}$ if and only if
\begin{equation}\label{realityCond}
T_{A_1B_1...A_pB_p}=(-1)^pT^+_{A_1B_1...A_pB_p}.
\end{equation}
Any spinor field $T_{A...H}$ admits a decomposition into products of totally symmetric spinor fields and epsilon spinors which can be written schematically in the form
\begin{equation}\label{symmetricCont}
T_{A...H}=T_{(A...H)}+\sum \epsilon's\,\times\,\mbox{\emph{symmetrized contractions of }}T.
\end{equation}
It will be important that if $T_{A_1B_1...A_pB_p}$ arises from $T_{{\bf a}_1...{\bf a}_p}$ then
\begin{equation*}
T_{(A_1B_1...A_pB_p)}={\cal C}(T_{{\bf a}_1...{\bf a}_p})\alpha^{a_1}\,_{A_1B_1}...\alpha^{a_p}\,_{A_pB_p}.
\end{equation*}
To discuss vector analysis in terms of spinors, a complex frame field and its dual 1-form field are defined by
\begin{equation*}
c_{AB}=\alpha^a\,_{AB}c_{\bf a},\,\,\,\chi^{AB}=\alpha^{AB}\,_a\chi^{\bf a},
\end{equation*}
so that $h(c_{AB},c_{CD})=h_{ABCD}$. From this one sees that $c_{00}$ and $c_{11}$ are null vectors orthogonal to $c_{01}$. The derivative of a function $f$ in the direction of $c_{AB}$ is denoted by $c_{AB}(f)=f_{,a}c^a\,_{AB}$ and the spinor connection coefficients are defined by
\begin{equation*}
\Gamma_{AB}\,^C\,_D=\tfrac{1}{2}\Gamma_{\bf a}\,^{\bf b}\,_{\bf c}\alpha^a\,_{AB}\alpha^{CH}\,_b\alpha^c\,_{DH},\,\,\,\Gamma_{ABCD}=\Gamma_{(AB)(CD)},
\end{equation*}
then the covariant derivative of a spinor field $\iota^A$ is given by
\begin{equation*}
D_{AB}\iota^C=c_{AB}(\iota^C)+\Gamma_{AB}\,^C\,_B\iota^B.
\end{equation*}
If it is required to satisfy the Leibniz rule with respect to tensor products, then covariant derivatives in the $c_{\bf a}$-frame formalism translate under contractions with the van de Waerden symbols into spinor covariant derivatives and vice versa.
We also have
\begin{equation}\label{spinorCommutator}
(D_{CD}D_{EF}-D_{EF}D_{CD})\iota^A=R^A\,_{BCDEF}\iota^B,
\end{equation}
\begin{equation}\label{spinorRicci}
R_{ABCDEF}=\tfrac{1}{2}\left[\left(S_{ABCE}-\tfrac{1}{6}Rh_{ABCE}\right)\epsilon_{DF}+\left(S_{ABDF}-\tfrac{1}{6}Rh_{ABDF}\right)\epsilon_{CE}\right],
\end{equation}
where $R$ is the Ricci scalar of $h$ and $S_{ABCD}=S_{\bf ab}\alpha^a_{AB}\alpha^b_{CD}=S_{(ABCD)}$ represents the trace free part of the Ricci tensor of $h$.\\
Equations \eqref{D2phiT}, \eqref{DST} take in the space-spinor formalism the form
\begin{eqnarray}\label{DDphiSpinor}
D^P\,_BD_{AP}\phi& = & -\tfrac{1}{4}\epsilon_{AB}\phi\left(R+\frac{10}{1+\Omega-\phi^2}\left[\frac{1}{4}\phi^2 D_{PQ}\Omega D^{PQ}\Omega\right.\right.\\
\nonumber&& \left.\left.-\frac{}{}(1+\Omega)\phi D^{PQ}\Omega D_{PQ}\phi+(1+\Omega)^2 D_{PQ}\phi D^{PQ}\phi\right]\right),
\end{eqnarray}
\begin{eqnarray}\label{DSSpinor}
&& D^{P}\,_{A}S_{BCDP}\\
\nonumber&& =\frac{1}{1+\Omega-\phi^2}\left\{\frac{}{}\Omega\phi D_{A}\,^{P}\Omega D_{(BC}D_{D)P}\phi-2\Omega(1+\Omega)D_{A}\,^{P}\phi D_{(BC}D_{D)P}\phi\right.\\
\nonumber&& +(1+\Omega)\phi D_{PQ}\Omega D_{(BC}D^{PQ}\phi\epsilon_{D)A}-2(1+\Omega)^2 D_{PQ}\phi D_{(BC}D^{PQ}\phi\epsilon_{D)A}\\
\nonumber&& +\frac{1}{2}\left[1+(-1+\Omega)\phi^2\right]D_{A}\,^{P}\Omega S_{PBCD}-\Omega^2\phi D_{A}\,^{P}\phi S_{PBCD}\\
\nonumber&& +\frac{1}{2}\Omega\phi^2 D^{PQ}\Omega S_{PQ(BC}\epsilon_{D)A}-\Omega(1+\Omega)\phi D^{PQ}\phi S_{PQ(BC}\epsilon_{D)A}\\
\nonumber&& +\frac{1}{6}(1+\Omega)\phi^2 R D_{(BC}\Omega\epsilon_{D)A}-\frac{1}{3}(1+\Omega)^2\phi R D_{(BC}\phi\epsilon_{D)A}\\
\nonumber&& +2\phi \left(D_{A}\,^{P}\phi D_{P(B}\Omega D_{CD)}\Omega-D_{A}\,^{P}\Omega D_{P(B}\Omega D_{CD)}\phi\right)\\
\nonumber&& \left.+4\frac{}{}(1+\Omega)\left(D_{A}\,^{P}\Omega D_{P(B}\phi D_{CD)}\phi-D_{A}\,^{P}\phi D_{P(B}\phi D_{CD)}\Omega\right)\right\}\\
\nonumber&& +\frac{1}{(1+\Omega-\phi^2)^2}\left\{\frac{1}{24}\phi^2\left[-6+(20+3\Omega)\phi^2\right]D^{PQ}\Omega D_{PQ}\Omega D_{(BC}\Omega\epsilon_{D)A}\right.\\
\nonumber&& -\frac{1}{6}(14+23\Omega+3\Omega^2)\phi^3 D^{PQ}\Omega D_{PQ}\phi D_{(BC}\Omega\epsilon_{D)A}\\
\nonumber&& +\frac{1}{6}(1+\Omega)\left[6(1+\Omega)+(8+23\Omega+3\Omega^2)\phi^2\right]D^{PQ}\phi D_{PQ}\phi D_{(BC}\Omega\epsilon_{D)A}\\
\nonumber&& -\frac{1}{12}\phi\left[-12(1+\Omega)+(26+23\Omega+3\Omega^2)\phi^2\right]D^{PQ}\Omega D_{PQ}\Omega D_{(BC}\phi\epsilon_{D)A}\\
\nonumber&& +\frac{1}{3}(1+\Omega)^2\left[-6+(20+3\Omega)\phi^2\right]D^{PQ}\Omega D_{PQ}\phi D_{(BC}\phi\epsilon_{D)A}\\
\nonumber&& \left.-\frac{1}{3}(1+\Omega)^2\left(14+23\Omega+3\Omega^2\right)\phi D^{PQ}\phi D_{PQ}\phi D_{(BC}\phi\epsilon_{D)A}\right\}.
\end{eqnarray}
Equations \eqref{DDOmegaT}, \eqref{DRT} are translated into the space-spinor formalism by making the index replacements $a\rightarrow AB$, $b\rightarrow CD$, $c\rightarrow EF$.\\
We use equations \eqref{DDphiSpinor}, \eqref{DSSpinor}, the spinor version of equations \eqref{DDOmegaT}, \eqref{DRT} and the theory of `exact sets of fields' to prove the next result.
\begin{lem}\label{formalExpansion}
Let there be two given sequences
\begin{equation*}
\hat{{\cal D}}^{\phi}_n=\{\psi_{A_1B_1},\psi_{A_2B_2A_1B_1},\psi_{A_3B_3A_2B_2A_1B_1},...\},
\end{equation*}
\begin{equation*}
\hat{{\cal D}}^S_n=\{\Psi_{A_2B_2A_1B_1},\Psi_{A_3B_3A_2B_2A_1B_1},\Psi_{A_4B_4A_3B_3A_2B_2A_1B_1},...\},
\end{equation*}
of totally symmetric spinors satisfying the reality condition \eqref{realityCond}. Assume that there exists a solution $h$, $\phi$, $\Omega$, $R$, $S_{ABCD}$ to the conformal stationary field equations \eqref{RicH},\eqref{D2phiT},\eqref{DDOmegaT},\eqref{DRT},\eqref{DST} satisfying \eqref{condOmega},\eqref{initialCondR} so that the spinors given by $\hat{{\cal D}}^{\phi}_n,\hat{{\cal D}}^S_n$ coincide with the null data ${\cal D}^{\phi*}_n,{\cal D}^{S*}_n$ given by \eqref{nullDataPhi}, \eqref{nullDataS} of the metric $h$ in terms of an $h$-orthonormal normal frame centered at $i$, i.e.,
\begin{equation}\label{identificationPsi1}
\psi_{A_pB_p...A_1B_1}=D_{(A_pB_p}...D_{A_1B_1)}\phi(i),\,\,\,p\geq 1,
\end{equation}
\begin{equation}\label{identificationPsi2}
\Psi_{A_pB_p...A_1B_1}=D_{(A_pB_p}...D_{A_3B_3}S_{A_2B_2A_1B_1)}(i),\,\,\,p\geq 2.
\end{equation}
Then the coefficients of the normal expansions \eqref{normalExpansion} of the fields \eqref{unknownFieldsNormalCoord},
\begin{eqnarray*}
& D_{A_pB_p}...D_{A_1B_1}\phi,\,\,\,D_{A_pB_p}...D_{A_1B_1}\Omega,\,\,\,D_{A_pB_p}...D_{A_1B_1}R,\\
& D_{A_pB_p}...D_{A_1B_1}S_{ABCD}(i),\,\,\,p\geq 0,
\end{eqnarray*}
are uniquely determined by the data $\hat{{\cal D}}^{\phi}_n,\hat{{\cal D}}^S_n$ and satisfy the reality conditions.
\end{lem}
\begin{proof}
It holds $\phi(i)=0$, $D_{AB}\phi(i)=\psi_{AB}$ and $S_{ABCD}(i)=\Psi_{ABCD}$ by assumption and the expansion coefficients for $\Omega$ and $R$ of lowest order are given by \eqref{condOmega} and \eqref{initialCondR}. Assume the expansion coefficients of $\phi$ and $\Omega$ up to order $p$ and the expansion coefficients of $R$ and $S_{ABCD}$ up to order $p-1$ are known.\\
To discuss the induction step we start with $D_{A_{p+1}B_{p+1}}...D_{A_1B_1}\phi(i)$ and its decomposition in the form \eqref{symmetricCont}. By assumption, the totally symmetric part of it is given by $\psi_{A_{p+1}B_{p+1}...A_1B_1}$. The other terms in the decomposition contain contractions. Let's consider $A_i$ contracted with $A_j$. We can commute the operators $D_{A_iB_i}$ and $D_{A_jB_j}$ with other covariant derivatives, generating by \eqref{spinorCommutator} and \eqref{spinorRicci} only terms of lower order, until we have
\begin{equation*} D_{A_{p+1}B_{p+1}}...D_{A_{i+1}B_{i+1}}D_{A_{i-1}B_{i-1}}...D_{A_{j+1}B_{j+1}}D_{A_{j-1}B_{j-1}}...D_{A_1B_1}D^P\,_{B_{i}}D_{PB_{j}}\phi(i).
\end{equation*}
Equation \eqref{DDphiSpinor} then shows how to express the resulting term by quantities of lower order that are already known.\\
For $D_{A_{p+1}B_{p+1}}...D_{A_1B_1}\Omega(i)$ and $D_{A_pB_p}...D_{A_1B_1}R(i)$ we just use the spinor versions of \eqref{DDOmegaT} and \eqref{DRT} to express them by quantities of lower order.\\
Finally, dealing with $D_{A_pB_p}...D_{A_1B_1}S_{CDEF}(i)$ is quite similar to\\
$D_{A_{p+1}B_{p+1}}...D_{A_1B_1}\phi(i)$. The symmetric term is known by the data. If a contraction is performed between a derivative index and one of $C,D,E,F$ then \eqref{DSSpinor} is used after interchanging derivatives. If the contraction is between two derivatives, the general identities
\begin{equation*}
D_{H(A}D^H\,_{B)}S_{SDEF}=-2S_{H(CDE}S_{F)}\,^H\,_{AB}+\tfrac{1}{3}RS_{H(CDE}h_{F)}\,^H\,_{AB},
\end{equation*}
\begin{equation*}
D_{AB}D^{AB}S_{CDEF}=-2D_F\,^GD_G\,^HS_{CDEH}+3S_{GH(CD}S_{E)F}\,^{GH}+\tfrac{1}{2}RS_{CDEF},
\end{equation*}
implied by \eqref{spinorCommutator}, \eqref{spinorRicci}, together with \eqref{DSSpinor}  show that the corresponding term can be expressed in terms of quantities of lower order. The induction step is completed.\\
That the expansion coefficients satisfy the reality condition is a consequence of the formalism and the fact that they are satisfyed by the data.
\end{proof}
In order to show the convergence of the formal series determined in the previous Lemma we need to impose estimates on the free coefficients given by $\hat{{\cal D}}^{\phi}_n,\hat{{\cal D}}^S_n$. For this we have the following result.
\begin{lem}\label{estimatesNullData}
A necessary condition for the formal series determined in Lemma \ref{formalExpansion} to be absolutely convergent near the origin is that the data given by $\hat{{\cal D}}^{\phi}_n$, $\hat{{\cal D}}^S_n$ satisfy estimates of the type
\begin{equation}\label{estimates-psi}
|\psi_{A_pB_p...A_1B_1}|\leq \frac{p!M}{r^p},\,\,p=1,2,3,...,
\end{equation}
\begin{equation}\label{estimates-Psi}
|\Psi_{A_pB_p...A_1B_1CDEF}|\leq \frac{p!M}{r^p},\,\,p=0,1,2,...,
\end{equation}
with some constants $M,r>0$.
\end{lem}
We skip the proof of this lemma because it uses the same argument as the proof of Lemma 3.2 in \cite{Friedrich07}.\\
Lemma \ref{formalExpansion} shows that the null data determines a formal solution to the stationary field equations. As shown by Beig and Simon \cite{BeigSimon81}, the multipole moments do the same. Thus there is a bijective map $\Theta$ from the null data to the multipoles sequences, $\Theta:\{{\cal D}_n^\phi,{\cal D}_n^S\}\rightarrow\{{\cal D}_{mp}^M,{\cal D}_{mp}^S\}$. Instead of using this argument, we can try to gain more information on the relation starting from \eqref{multipoleField1}, \eqref{multipoleField2}. It is convenient to work in space-spinor form, that means that we are using the $h$-orthonormal frame and normal coordinates previously defined. We get the following result.
\begin{lem}
The spinor fields $P^M_{A_pB_p...A_1B_1}$, $P^S_{A_pB_p...A_1B_1}$, near $i$, given by \eqref{multipoleField1}, \eqref{multipoleField2}, are of the form
\begin{eqnarray}\label{PM}
&&P^M_{A_pB_p...A_1B_1}\\
\nonumber &&=-\frac{1}{2}\left(1+\Omega-\phi^2\right)^{-\frac{1}{2}}\left(1+2\Omega-\phi^2\right)D_{(A_pB_p}...D_{A_3B_3}S_{A_2B_2A_1B_1)}\\
\nonumber &&-\left(1+\Omega-\phi^2\right)^{-\frac{1}{2}}\phi D_{(A_pB_p}...D_{A_1B_1)}\phi\\
\nonumber &&+\left(1+\Omega-\phi^2\right)^{-\frac{3}{2}}\phi\left(p-2\Omega^2\right)D_{(A_pB_p}\Omega D_{A_{p-1}B_{p-1}}...D_{A_1B_1)}\phi\\
\nonumber &&-\left(1+\Omega-\phi^2\right)^{-\frac{1}{2}}(1+\Omega)\left(p-2\Omega^2\right)D_{(A_pB_p}\phi D_{A_{p-1}B_{p-1}}...D_{A_1B_1)}\phi\\
\nonumber &&+F^M_{A_pB_p...A_1B_1},\,\,\,\,\,p\geq 3,
\end{eqnarray}
\begin{equation}\label{PS}
P^S_{A_pB_p...A_1B_1}=D_{(A_pB_p}...D_{A_1B_1)}\phi+F^S_{A_pB_p...A_1B_1},\,\,\,\,\,p\geq 2,
\end{equation}
with symmetric spinor-valued functions $F^M_{A_pB_p...A_1B_1}$ and $F^S_{A_pB_p...A_1B_1}$. The function $F^M_{A_pB_p...A_1B_1}$, $p\geq 3$, is at each point a real linear combination of symmetrized tensor products of
\begin{equation*}
D_{(A_{q-1}B_{q-1}}...D_{A_1B_1)}\phi,\,D_{(A_qB_q}...D_{A_3B_3}S_{A_2B_2A_1B_1)},\,D_{AB}\Omega,\,\,\,2\leq q\leq p-1,
\end{equation*}
with coefficients that depend on $\Omega$ and $\phi$. The function $F^S_{A_pB_p...A_1B_1}$, $p\geq 2$, is a real linear combination of symmetrized tensor products of
\begin{equation*}
D_{(A_{q-2}B_{q-2}}...D_{A_1B_1)}\phi,\,D_{(A_qB_q}...D_{A_3B_3}S_{A_2B_2A_1B_1)},\,\,\,2\leq q\leq p.
\end{equation*}
\end{lem}
\begin{proof}
From \eqref{Omega} we get
\begin{equation*}
\phi_M=\left(1+\Omega-\phi^2\right)^\frac{1}{2},
\end{equation*}
and by direct calculations from \eqref{multipoleField1}, \eqref{multipoleField2} we see that \eqref{PM} is valid for $p=3$ and that \eqref{PS} is valid for $p=2$, with the stated properties for $F^M_{A_3B_3A_2B_2A_1B_1}$ and $F^S_{A_2B_2A_1B_1}$. Assuming that the lemma is true for $p\leq k$, inserting \eqref{PM} and \eqref{PS} into the recursion relation \eqref{multipoleField2}, and using the symmetrized spinor version of \eqref{DDOmegaT}, we see that the lemma is true for $p=k+1$.
\end{proof}
Using \eqref{multipoleField1}, \eqref{multipoleField2}, \eqref{multipoles} and the identification \eqref{identificationPsi1}, \eqref{identificationPsi2} we get for the lower order multipoles
\begin{eqnarray}
\label{lowerOrder1} & \nu^M_{A_1B_1}=0,\,\,\,\,\,\nu^M_{A_2B_2A_1B_1}=-\frac{1}{2}\Psi_{A_2B_2A_1B_1}-\psi_{(A_2B_2}\psi_{A_1B_1)},\\
\label{lowerOrder2} & \nu^M_{A_3B_3A_2B_2A_1B_1}=-\frac{1}{2}\Psi_{A_3B_3A_2B_2A_1B_1}-\psi_{(A_3B_3}\psi_{A_2B_2A_1B_1)},\\
\label{lowerOrder3} & \nu^S_{A_1B_1}=\psi_{A_1B_1},\,\,\,\,\,\nu^S_{A_2B_2A_1B_1}=\psi_{A_2B_2A_1B_1}.
\end{eqnarray}
Also restricting \eqref{PM} and \eqref{PS} to $i$ and with the identification \eqref{identificationPsi1}, \eqref{identificationPsi2} we get
\begin{eqnarray}
\label{nullDataToMultipolesM} \nu^M_{A_pB_p...A_1B_1} & = & -\frac{1}{2}\Psi_{A_pB_p...A_1B_1}-p\psi_{(A_pB_p}\psi_{A_{p-1}B_{p-1}...A_1B_1)}\\
\nonumber && +f^M_{A_pB_p...A_1B_1},\,\,\,p\geq 3,\\
\label{nullDataToMultipolesS} \nu^S_{A_pB_p...A_1B_1} & = & \psi_{A_pB_p...A_1B_1}+f^S_{A_pB_p...A_1B_1},\,\,\,p\geq 2,
\end{eqnarray}
where $f^M_{A_pB_p...A_1B_1}$, $p\geq 3$, is a real linear combination of symmetrized tensor products of
\begin{equation*}
\psi_{A_{q-1}B_{q-1}...A_1B_1},\,\,\,\,\,\Psi_{A_qB_q...A_1B_1},\,\,\,\,\,2\leq q\leq p-1,
\end{equation*}
and $f^S_{A_pB_p...A_1B_1}$, $p\geq 2$, is a real linear combination of symmetrized tensor products of
\begin{equation*}
\psi_{A_{q-2}B_{q-2}...A_1B_1},\,\,\,\,\,\Psi_{A_qB_q...A_1B_1},\,\,\,\,\,2\leq q\leq p.
\end{equation*}
Equations \eqref{lowerOrder1}, \eqref{lowerOrder2}, \eqref{lowerOrder3}, \eqref{nullDataToMultipolesM} and \eqref{nullDataToMultipolesS} give a nonlinear map $\Theta$, that can be read as a map
\begin{equation*}
\Theta:\{\hat{\cal D}_n^\phi,\hat{\cal D}_n^S\}\rightarrow\{\hat{\cal D}_{mp}^M,\hat{\cal D}_{mp}^S\}
\end{equation*}
of the set of \emph{abstract null data} into the set of \emph{abstract multipoles} (i.e., sequences of symmetric spinors not necessarily derived from a metric). Now is fairly easy to show that the map can be inverted.
\begin{cor}
The map $\Theta$ that maps sequences of abstract null data $\{\hat{\cal D}_n^\phi,\hat{\cal D}_n^S\}$ onto sequences of abstract multipoles $\{\hat{\cal D}_{mp}^M,\hat{\cal D}_{mp}^S\}$ is bijective.
\end{cor}
\begin{proof}
From \eqref{lowerOrder2}, \eqref{lowerOrder3} we see that $f^M_{A_3B_3A_2B_2A_1B_1}=0$, $f^S_{A_2B_2A_1B_1}=0$, with this and the stated properties for $f^M_{A_pB_p...A_1B_1}$ and $f^S_{A_pB_p...A_1B_1}$ an inverse for $\Theta$ can be constructed inverting the relations \eqref{nullDataToMultipolesM} and \eqref{nullDataToMultipolesS} recursively.
\end{proof}
Hence, for a given metric $h$, the sequences of multipoles and the sequences of null data in a given standard frame carry the same information on $h$. As said, we prefer to work with the null data because they are linear in $\phi$ and $S_{ABCD}$.

\section{The characteristic initial value problem}\label{characteristicProblem}
After showing that the null data determine the solution, one would have to show that the estimates \eqref{estimates-psi}, \eqref{estimates-Psi} imply Cauchy estimates for the expansion coefficients
\begin{equation*}
|D_{A_pB_p}...D_{A_1B_1}T(i)|\leq \frac{p!M}{r^p},\,\,\,A_p,B_p,...,A_1,B_1=0,1,\,\,\,p=0,1,2,...,
\end{equation*}
where $T$ is any of $\phi$, $\Omega$, $R$, $S_{ABCD}$. This would ensure the convergence of the normal expansion at $i$. The induction procedure used so far for calculating the expansion coefficients from the null data generates additional non-linear terms each time one interchanges a derivative or uses the conformal field equations. Thus, it does not seem suited for deriving estimates. Instead, we use the intrinsic geometric nature of the problem and the data to formulate the problem as a boundary value problem to which Cauchy-Kowalevskaya type arguments apply.\\
As the fields $h$, $\phi$, $\Omega$, $R$, $S_{ABCD}$ are real analytic in the normal coordinates $x^a$ and a standard frame $c_{AB}$ centered at $i$, they can be extended near $i$ by analyticity into the complex domain and considered as holomorphic fields on a complex analytic manifold $N_c$. Choosing $N_c$ to be a sufficiently small neighbourhood of $i$, we can assume the extended coordinates, again denoted by $x^a$, to define a holomorphic coordinate system on $N_c$ which identifies $N_c$ with an open neighbourhood of the origin in $\mathbb{C}$. The original manifold $N$ is then a real, 3-dimensional, real analytic submanifold of the real, 6-dimensional, real analytic manifold underlying $N_c$. Under the analytic extension the main differential geometric concepts and formulas remain valid. The coordinates $x^a$ and the extended frame, again denoted by $c_{AB}$, satisfy the same defining equations and the extended fields, denoted again by $h$, $\phi$, $\Omega$, $R$, $S_{ABCD}$, satisfy the conformal stationary vacuum field equations as before.\\
The analytic function $\Gamma=\delta_{ab}x^ax^b$ on $N$ extends to a holomorphic function on $N_c$. On $N$ it vanishes only at $i$, but the set
\begin{equation*}
{\cal N}_i=\{p\in N_c|\Gamma(p)=0\},
\end{equation*}
is an irreducible analytical set such that ${\cal N}_i\backslash\{i\}$ is a 2-dimensional complex submanifold of $N_c$. It is the cone swept out by the complex null geodesics through $i$ and we will refer to it as the \emph{null cone at} $i$.\\
Now let $u\rightarrow x^a(u)$ be a null geodesic through $i$ such that $x^a(0)=0$. Its tangent vector is then of the form $\dot{x}^{AB}=\iota^A\iota^B$ with a spinor field $\iota^A=\iota^A(u)$ satisfying $D_{\dot{x}}\iota^A=0$ along the geodesic. Then
\begin{equation}\label{phiNull}
\phi(u)=\phi(x(u)),
\end{equation}
\begin{equation}\label{SNull}
S_0(u)=\dot{x}^a\dot{x}^bS_{ab}(x(u))=\iota^A\iota^B\iota^C\iota^DS_{ABCD}(x(u)),
\end{equation}
are analytic functions of $u$ with Taylor expansion
\begin{equation*}
\phi(u)=\sum_{p=0}^\infty\frac{1}{p!}u^p\frac{d^p\phi}{du^p}(0),\,\,\,S_0(u)=\sum_{p=0}^\infty\frac{1}{p!}u^p\frac{d^pS_0}{du^p}(0),
\end{equation*}
where
\begin{equation*}
\frac{d^p\phi}{du^p}(0)=\dot{x}^{a_p}...\dot{x}^{a_1}D_{a_p}...D_{a_1}\phi(0)=\iota^{A_p}\iota^{B_p}...\iota^{A_1}\iota^{B_1}D_{(A_pB_p}...D_{A_1B_1)}\phi(i),
\end{equation*}
\begin{equation*}
\frac{d^pS_0}{du^p}(0)=\iota^{A_p}\iota^{B_p}...\iota^{A_1}\iota^{B_1}\iota^C\iota^D\iota^E\iota^FD_{(A_pB_p}...D_{A_1B_1}S_{CDEF)}(i).
\end{equation*}
This shows that knowing these expansion coefficients for initial null vectors $\iota^A\iota^B$ covering an open subset of the null directions at $i$ is equivalent to knowing the null data $\hat{{\cal D}}^{\phi}_n,\hat{{\cal D}}^S_n$ of the metric $h$.\\
Our problem can thus be formulated as the boundary value problem for the conformal stationary vacuum equations with data given by the functions \eqref{phiNull}, \eqref{SNull} on ${\cal N}_i$, where the $\iota^A\iota^B$ are parallely propagated null vectors tangent to ${\cal N}_i$.\\
${\cal N}_i$ is not a smooth hypersurface but an analytic set with a vertex at the point $i$, and we need a setting in which the mechanism of calculating the expansion coefficients allows us to derive estimates on the coefficients from the conditions imposed on the data. That is done in the next subsections.

\subsection{The geometric gauge}
We need to choose a gauge suitably adapted to the singular set ${\cal N}_i$. The coordinates and the frame field will then necessarily be singular and the frame will no longer define a smooth lift to the bundle of frames but a subset which becomes tangent to the fibres over some points.\\
We will use the principal bundle of normalized spin frames $SU(N)\xrightarrow{\pi}N$ with structure group $SU(2)$, which is the group of complex $2\times2$ matrices $(s^A\,_B)_{A,B=0,1}$ satisfying
\begin{equation}\label{conditionsSL}
\epsilon_{AB}s^A\,_Cs^B\,_D=\epsilon_{CD},\,\,\,\tau_{AB'}s^A\,_C\bar{s}^{B'}\,_{D'}=\tau_{CD'}.
\end{equation}
The $2:1$ covering homomorphism of $SU(2)$ onto $SO(3,\mathbb{R})$ is performed via
\begin{equation*}
SU(2)\ni s^A\,_B\rightarrow s^a\,_b=\alpha^a\,_{AB}s^A\,_Cs^B\,_D\alpha^{CD}\,_b\in SO(3,\mathbb{R}).
\end{equation*}
Under holomorphic extension the map above extends to a $2:1$ covering homomorphism of the group $SL(2,\mathbb{C})$ onto the group $SO(3,\mathbb{C})$, where $SL(2,\mathbb{C})$ denotes the group of complex $2\times 2$ matrices satisfying only the first of conditions \eqref{conditionsSL}.\\
A point $\delta\in SU(N)$ is given by a pair of spinors $\delta=(\delta^A_0,\delta^A_1)$ at a given point of $N$ which satisfies
\begin{equation}\label{conditionSpinFrame}
\epsilon(\delta_A,\delta_B)=\epsilon_{AB},\,\,\,\epsilon(\delta_A,\delta^+_{B'})=\tau_{AB'},
\end{equation}
and the action of the structure group is given for $s\in SU(2)$ by
\begin{equation*}
\delta\rightarrow\delta\cdot s\mbox{   where   }(\delta\cdot s)_A=s^B\,_A\delta_B.
\end{equation*}
The projection $\pi$ maps a frame $\delta$ into its base point in $N$. The bundle of spin frames is mapped by a $2:1$ bundle morphism $SU(N)\xrightarrow{p}SO(N)$ onto the bundle $SO(N)\xrightarrow{\pi'}N$ of oriented, orthonormal frames on $N$ so that $\pi'\circ p=\pi$. For any spin frame $\delta$ we can identify by \eqref{conditionSpinFrame} the matrix $(\delta^A_B)_{A,B=0,1}$ with an element of the group $SU(2)$. With this reading the map $p$ will be assumed to be realized by
\begin{equation*}
SU(N)\ni\delta\rightarrow p(\delta)_{AB}=\delta^E_A\delta^F_Bc_{EF}\in SO(N),
\end{equation*}
where $c_{AB}$ denotes the normal frame field on $N$ introduced before. We refer to $p(\delta)$ as the frame associated with the spin frame $\delta$.\\
Under holomorphic extension the bundle $SU(N)\xrightarrow{\pi} N$ is extended to the principal bundle $SL(N_c)\xrightarrow{\pi}N_c$ of spin frames $\delta=(\delta^A_0,\delta^A_1)$ at given points of $N_c$ which satisfy only the first of conditions \eqref{conditionSpinFrame}. Its structure group is $SL(2,\mathbb{C})$. The bundle $SU(N)\xrightarrow{\pi} N$ is embedded into $SL(N_c)\xrightarrow{\pi}N_c$ as a real analytic subbundle. The bundle morphism $p$ extends to a $2:1$ bundle morphism, again denoted by $p$, of $SL(N_c)\xrightarrow{\pi}N_c$ onto the bundle $SO(N_c)\xrightarrow{\pi'}N_c$ of oriented, normalized frames of $N_c$ with structure group $SO(3,\mathbb{C})$. We shall make use of several structures on $SM(N_c)$.\\
With each $\alpha\in sl(2,\mathbb{C})$, i.e., $\alpha=(\alpha^A\,_B)$ with $\alpha_{AB}=\alpha_{BA}$, is associated a \emph{vertical vector field} $Z_\alpha$ tangent to the fibres, which is given at $\delta\in SL(N_c)$ by $Z_\alpha(\delta)=\tfrac{d}{dv}(\delta\cdot \exp(v\alpha))|_{v=0}$, where $v\in\mathbb{C}$ and $\exp$ denotes the exponential map $sl(2,\mathbb{C})\rightarrow SL(2,\mathbb{C})$.\\
The $\mathbb{C}^3$-valued \emph{soldering form} $\sigma^{AB}=\sigma^{(AB)}$ maps a tangent vector $X\in T_\delta SL(N_c)$ onto the components of its projection $T_\delta(\pi)X\in T_{\pi(\delta)}N_c$ in the frame $p(\delta)$ associated with $\delta$ so that $T_\delta(\pi)X=\langle\sigma^{AB},X\rangle p(\delta)_{AB}$. It follows that $\langle\sigma^{AB},Z_\alpha\rangle =0$ for any vertical vector field $Z_\alpha$.\\
The $sl(2,\mathbb{C})$-valued \emph{connection form} $\omega^A\,_B$ on $SL(N_c)$ transforms with the adjoint transformation under the action of $SL(2,\mathbb{C})$ and maps any vertical vector field $Z_\alpha$ onto its generator so that $\langle\omega^A\,_B,Z_\alpha\rangle=\alpha^A\,_B$.\\
With $x^{AB}=x^{(AB)}\in\mathbb{C}^3$ is associated the \emph{horizontal vector field} $H_x$ on $SL(N_c)$ which is horizontal in the sense that $\langle\omega^A\,_B,H_x\rangle=0$ and which satisfies $\langle\sigma^{AB},H_x\rangle =x^{AB}$. Denoting by $H_{AB},A,B=0,1$, the horizontal vector fields satisfying $\langle\sigma^{AB},H_{CD}\rangle =h^{AB}\,_{CD}$, it follows that $H_x=x^{AB}H_{AB}$. An integral curve of a horizontal vector field projects onto an $h$-geodesic and represents a spin frame field which is parallelly transported along this geodesic.\\
A holomorphic spinor field $\psi$ on $N_c$ si represented on $SL(N_c)$ by a holomorphic spinor-valued function $\psi_{A_1...A_j}(\delta)$ on $SL(N_c)$, given by the components of $\psi$ in the frame $\delta$. We shall use the notation $\psi_k=\psi_{(A_1...A_j)_k},k=0,..,j$, where $(......)_k$ denotes the operation \emph{`symmetrize and set $k$ indices equal to $1$ the rest equal to $0$'}. These functions completely specify $\psi$ if $\psi$ is symmetric. They are then referred to as the \emph{essential components of $\psi$}.

\subsection{The submanifold $\hat{N}$ of $SL(N_c)$}
Using the available geometrical structure we construct a three-dimensional submanifold $\hat{N}$ of $SL(N_c)$ in such a way that it induces coordinates in $N_c$. By the construction procedure the induced coordinates are suitable adapted to the set ${\cal N}_i$.\\
We start by choosing a spin frame $\delta^*$ such that $\pi(\delta^*)=i$ and $p(\delta^*)_{AB}=c_{AB}$. The curve
\begin{equation*}
\mathbb{C}\ni v\rightarrow \delta(v)=\delta^*\cdot s(v)\in SL(N_c),
\end{equation*}
\begin{equation}\label{verticals}
s(v)=exp(v\alpha)=\left(\begin{array}{cc} 1 & 0 \\ v & 1 \end{array}\right),\,\,\,\alpha=\left(\begin{array}{cc} 0 & 0 \\ 1 & 0 \end{array}\right)\in sl(2,\mathbb{C}),
\end{equation}
defines a vertical, 1-dimensional, holomorphic submanifold $I$ of $SL(N_c)$ on which $v$ defines a coordinate. The associated family of frames $e_{AB}(v)$ at $i$ is given by $e_{AB}(v)=s^C\,_A(v)s^D\,_B(v)c_{CD}$, and explicitly by
\begin{equation*}
e_{00}=c_{00}+2vc_{01}+v^2c_{11},\,\,\,e_{01}(v)=c_{01}+vc_{11},\,\,\,e_{11}(v)=c_{11}.
\end{equation*}
We perform the following construction in a neighbourhood of $I$. If it is chosen small enough all the following statements will be correct.\\
The set $I$ is moved with the flow of $H_{11}$ to obtain a holomorphic 2-manifold $U_0$ of $SL(N_c)$. We denote by $w$ the parameter on the integral curves of $H_{11}$ that vanishes on $I$, and we extend $v$ to $U_0$ by assuming it to be constant on the integral curves of $H_{11}$. All these integral curves are mapped by $\pi$ onto the null geodesic $\gamma(w)$ with affine parameter $w$ and tangent vector $\gamma'(0)=c_{11}$ at $\gamma(0)=i$. The parameter $v$ specifies which frame fields are parallelly propagated along $\gamma$.\\
$U_0$ is now moved with the flow of $H_{00}$ to obtain a holomorphic 3-submanifold $\hat{N}$ of $SL(N_c)$. We denote by $u$ the parameter on the integral curves of $H_{11}$ that vanishes on $U_0$ and we extend $v$ and $w$ to $\hat{N}$ by assuming them to be constant along the integral curves of $H_{00}$. The functions $z^1=u,z^2=v,z^3=w$ define holomorphic coordinates on $\hat{N}$. We denote again $\pi$ the restiction of the projection to $\hat{N}$.\\
The projections of the integral curves of $H_{00}$ with a fixed value of $w$ sweep out, together with $\gamma$, the null cone ${\cal N}_{\gamma(w)}$ near $\gamma(w)$, which is generated by the null geodesics through the point $\gamma(w)$. On the null geodesics $u$ is an affine parameter which vanishes at $\gamma(w)$ while $v$ parametrizes the different generators. The set $W_0=\{w=0\}$ projects onto ${\cal N}_i\backslash\gamma$ and will define the initial data set for our problem. The map $\pi$ induces a biholomorphic diffeomorphism of $\hat{N}'\equiv\hat{N}\backslash U_0$ onto $\pi(\hat{N}')$. The singularity of the gauge at points of $U_0$ consists in $\pi$ dropping rank on $U_0$, where $\partial_v=Z_{\alpha}$. The null curve $\gamma(w)$ will be referred to as the \emph{the singular generator of ${\cal N}_i$ in the gauge determined by the spin frame $\delta^*$ resp. the corresponding frame $c_{AB}$ at $i$}.\\
The soldering an the connection form pull back to holomorphic 1-forms on $\hat{N}$, which will be denoted again by $\sigma^{AB}$ and $\omega^A\,_B$. If the pull back of the curvature form $\Omega^A\,_B=\tfrac{1}{2}r^A\,_{BCDEF}\sigma^{CD}\wedge\sigma^{EF}$ to $\hat{N}$ is denoted again by $\Omega^A\,_B$, then the soldering and the connection form satisfy the structural equations
\begin{equation*}
d\sigma^{AB}=-\omega^A\,_C\wedge\sigma^{CB}-\omega^B\,_C\wedge\sigma^{AC},\,\,\,d\omega^A\,_B=-\omega^A\,_C\wedge\omega^C\,_B+\Omega^A\,_B.
\end{equation*}
Using the way in which $\hat{N}$ is constructed, and in terms of the coordinates $z^a$, we get $\sigma^{AB}=\sigma^{AB}\,_adz^a$ on $\hat{N}'$, where
\begin{equation*}
(\sigma^{AB}\,_a)=\left(\begin{array}{ccc} 1 & \sigma^{00}\,_2 & \sigma^{00}\,_3 \\ 0 & \sigma^{01}\,_2 & \sigma^{01}\,_3 \\ 0 & 0 & 1 \end{array}\right)=\left(\begin{array}{ccc} 1 & {\cal O}(u^3) & {\cal O}(u^2) \\ 0 & u+{\cal O}(u^3) & {\cal O}(u^2) \\ 0 & 0 & 1 \end{array}\right)\mbox{ as }u\rightarrow 0.
\end{equation*}
On $\hat{N}'$ there exist unique, holomorphic vector fields $e_{AB}$ which satisfy
\begin{equation*}
\langle\sigma^{AB},e_{CD}\rangle=h^{AB}\,_{CD}.
\end{equation*}
If one writes $e_{AB}=e^a\,_{AB}\partial_{z^a}$, then
\begin{equation*}
(e^a\,_{AB})=\left(\begin{array}{ccc} 1 & e^1\,_{01} & e^1\,_{11} \\ 0 & e^2\,_{01} & e^2\,_{11} \\ 0 & 0 & 1 \end{array}\right)=\left(\begin{array}{ccc} 1 & {\cal O}(u^2) & {\cal O}(u^2) \\ 0 & \tfrac{1}{2u}+{\cal O}(u) & {\cal O}(u) \\ 0 & 0 & 1 \end{array}\right)\mbox{ as }u\rightarrow 0.
\end{equation*}
We shall write
\begin{equation*}
e^a\,_{AB}=e^{*a}\,_{AB}+\hat{e}^a\,_{AB},
\end{equation*}
with singular part
\begin{equation*}
e^{*a}\,_{AB}=\delta^a_1\epsilon_A\,^0\epsilon_B\,^0+\delta^a_2\frac{1}{u}\epsilon_{(A}\,^0\epsilon_{B)}\,^1+\delta^a_3\epsilon_A\,^1\epsilon_B\,^1,
\end{equation*}
and holomorphic functions $\hat{e}^a\,_{AB}$ on $\hat{N}$ which satisfy
\begin{equation}\label{FrameGauge}
\hat{e}^a_{AB}={\cal O}(u)\mbox{ as }u\rightarrow 0.
\end{equation}
We define the connection coefficients on $\hat{N}'$ by $\omega^A\,_B=\Gamma_{CD}\,^A\,_B\sigma^{CD}$ with $\Gamma_{CDAB}\equiv\langle\omega_{AB},e_{CD}\rangle$, so that $\Gamma_{ABCD}=\Gamma_{(AB)(CD)}$, and from the definition of the frame
\begin{equation*}
\Gamma_{00AB}=0\mbox{ on }\hat{N},\,\,\,\Gamma_{11AB}=0\mbox{ on }U_0,
\end{equation*}
and it follows that
\begin{equation*}
\Gamma_{ABCD}=\Gamma^*_{ABCD}+\hat{\Gamma}_{ABCD},
\end{equation*}
with singular part
\begin{equation*}
\Gamma^*_{ABCD}=-\frac{1}{u}\epsilon_{(A}\,^0\epsilon_{B)}\,^1\epsilon_C\,^0\epsilon_D\,^0,
\end{equation*}
and holomorphic functions $\hat{\Gamma}_{ABCD}$ on $\hat{N}$ which satisfy
\begin{equation}\label{GammaGauge}
\hat{\Gamma}_{ABCD}={\cal O}(u)\mbox{ as }u\rightarrow 0.
\end{equation}

\subsection{Tensoriality and expansion type}\label{ExpansionType}
As the induced map $\pi$ of $\hat{N}$ into $N_c$ is singular on $U_0$, not every holomorphic function of the $z^a$ can arise as a pull-back to $\hat{N}$ of a holomorphic function on $N_c$. The latter must have a special type of expansion in terms of the $z^a$ which reflects the particular relation between the `angular' corrdinate $v$ and the `radial' coordinate $u$. We take from \cite{Friedrich07} the following definition and lemma.
\begin{defin}
A holomorphic function $f$ on $\hat{N}$ is said to be of $v$-finite expansion type $k_f$, with $k_f$ an integer, if it has in terms of the coordinates $u$, $v$, and $w$ a Taylor expansion at the origin of the form
\begin{equation*}
f=\sum_{p=0}^\infty\sum_{m=0}^\infty\sum_{n=0}^{2m+k_f}f_{m,n,p}u^mv^nw^p,
\end{equation*}
where it is assumed that $f_{m,n,p}=0$ if $2m+k_f<0$.
\end{defin}
\begin{lem}\label{expansionType}
Let $\phi_{A_1...A_j}$ be a holomorphic, symmetric, spinor-valued function on $SL(N_c)$. Then the restrictions of its essential components $\phi_k=\phi_{(A_1...A_j)_k}$, $0\leq k\leq j$, to $\hat{N}$ satisfy
\begin{equation*}
\partial_v\phi_k=(j-k)\phi_{k+1},\,\,\,k=0,...,j,\mbox{ on }U_0,
\end{equation*}
(where we set $\phi_{j+1}=0$) and $\phi_k$ is of expansion type $j-k$.
\end{lem}

\subsection{The null data on $W_0$}
As we have seen, prescribing the null data is equivalent to knowing $\phi$ and $S_0$ in the null cone. Now we need to know how this fit into our particular gauge. For this we derive an expansion of the restriction of $\phi$ and $S_0$ to the hypersurface $W_0$.\\
Consider the normal frame $c_{AB}$ on $N_c$ near $i$ which agrees at $i$ with the frame associated with $\delta^*$ and denote the null data of $h$ in this frame by
\begin{equation*}
{\cal D}^{\phi*}_n=\{D^*_{(A_pB_p}...D^*_{A_1B_1)}\phi(i),\,\,\,p=1,2,3,...\},
\end{equation*}
\begin{equation*}
{\cal D}^{S*}_n=\{D^*_{(A_pB_p}...D^*_{A_1B_1}S^*_{ABCD)}(i),\,\,\,p=0,1,2,3,...\}.
\end{equation*}
Choose now a fixed value of $v$ and consider $s(v)$ as in \eqref{verticals}, then the vector $H_{00}(\delta^*\cdot s)$ projects onto the null vector $e_{00}=s^A\,_0s^B\,_0c_{AB}$ at $i$ and is tangent to a null geodesic $\eta=\eta(u,v)$ on ${\cal N}_i$ with affine parameter $u$, $u=0$ at $i$. The integral curve of $H_{00}$ through $\delta^*\cdot s$ projects onto this null geodesic. Using the explicit expression for $s=s(v)$ follows that
\begin{eqnarray}\label{expansionPhi}
\nonumber\phi(u,v) & = & \phi|_{\eta(u,v)}=\sum_{m=0}^\infty\frac{1}{m!}u^mD^m_{00}\phi|_{\eta(0,v)}\\
\nonumber& = & \sum_{m=0}^\infty\frac{1}{m!}u^ms^{A_m}\,_0s^{B_m}\,_0...s^{A_1}\,_0s^{B_1}\,_0D^*_{(A_mB_m}...D^*_{A_1B_1)}\phi(i)\\
& = & \sum_{m=0}^\infty\sum_{n=0}^{2m}\psi_{m,n}u^mv^n,
\end{eqnarray}
with
\begin{equation*}
\psi_{m,n}=\frac{1}{m!}\binom{2m}{n}D^*_{(A_mB_m}...D^*_{A_1B_1)_n}\phi(i),\,\,\,0\leq n\leq 2m.
\end{equation*}
In the same way
\begin{eqnarray}\label{expansionS}
S_0(u,v) & = & S_{0000}|_{\eta(u,v)}\\
\nonumber& = & s^A\,_0s^B\,_0s^C\,_0s^D\,_0S^*_{ABCD}|_{\eta(u,v)}=\sum_{m=0}^\infty\sum_{n=0}^{2m+4}\Psi_{m,n}u^mv^n,
\end{eqnarray}
with
\begin{equation*}
\Psi_{m,n}=\frac{1}{m!}\binom{2m+4}{n}D^*_{(A_mB_m}...D^*_{A_1B_1}S^*_{ABCD)_n}(i),\,\,\,0\leq n\leq 2m.
\end{equation*}
This shows how to determine $\phi(u,v),S_0(u,v)$ from the null data ${\cal D}^{\phi*}_n,{\cal D}^{S*}_n$ and vice versa.

\section{The conformal stationary vacuum field equations on $\hat{N}$}\label{sectionConformalEquations}
Now we can use the frame calculus in its standard form. Given the fields $\Omega$, $\phi$, $R$ and $S_{ABCD}$, and using the frame $e_{AB}$ and the connection coefficients $\Gamma_{ABCD}$ on $\hat{N}$, we set
\begin{eqnarray*}
r_{ABCDEF} & \equiv & e_{CD}(\Gamma_{EFAB})-e_{EF}(\Gamma_{CDAB})+\Gamma_{EF}\,^{K}\,_{C}\Gamma_{DKAB}\\
&& +\Gamma_{EF}\,^{K}\,_{D}\Gamma_{CKAB}-\Gamma_{CD}\,^{K}\,_{E}\Gamma_{KFAB}-\Gamma_{CD}\,^{K}\,_{F}\Gamma_{EKAB}\\
&& +\Gamma_{EF}\,^{K}\,_{B}\Gamma_{CDAK}-\Gamma_{CD}\,^{K}\,_{B}\Gamma_{EFAK}-t_{CD}\,^{GH}\,_{EF}\Gamma_{GHAB},
\end{eqnarray*}
and we define there the quantities $t_{AB}\,^{EF}\,_{CD}$, $R_{ABCDEF}$, $A_{AB}$, $\Sigma_{AB}$, $\Phi_{AB}$, $\Pi_{AB}$, $\Sigma_{ABCD}$ and $H_{ABCD}$ by
\begin{eqnarray*}
t_{AB}\,^{EF}\,_{CD}e^{a}\,_{EF} & \equiv & 2\Gamma_{AB}\,^{E}\,_{(C}e^{a}\,_{D)E}-2\Gamma_{CD}\,^{E}\,_{(A}e^{a}\,_{B)E}\\
&& -e^{a}\,_{CD,b}e^{b}\,_{AB}+e^{a}\,_{AB,b}e^{b}\,_{CD},
\end{eqnarray*}
\begin{eqnarray*}
R_{ABCDEF} & \equiv & r_{ABCDEF}-\frac{1}{2}\left[\left(S_{ABCE}-\frac{1}{6}R h_{ABCE}\right)\epsilon_{DF}\right.\\
&& \left.+\left(S_{ABDF}-\frac{1}{6}R h_{ABDF}\right)\epsilon_{CE}\right],
\end{eqnarray*}
\begin{equation*}
A_{AB}\equiv D_{AB}\phi-\phi_{AB},
\end{equation*}
\begin{equation*}
\Sigma_{AB}\equiv D_{AB}\Omega-\Omega_{AB},
\end{equation*}
\begin{eqnarray*}
\Phi_{AB} & \equiv &  D^{P}\,_{B}\phi_{AP}+\tfrac{1}{4}\epsilon_{AB}\phi\bigg(R+\frac{10}{1+\Omega-\phi^2}\bigg[\frac{1}{4}\phi^2\Omega_{PQ}\Omega^{PQ}\\
&&-(1+\Omega)\phi\Omega^{PQ}\phi_{PQ}+(1+\Omega)^2\phi_{PQ}\phi^{PQ}\bigg]\bigg),
\end{eqnarray*}
\begin{eqnarray*}
\Pi_{AB} & \equiv & D_{AB}R-\frac{1}{1+\Omega-\phi^2}\bigg\{\frac{}{}2(4+7\Omega)\phi\Omega^{PQ}D_{AB}\phi_{PQ}\\
&& -4(1+\Omega)(4+7\Omega)\phi^{PQ}D_{AB}\phi_{PQ}\\
&& +[3+(-3+7\Omega)\phi^2]\Omega^{PQ}S_{PQAB}-2\Omega(4+7\Omega)\phi\phi^{PQ}S_{PQAB}\\
&& +\frac{1}{3}(4+7\Omega)\phi^2 R\Omega_{AB}-\frac{2}{3}(1+\Omega)(4+7\Omega)\phi R\phi_{AB}\bigg\}\\
&& -\frac{1}{3(1+\Omega-\phi^2)^2}\bigg\{\frac{1}{2}\phi^2\left[-12+(40+21\Omega)\phi^2\right]\Omega^{PQ}\Omega_{PQ}\Omega_{AB}\\
&& -2\phi\left[-18(1+\Omega)+(46+61\Omega+21\Omega^2)\phi^2\right]\Omega^{PQ}\phi_{PQ}\Omega_{AB}\\
&& +2(1+\Omega)\left[-24(1+\Omega)+(52+61\Omega+21\Omega^2)\phi^2\right]\phi^{PQ}\phi_{PQ}\Omega_{AB}\\
&& -\phi\left[12(1+\Omega)+(16+61\Omega+21\Omega^2)\phi^2\right]\Omega^{PQ}\Omega_{PQ}\phi_{AB}\\
&& +4(1+\Omega)\left[6(1+\Omega)+(22+61\Omega+21\Omega^2)\phi^2\right]\Omega^{PQ}\phi_{PQ}\phi_{AB}\\
&& -4(1+\Omega)^2(7+3\Omega)(4+7\Omega)\phi\phi^{PQ}\phi_{PQ}\phi_{AB}\bigg\},
\end{eqnarray*}
\begin{eqnarray*}
\Sigma_{ABCD} & \equiv & D_{AB}\Omega_{CD}+\Omega S_{ABCD}+\frac{1}{3}(1+\Omega)R h_{ABCD}\\
&& -\frac{1}{1+\Omega-\phi^2}\bigg\{\frac{1}{2}\left[1+(-1+\Omega)\phi^2\right]\Omega_{AB}\Omega_{CD}\\
&& -\Omega^2\phi(\Omega_{AB}\phi_{CD}+\Omega_{CD}\phi_{AB})+2\Omega^2(1+\Omega)\phi_{AB}\phi_{CD}\\
&& -\frac{4}{3}(2+3\Omega)\bigg[\frac{1}{4}\phi^2\Omega_{PQ}\Omega^{PQ}-(1+\Omega)\phi\Omega_{PQ}\phi^{PQ}\\
&& +(1+\Omega)^2\phi_{PQ}\phi^{PQ}\bigg]h_{ABCD}\bigg\},
\end{eqnarray*}
\begin{eqnarray*}
H_{ABCD} & \equiv & D^{P}\,_{A}S_{BCDP}\\
&& -\frac{1}{1+\Omega-\phi^2}\bigg\{\frac{}{}\Omega\phi\Omega_{A}\,^{P}D_{(BC}\phi_{D)P}-2\Omega(1+\Omega)\phi_{A}\,^{P}D_{(BC}\phi_{D)P}\\
&& +(1+\Omega)\phi\Omega_{PQ}D_{(BC}\phi^{PQ}\epsilon_{D)A}-2(1+\Omega)^2\phi_{PQ}D_{(BC}\phi^{PQ}\epsilon_{D)A}\\
&& +\frac{1}{2}\left[1+(-1+\Omega)\phi^2\right]\Omega_{A}\,^{P}S_{PBCD}-\Omega^2\phi\phi_{A}\,^{P}S_{PBCD}\\
&& +\frac{1}{2}\Omega\phi^2\Omega^{PQ}S_{PQ(BC}\epsilon_{D)A}-\Omega(1+\Omega)\phi\phi^{PQ}S_{PQ(BC}\epsilon_{D)A}\\
&& +\frac{1}{6}(1+\Omega)\phi^2 R\Omega_{(BC}\epsilon_{D)A}-\frac{1}{3}(1+\Omega)^2\phi R\phi_{(BC}\epsilon_{D)A}\\
&& +2\phi \left(\phi_{A}\,^{P}\Omega_{P(B}\Omega_{CD)}-\Omega_{A}\,^{P}\Omega_{P(B}\phi_{CD)}\right)\\
&& +4(1+\Omega)\left(\Omega_{A}\,^{P}\phi_{P(B}\phi_{CD)}-\phi_{A}\,^{P}\phi_{P(B}\Omega_{CD)}\right)\bigg\}\\
&& -\frac{1}{3(1+\Omega-\phi^2)^2}\bigg\{\frac{1}{8}\phi^2\left[-6+(20+3\Omega)\phi^2\right]\Omega^{PQ}\Omega_{PQ}\Omega_{(BC}\epsilon_{D)A}\\
&& -\frac{1}{2}(14+23\Omega+3\Omega^2)\phi^3\Omega^{PQ}\phi_{PQ}\Omega_{(BC}\epsilon_{D)A}\\
&& +\frac{1}{2}(1+\Omega)\left[6(1+\Omega)+(8+23\Omega+3\Omega^2)\phi^2\right]\phi^{PQ}\phi_{PQ}\Omega_{(BC}\epsilon_{D)A}\\
&& -\frac{1}{4}\phi\left[-12(1+\Omega)+(26+23\Omega+3\Omega^2)\phi^2\right]\Omega^{PQ}\Omega_{PQ}\phi_{(BC}\epsilon_{D)A}\\
&& +(1+\Omega)^2\left[-6+(20+3\Omega)\phi^2\right]\Omega^{PQ}\phi_{PQ}\phi_{(BC}\epsilon_{D)A}\\
&& -(1+\Omega)^2\left(7+\Omega\right)\left(2+3\Omega\right)\phi\phi^{PQ}\phi_{PQ}\phi_{(BC}\epsilon_{D)A}\bigg\}.
\end{eqnarray*}
The tensor fields on the left hand side have been introduced as labels for the equations and for discussing in an ordered manner the interdependencies of the equations. In terms of these tensor fields, the conformal stationary vacuum equations read
\begin{eqnarray*}
t_{AB}\,^{EF}\,_{CD}e^{a}\,_{EF}=0,\hspace{1cm}R_{ABCDEF}=0,\hspace{1cm}A_{AB}=0,\hspace{1cm}\Sigma_{AB}=0,\\
\Phi_{AB}=0,\hspace{1cm}\Pi_{AB}=0,\hspace{1cm}\Sigma_{ABCD}=0,\hspace{1cm}H_{ABCD}=0.
\end{eqnarray*}
The first equation is Cartan's first structural equation with the requirement that the metric conexion be torsion free. The second equation is Cartan's second structural equation, requiring the Ricci tensor to coincide with the appropriate combination of the trace free tensor $S_{ab}$ and the scalar $R$. The third and fourth equations define the symmetric spinors $\phi_{AB}$ and $\Omega_{AB}$ respectively. The rest of the equations have already been considered.\\
We want to calculate, using our particular gauge, a formal expansion of the conformal fields using the initial data in the form $\phi(u,v)$, $S_0(u,v)$. As the system of conformal stationary vacuum field equations is an overdetermined system, we have to choose a subsystem of it. In the rest of this section we choose a particular subsystem, writing the chosen equations in our gauge, and at the end we see how a formal expansion is determined by these equations and the initial data.

\subsection{The $A_{00}=0$ equation}
The first equation that needs particular attention is the equation $A_{00}=0$. In our gauge it reads
\begin{equation*}
\partial_{u}\phi=\phi_{00}.
\end{equation*}
This equation is used in the following to calculate $\phi_{00}$ each time we know $\phi$ as a function of $u$. In particular, as $\phi$ will be prescribed on $W_0$ as part of the initial data, this equation allows us to calculate $\phi_{00}$ there inmediately.

\subsection{The `$\partial_u$-equations'}\label{dUequations}
We now present what we will refer to as the `$\partial_u$-equations'. These equations are chosen because they have the following features. They are a system of PDE's for the set of functions $\hat{e}^a\,_{A1}$, $\hat{\Gamma}_{A1CD}$, $\Omega$, $\Omega_{AB}$, $\phi_{A1}$, $R$, $S_{1}$, $S_{2}$, $S_{3}$ and $S_{4}$, which comprise all the unknowns with the exceptions of the free data $\phi$, $S_{0}$ and the derived function $\phi_{00}$. They are all interior equations on the hypersurfaces $\{w=w_0\}$ in the sense that only derivatives in the directions of $u$ and $v$ are involved, in particular, if we consider the hypersurface $\{w=0\}$, they are all inner equations in ${\cal N}_i$. Also they split into a hierarchy that will be presented in the next section.\\
The $\partial_u$-equations:\\
Equations $t_{AB}\,^{EF}\,_{00}e^{a}\,_{EF}=0:$
\begin{eqnarray*}
&& \partial_{u}\hat{e}^{1}\,_{01}+\frac{1}{u}\hat{e}^{1}\,_{01}=-2\hat{\Gamma}_{0101}+2\hat{\Gamma}_{0100}\hat{e}^{1}\,_{01},\\
&& \partial_{u}\hat{e}^{2}\,_{01}+\frac{1}{u}\hat{e}^{2}\,_{01}=\frac{1}{u}\hat{\Gamma}_{0100}+2\hat{\Gamma}_{0100}\hat{e}^{2}\,_{01},\\
&& \partial_{u}\hat{e}^{1}\,_{11}=-2\hat{\Gamma}_{1101}+2\hat{\Gamma}_{1100}\hat{e}^{1}\,_{01},\\
&& \partial_{u}\hat{e}^{2}\,_{11}=\frac{1}{u}\hat{\Gamma}_{1100}+2\hat{\Gamma}_{1100}\hat{e}^{2}\,_{01}.
\end{eqnarray*}
Equations $R_{AB00EF}=0:$
\begin{eqnarray*}
&& \partial_{u}\hat{\Gamma}_{0100}+\frac{2}{u}\hat{\Gamma}_{0100}-2\hat{\Gamma}^{2}_{0100}=\frac{1}{2}S_{0},\\
&& \partial_{u}\hat{\Gamma}_{0101}+\frac{1}{u}\hat{\Gamma}_{0101}-2\hat{\Gamma}_{0100}\hat{\Gamma}_{0101}=\frac{1}{2}S_{1},\\
&& \partial_{u}\hat{\Gamma}_{0111}+\frac{1}{u}\hat{\Gamma}_{0111}-2\hat{\Gamma}_{0100}\hat{\Gamma}_{0111}=\frac{1}{2}S_{2}-\frac{1}{12}R,\\
&& \partial_{u}\hat{\Gamma}_{1100}+\frac{1}{u}\hat{\Gamma}_{1100}-2\hat{\Gamma}_{0100}\hat{\Gamma}_{1100}=S_{1},\\
&& \partial_{u}\hat{\Gamma}_{1101}-2\hat{\Gamma}_{1100}\hat{\Gamma}_{0101}=S_{2}+\frac{1}{12}R,\\
&& \partial_{u}\hat{\Gamma}_{1111}-2\hat{\Gamma}_{1100}\hat{\Gamma}_{0111}=S_{3}.
\end{eqnarray*}
Equation $\Sigma_{00}=0:$
\begin{eqnarray*}
\partial_{u}\Omega=\Omega_{00}.
\end{eqnarray*}
Equations $\Phi_{A0}=0:$
\begin{equation*}
\partial_{u}\phi_{01}=\frac{1}{2u}(\partial_{v}\phi_{00}-2\phi_{01})+\hat{e}^{1}\,_{01}\partial_{u}\phi_{00}+\hat{e}^{2}\,_{01}\partial_{u}\phi_{00}-2\hat{\Gamma}_{0101}\phi_{00}+2\hat{\Gamma}_{0100}\phi_{01},
\end{equation*}
\begin{eqnarray*}
\partial_{u}\phi_{11}-\frac{1}{2u}\left(\partial_{v}\phi_{01}-\phi_{11}\right)-\hat{e}^{1}\,_{01}\partial_{u}\phi_{01}-\hat{e}^{2}\,_{01}\partial_{v}\phi_{01}=-\hat{\Gamma}_{0111}\phi_{00}+\hat{\Gamma}_{0100}\phi_{11}\\
-\frac{1}{4}\phi\bigg\{R+\frac{10}{1+\Omega-\phi^2}\bigg[\frac{1}{2}\phi^2\Omega_{00}\Omega_{11}-\frac{1}{2}\left[\phi\Omega_{01}-2(1+\Omega)\phi_{01}\right]^2\\
-(1+\Omega)\phi\left(\Omega_{00}\phi_{11}+\Omega_{11}\phi_{00}\right)+2(1+\Omega)^2\phi_{00}\phi_{11}\bigg]\bigg\}.
\end{eqnarray*}
Equations $\Sigma_{00CD}=0:$
\begin{eqnarray*}
\partial_{u}\Omega_{00}&=&-\Omega S_{0}+\frac{1}{1+\Omega-\phi^2}\bigg\{\frac{1}{2}\left[1+(-1+\Omega)\phi^2\right]\Omega^{2}_{00}\\
&&-2\Omega^2\phi\Omega_{00}\phi_{00}+2\Omega^2(1+\Omega)\phi^2_{00}\bigg\},\\
\partial_{u}\Omega_{01}&=&-\Omega S_{1}+\frac{1}{1+\Omega-\phi^2}\bigg\{\frac{1}{2}\left[1+(-1+\Omega)\phi^2\right]\Omega_{00}\Omega_{01}\\
&&-\Omega^2\phi\left(\Omega_{00}\phi_{01}+\Omega_{01}\phi_{00}\right)+2\Omega^2(1+\Omega)\phi_{00}\phi_{01}\bigg\},
\end{eqnarray*}
\begin{eqnarray*}
\partial_{u}\Omega_{11}&=&-\Omega S_{2}-\frac{1}{3}(1+\Omega)R\\
&&+\frac{1}{3(1+\Omega-\phi^2)}\bigg\{\frac{1}{2}\left[3-(11+9\Omega)\phi^2\right]\Omega_{00}\Omega_{11}\\
&&+2(2+3\Omega)\left[\phi\Omega_{01}-2(1+\Omega)\phi_{01}\right]^2\\
&&+(8+20\Omega+9\Omega^2)\left[\phi(\Omega_{00}\phi_{11}+\Omega_{11}\phi_{00})-2(1+\Omega)\phi_{00}\phi_{11}\right]\bigg\}.
\end{eqnarray*}
Equation $\Pi_{00}=0:$
\begin{eqnarray*}
\partial_{u}R-\frac{1}{1+\Omega-\phi^2}\bigg\{\frac{}{}2(4+7\Omega)\phi(\Omega_{11}\partial_{u}\phi_{00}-2\Omega_{01}\partial_{u}\phi_{01}+\Omega_{00}\partial_{u}\phi_{11})\\
-4(1+\Omega)(4+7\Omega)(\phi_{11}\partial_{u}\phi_{00}-2\phi_{01}\partial_{u}\phi_{01}+\phi_{00}\partial_{u}\phi_{11})\bigg\}\\
=\frac{1}{1+\Omega-\phi^2}\bigg\{[3+(-3+7\Omega)\phi^2](\Omega_{11}S_{0}-2\Omega_{01}S_{1}+\Omega_{00}S_{2})\\
-2\Omega(4+7\Omega)\phi(\phi_{11}S_{0}-2\phi_{01}S_{1}+\phi_{00}S_{2})\\
+\frac{1}{3}(4+7\Omega)\phi\left[\phi\Omega_{00}-2(1+\Omega)\phi_{00}\right]R\bigg\}\\
+\frac{1}{3(1+\Omega-\phi^2)^2}\bigg\{\frac{}{}\phi^2\left[-12+(40+21\Omega)\phi^2\right](\Omega_{00}\Omega_{11}-\Omega^2_{01})\Omega_{00}\\
-2\phi\left[-18(1+\Omega)+(46+61\Omega+21\Omega^2)\phi^2\right](\Omega_{00}\phi_{11}-2\Omega_{01}\phi_{01}+\Omega_{11}\phi_{00})\Omega_{00}\\
+4(1+\Omega)\left[-24(1+\Omega)+(52+61\Omega+21\Omega^2)\phi^2\right](\phi_{00}\phi_{11}-\phi^2_{01})\Omega_{00}\\
-2\phi\left[12(1+\Omega)+(16+61\Omega+21\Omega^2)\phi^2\right](\Omega_{00}\Omega_{11}-\Omega^2_{01})\phi_{00}\\
+4(1+\Omega)\left[6(1+\Omega)+(22+61\Omega+21\Omega^2)\phi^2\right](\Omega_{00}\phi_{11}-2\Omega_{01}\phi_{01}+\Omega_{11}\phi_{00})\phi_{00}\\
-8(1+\Omega)^2(7+3\Omega)(4+7\Omega)\phi(\phi_{00}\phi_{11}-\phi^2_{01})\phi_{00}\bigg\}.
\end{eqnarray*}
Equations $H_{0(ABC)_{k}}=0,\mbox{\scriptsize{k=0,1,2,3}}:$
\begin{eqnarray*}
\partial_{u}S_{1}-\frac{1}{2u}\left(\partial_{v}S_{0}-4S_{1}\right)-\hat{e}^{1}\,_{01}\partial_{u}S_{0}-\hat{e}^{2}\,_{01}\partial_{v}S_{0}\\
+\frac{1}{1+\Omega-\phi^2}\big\{\Omega\left[\phi\Omega_{01}-2(1+\Omega)\phi_{01}\right]\partial_{u}\phi_{00}-\Omega\left[\phi\Omega_{00}-2(1+\Omega)\phi_{00}\right]\partial_{u}\phi_{01}\big\}\\
=-4\hat{\Gamma}_{0101}S_{0}+4\hat{\Gamma}_{0100}S_{1}\\
-\frac{1}{1+\Omega-\phi^2}\bigg\{\frac{1}{2}\left[1+(-1+\Omega)\phi^2\right](\Omega_{01}S_{0}-\Omega_{00}S_{1})-\Omega^2\phi(\phi_{01}S_{0}-\phi_{00}S_{1})\\
+2\left[\phi\Omega_{00}-2(1+\Omega)\phi_{00}\right](\Omega_{00}\phi_{01}-\Omega_{01}\phi_{00})\bigg\},
\end{eqnarray*}
\begin{eqnarray*}
\partial_{u}S_{2}-\frac{1}{2u}(\partial_{v}S_{1}-3S_{2})-\hat{e}^1\,_{01}\partial_{u}S_{1}-\hat{e}^2\,_{01}\partial_{v}S_{1}\\
+\frac{1}{3(1+\Omega-\phi^2)}\bigg\{\frac{}{}-(1+\Omega)\left[\phi\Omega_{11}-2(1+\Omega)\phi_{11}\right]\partial_{u}\phi_{00}\\
+(2+5\Omega)\left[\phi\Omega_{01}-2(1+\Omega)\phi_{01}\right]\partial_{u}\phi_{01}-(1+2\Omega)\left[\phi\Omega_{00}-2(1+\Omega)\phi_{00}\right]\partial_{u}\phi_{11}\\
-2\Omega\left[\phi\Omega_{00}-2(1+\Omega)\phi_{00}\right]\left[\frac{1}{2u}(\partial_{v}\phi_{01}-\phi_{11})+\hat{e}^1\,_{01}\partial_{u}\phi_{01}+\hat{e}^2\,_{01}\partial_{v}\phi_{01}\right]\bigg\}\\
=-\hat{\Gamma}_{0111}S_{0}-2\hat{\Gamma}_{0101}S_{1}+3\hat{\Gamma}_{0100}S_{2}\\
-\frac{1}{1+\Omega-\phi^2}\bigg\{\frac{2}{3}\Omega\left[\phi\Omega_{00}-2(1+\Omega)\phi_{00}\right]\left[\hat{\Gamma}_{0111}\phi_{00}-\hat{\Gamma}_{0100}\phi_{11}\right]\\
+\frac{1}{2}\left[1+(-1+\Omega)\phi^2\right](\Omega_{01}S_{1}-\Omega_{00}S_{2})-\Omega^2\phi(\phi_{01}S_{1}-\phi_{00}S_{2})\\
-\frac{1}{6}\Omega\phi^2(\Omega_{11}S_{0}-2\Omega_{01}S_{1}+\Omega_{00}S_{2})+\frac{1}{3}\Omega(1+\Omega)\phi(\phi_{11}S_{0}-2\phi_{01}S_{1}+\phi_{00}S_{2})\\
-\frac{1}{18}(1+\Omega)\phi\left[\phi\Omega_{00}-2(1+\Omega)\phi_{00}\right]R\\
+2\left[\phi\Omega_{01}-2(1+\Omega)\phi_{01}\right](\Omega_{00}\phi_{01}-\Omega_{01}\phi_{00})\bigg\}\\
+\frac{1}{9(1+\Omega-\phi^2)^2}\bigg\{\frac{1}{4}\phi^2\left[-6+(20+3\Omega)\phi^2\right](\Omega_{00}\Omega_{11}-\Omega^2_{01})\Omega_{00}\\
-\frac{1}{2}(14+23\Omega+3\Omega^2)\phi^3(\Omega_{00}\phi_{11}-2\Omega_{01}\phi_{01}+\Omega_{11}\phi_{00})\Omega_{00}\\
+(1+\Omega)\left[6(1+\Omega)+(8+23\Omega+3\Omega^2)\phi^2\right](\phi_{00}\phi_{11}-\phi^2_{01})\Omega_{00}\\
-\frac{1}{2}\phi\left[-12(1+\Omega)+(26+23\Omega+3\Omega^2)\phi^2\right](\Omega_{00}\Omega_{11}-\Omega^2_{01})\phi_{00}\\
+(1+\Omega)^2\left[-6+(20+3\Omega)\phi^2\right](\Omega_{00}\phi_{11}-2\Omega_{01}\phi_{01}+\Omega_{11}\phi_{00})\phi_{00}\\
-2(1+\Omega)^2\left(7+\Omega\right)\left(2+3\Omega\right)\phi(\phi_{00}\phi_{11}-\phi^2_{01})\phi_{00}\bigg\},
\end{eqnarray*}
\begin{eqnarray*}
\partial_{u}S_{3}-\frac{1}{2u}(\partial_{v}S_{2}-2S_{3})-\hat{e}^1\,_{01}\partial_{u}S_{2}-\hat{e}^2\,_{01}\partial_{v}S_{2}+\frac{1}{3(1+\Omega-\phi^2)}\\
\bigg\{-2(1+\Omega)\left[\phi\Omega_{11}-2(1+\Omega)\phi_{11}\right]\partial_{u}\phi_{01}+\Omega\left[\phi\Omega_{01}-2(1+\Omega)\phi_{01}\right]\partial_{u}\phi_{11}\\
+2(2+3\Omega)\left[\phi\Omega_{01}-2(1+\Omega)\phi_{01}\right]\left[\frac{1}{2u}(\partial_{v}\phi_{01}-\phi_{11})+\hat{e}^1\,_{01}\partial_{u}\phi_{01}+\hat{e}^2\,_{01}\partial_{v}\phi_{01}\right]\\
-(2+5\Omega)\left[\phi\Omega_{00}-2(1+\Omega)\phi_{00}\right]\left[\frac{1}{2u}\partial_{v}\phi_{11}+\hat{e}^1\,_{01}\partial_{u}\phi_{11}+\hat{e}^2\,_{01}\partial_{v}\phi_{11}\right]\bigg\}\\
=-2\hat{\Gamma}_{0111}S_{1}+2\hat{\Gamma}_{0100}S_{3}\\
-\frac{1}{1+\Omega-\phi^2}\bigg\{-\frac{2}{3}(2+3\Omega)\left[\phi\Omega_{01}-2(1+\Omega)\phi_{01}\right]\left[\hat{\Gamma}_{0111}\phi_{00}-\hat{\Gamma}_{0100}\phi_{11}\right]\\
+\frac{2}{3}(2+5\Omega)\left[\phi\Omega_{00}-2(1+\Omega)\phi_{00}\right]\left[\hat{\Gamma}_{0111}\phi_{01}-\hat{\Gamma}_{0101}\phi_{11}\right]\\
+\frac{1}{2}\left[1+(-1+\Omega)\phi^2\right](\Omega_{01}S_{2}-\Omega_{00}S_{3})-\Omega^2\phi(\phi_{01}S_{2}-\phi_{00}S_{3})\\
-\frac{1}{3}\Omega\phi^2(\Omega_{11}S_{1}-2\Omega_{01}S_{2}+\Omega_{00}S_{3})+\frac{2}{3}\Omega(1+\Omega)\phi(\phi_{11}S_{1}-2\phi_{01}S_{2}+\phi_{00}S_{3})\\
-\frac{1}{9}(1+\Omega)\phi\left[\phi\Omega_{01}-2(1+\Omega)\phi_{01}\right]R\\
\left.\frac{}{}+2\left[\phi\Omega_{11}-2(1+\Omega)\phi_{11}\right](\Omega_{00}\phi_{01}-\Omega_{01}\phi_{00})\right\}\\
+\frac{2}{9(1+\Omega-\phi^2)^2}\bigg\{\frac{1}{4}\phi^2\left[-6+(20+3\Omega)\phi^2\right](\Omega_{00}\Omega_{11}-\Omega^2_{01})\Omega_{01}\\
-\frac{1}{2}(14+23\Omega+3\Omega^2)\phi^3(\Omega_{00}\phi_{11}-2\Omega_{01}\phi_{01}+\Omega_{11}\phi_{00})\Omega_{01}\\
+(1+\Omega)\left[6(1+\Omega)+(8+23\Omega+3\Omega^2)\phi^2\right](\phi_{00}\phi_{11}-\phi^2_{01})\Omega_{01}\\
-\frac{1}{2}\phi\left[-12(1+\Omega)+(26+23\Omega+3\Omega^2)\phi^2\right](\Omega_{00}\Omega_{11}-\Omega^2_{01})\phi_{01}\\
+(1+\Omega)^2\left[-6+(20+3\Omega)\phi^2\right](\Omega_{00}\phi_{11}-2\Omega_{01}\phi_{01}+\Omega_{11}\phi_{00})\phi_{01}\\
-2(1+\Omega)^2\left(7+\Omega\right)\left(2+3\Omega\right)\phi(\phi_{00}\phi_{11}-\phi^2_{01})\phi_{01}\bigg\},
\end{eqnarray*}
\begin{eqnarray*}
\partial_{u}S_{4}-\frac{1}{2u}(\partial_{v}S_{3}-S_{4})-\hat{e}^1\,_{01}\partial_{u}S_{3}-\hat{e}^2\,_{01}\partial_{v}S_{3}+\\
\frac{1}{1+\Omega-\phi^2}\bigg\{-(1+\Omega)\left[\phi\Omega_{11}-2(1+\Omega)\phi_{11}\right]\partial_{u}\phi_{11}\\
+(2+3\Omega)\left[\phi\Omega_{01}-2(1+\Omega)\phi_{01}\right]\left[\frac{1}{2u}\partial_{v}\phi_{11}+\hat{e}^1\,_{01}\partial_{u}\phi_{11}+\hat{e}^2\,_{01}\partial_{v}\phi_{11}\right]\\
-(1+2\Omega)\left[\phi\Omega_{00}-2(1+\Omega)\phi_{00}\right]\left(\hat{e}^1\,_{11}\partial_u\phi_{11}+\hat{e}^2\,_{11}\partial_v\phi_{11}+\partial_w\phi_{11}\right)\bigg\}\\
=-3\hat{\Gamma}_{0111}S_{2}+2\hat{\Gamma}_{0101}S_{3}+\hat{\Gamma}_{0100}S_{4}\\
+\frac{1}{1+\Omega-\phi^2}\bigg\{2(2+3\Omega)\left[\phi\Omega_{01}-2(1+\Omega)\phi_{01}\right]\left[\hat{\Gamma}_{0111}\phi_{01}-\hat{\Gamma}_{0101}\phi_{11}\right]\\
-2(1+2\Omega)\left[\phi\Omega_{00}-2(1+\Omega)\phi_{00}\right]\left[\hat{\Gamma}_{1111}\phi_{01}-\hat{\Gamma}_{1101}\phi_{11}\right]\\
-\frac{1}{2}\left[1+(-1+\Omega)\phi^2\right](\Omega_{01}S_{3}-\Omega_{00}S_{4})+\Omega^2\phi(\phi_{01}S_{3}-\phi_{00}S_{4})\\
+\frac{1}{2}\Omega\phi^2(\Omega_{11}S_{2}-2\Omega_{01}S_{3}+\Omega_{00}S_{4})-\Omega(1+\Omega)\phi(\phi_{11}S_{2}-2\phi_{01}S_{3}+\phi_{00}S_{4})\\
+\frac{1}{6}(1+\Omega)\phi\left[\phi\Omega_{11}-2(1+\Omega)\phi_{11}\right]R\\
-2\phi\left[\Omega_{11}(\Omega_{01}\phi_{01}-\Omega_{11}\phi_{00})+\phi_{11}(\Omega_{00}\Omega_{11}-\Omega^2_{01})\right]\\
\left.\frac{}{}+4(1+\Omega)\left[\phi_{11}(\Omega_{01}\phi_{01}-\Omega_{00}\phi_{11})+\Omega_{11}(\phi_{00}\phi_{11}-\phi^2_{01})\right]\right\}\\
+\frac{1}{3(1+\Omega-\phi^2)^2}\bigg\{\frac{1}{4}\phi^2\left[-6+(20+3\Omega)\phi^2\right](\Omega_{00}\Omega_{11}-\Omega^2_{01})\Omega_{11}\\
-\frac{1}{2}(14+23\Omega+3\Omega^2)\phi^3(\Omega_{00}\phi_{11}-2\Omega_{01}\phi_{01}+\Omega_{11}\phi_{00})\Omega_{11}\\
+(1+\Omega)\left[6(1+\Omega)+(8+23\Omega+3\Omega^2)\phi^2\right](\phi_{00}\phi_{11}-\phi^2_{01})\Omega_{11}\\
-\frac{1}{2}\phi\left[-12(1+\Omega)+(26+23\Omega+3\Omega^2)\phi^2\right](\Omega_{00}\Omega_{11}-\Omega^2_{01})\phi_{11}\\
+(1+\Omega)^2\left[-6+(20+3\Omega)\phi^2\right](\Omega_{00}\phi_{11}-2\Omega_{01}\phi_{01}+\Omega_{11}\phi_{00})\phi_{11}\\
-2(1+\Omega)^2\left(7+\Omega\right)\left(2+3\Omega\right)\phi(\phi_{00}\phi_{11}-\phi^2_{01})\phi_{11}\bigg\}
\end{eqnarray*}

\subsection{The $\partial_u$-equations hierarchy}\label{hierarchy}
The system of $\partial_u$-equations splits into two groups, referred to as G1 and G2. Each of these groups splits into a hierarchy, which is seen as follows:\\
\\
G1.1: $R_{000001}=0$,\\
G1.2: $t_{01}\,^{EF}\,_{00}e^{2}\,_{EF}=0$,\\
G1.3: $t_{01}\,^{EF}\,_{00}e^{1}\,_{EF}=0,R_{010001}=0,\Sigma_{00}=0,\Sigma_{0000}=0,\Sigma_{0001}=0,H_{0000}=0$,\\
G1.4: $R_{110001}=0,\Sigma_{0011}=0,\Phi_{10}=0,\Pi_{00}=0,H_{0001}=0$,\\
G1.5: $R_{000011}=0$,\\
G1.6: $R_{010011}=0$,\\
G1.7: $t_{11}\,^{EF}\,_{00}e^{1}\,_{EF}=0$,\\
G1.8: $t_{11}\,^{EF}\,_{00}e^{2}\,_{EF}=0$,\\
\\
G2.1: $H_{0011}=0$,\\
G2.2: $R_{110011}=0$,\\
G2.3: $H_{0111}=0$.\\
\\
For dealing with the unknowns we separate them into three groups, $x_1$, $x_2$ and $x_3$. The unknowns involved in G1 are collected in $x_1$, that is $x_1$ $=$ ($\hat{e}^1\,_{01}$, $\hat{e}^2\,_{01}$, $\hat{e}^1\,_{11}$, $\hat{e}^2\,_{11}$, $\hat{\Gamma}_{0100}$, $\hat{\Gamma}_{0101}$, $\hat{\Gamma}_{0111}$, $\hat{\Gamma}_{1100}$, $\hat{\Gamma}_{1101}$, $\Omega$, $\Omega_{00}$, $\Omega_{01}$, $\Omega_{11}$, $\phi_{01}$, $\phi_{11}$, $R$, $S_{1}$, $S_{2}$). The set $x_2$ consist of the unknowns of $x_1$ plus $\phi$, $S_0$ and $\phi_{00}$. The unknowns in G2 are collected in $x_3$, that is $x_3$ $=$ ($\hat{\Gamma}_{1111}$, $S_3$, $S_4$). So all the unknowns are included in the union of $x_2$ and $x_3$.\\
The hierarchy is defined because it makes the following procedure possible. If $\phi$ and $S_{0}$ are prescribed on $\{w=w_0\}$ then G1.1 reduces to an ODE. Once we have its solution, G1.2 reduces to an ODE. Given its solution, G1.3 reduces to a system of ODE's, with coefficients that are calculated by operations interior to $\{w=w_0\}$ from the previously known or calculated functions. This procedure continues till G1.8. So, given $\phi$ and $S_{0}$ on $\{w=w_0\}$ and the appropriate inital data on $U_0\cap\{w=w_0\}$, the set $x_1$ can be determined on $\{w=w_0\}$ by solving a sequence of ODE's in the independent variable $u$.\\
The process to be followed with G2 is very similar, with the exception that to solve G2.3 it is necesary to know also $\partial_{w}\phi_{11}$ on $\{w=w_0\}$, this problem can be overcome solving G1 recursively and then analysing G2.

\subsection{The `$\partial_w$-equations'}
Our initial data, $\phi$ and $S_0$, is prescribed on $W_0$, and to determine their evolution off $W_0$ we need the equation $A_{11}=0$, which reads
\begin{eqnarray*}
\partial_{w}\phi+\hat{e}^1\,_{11}\partial_{u}\phi+\hat{e}^2\,_{11}\partial_{v}\phi=\phi_{11},
\end{eqnarray*}
and the equation $H_{1(ABC)_{0}}+H_{0(ABC)_{1}}=0$, which is given by
\begin{eqnarray*}
\partial_{w}S_{0}-\partial_{u}S_{2}+\hat{e}^1\,_{11}\partial_{u}S_{0}+\hat{e}^2\,_{11}\partial_{v}S_{0}\\
-\frac{1}{3(1+\Omega-\phi^2)}\bigg\{\frac{}{}(2+5\Omega)\left[\phi\Omega_{11}-2(1+\Omega)\phi_{11}\right]\partial_{u}\phi_{00}\\
-4(1+\Omega)\left[\phi\Omega_{01}-2(1+\Omega)\phi_{01}\right]\partial_{u}\phi_{01}+(2+\Omega)\left[\phi\Omega_{00}-2(1+\Omega)\phi_{00}\right]\partial_{u}\phi_{11}\\
-2\Omega\left[\phi\Omega_{00}-2(1+\Omega)\phi_{00}\right]\left[\frac{1}{2u}(\partial_{v}\phi_{01}-\phi_{11})+\hat{e}^1_{01}\partial_{u}\phi_{01}+\hat{e}^2_{01}\partial_{v}\phi_{01}\right]\bigg\}\\
=4\hat{\Gamma}_{1101}S_{0}-4\hat{\Gamma}_{1100}S_{1}\\
+\frac{1}{1+\Omega-\phi^2}\bigg\{\frac{2}{3}\Omega\left[\phi\Omega_{00}-2(1+\Omega)\phi_{00}\right]\left[\hat{\Gamma}_{0111}\phi_{00}-\hat{\Gamma}_{0100}\phi_{11}\right]\\
+\frac{1}{2}\left[1+(-1+\Omega)\phi^2\right](\Omega_{11}S_{0}-\Omega_{00}S_{2})-\Omega^2\phi(\phi_{11}S_{0}-\phi_{00}S_{2})\\
+\frac{1}{3}\Omega\phi^2(\Omega_{11}S_{0}-2\Omega_{01}S_{1}+\Omega_{00}S_{2})-\frac{2}{3}\Omega(1+\Omega)\phi(\phi_{11}S_{0}-2\phi_{01}S_{1}+\phi_{00}S_{2})\\
+\frac{1}{9}(1+\Omega)\phi\left[\phi\Omega_{00}-2(1+\Omega)\phi_{00}\right]R\\
+2\left[\phi\Omega_{00}-2(1+\Omega)\phi_{00}\right](\Omega_{00}\phi_{11}-\Omega_{11}\phi_{00})\bigg\}\\
+\frac{2}{9(1+\Omega-\phi^2)^2}\bigg\{\frac{1}{4}\phi^2\left[-6+(20+3\Omega)\phi^2\right](\Omega_{00}\Omega_{11}-\Omega^2_{01})\Omega_{00}\\
-\frac{1}{2}(14+23\Omega+3\Omega^2)\phi^3(\Omega_{00}\phi_{11}-2\Omega_{01}\phi_{01}+\Omega_{11}\phi_{00})\Omega_{00}\\
+(1+\Omega)\left[6(1+\Omega)+(8+23\Omega+3\Omega^2)\phi^2\right](\phi_{00}\phi_{11}-\phi^2_{01})\Omega_{00}\\
-\frac{1}{2}\phi\left[-12(1+\Omega)+(26+23\Omega+3\Omega^2)\phi^2\right](\Omega_{00}\Omega_{11}-\Omega^2_{01})\phi_{00}\\
+(1+\Omega)^2\left[-6+(20+3\Omega)\phi^2\right](\Omega_{00}\phi_{11}-2\Omega_{01}\phi_{01}+\Omega_{11}\phi_{00})\phi_{00}\\
-2(1+\Omega)^2\left(7+\Omega\right)\left(2+3\Omega\right)\phi(\phi_{00}\phi_{11}-\phi^2_{01})\phi_{00}\bigg\}.
\end{eqnarray*}
These two equations will be referred to as the `$\partial_w$-equations'.

\subsection{The initial conditions for the $\partial_u$-equations}\label{initialCond}
The initial conditions for the $\partial_u$-equations follow from our gauge conditions \eqref{FrameGauge}, \eqref{GammaGauge} which imply
\begin{eqnarray*}
&& \hat{e}^a\,_{A1}|_{I}=0,\,\,\,a=1,2,\,\,A=0,1,\\
&& \hat{\Gamma}_{A1CD}|_{I}=0,\,\,\,A,C,D=0,1.
\end{eqnarray*}
From \eqref{condOmega} we get
\begin{eqnarray*}
&& \Omega|_{I}=0,\\
&& \Omega_{AB}|_{I}=0,\,\,\,A,B=0,1,\\
&& R|_{I}=-6-8\partial_{u}\phi|_{I}\partial_{u}\partial_{v}^2\phi|_{I}+4\left(\partial_{u}\partial_{v}\phi|_{I}\right)^2,
\end{eqnarray*}
and from the required spinorial behaviour in order to have analytic solutions, as discussed in Section \ref{ExpansionType},
\begin{eqnarray}\label{spinRelations}
&& \phi_{A1}|_{I}=\frac{1}{2}\partial_{u}\partial^{1+A}_{v}\phi|_{I},\,\,\,A=0,1,\\
\nonumber && S_{k}|_{I}=\frac{(4-k)!}{4!}\partial_{v}^kS_{0}|_{I},
\end{eqnarray}
where $A_{00}=0$ has been used.

\subsection{Calculating the formal expansion}
As the system of equations is overdetermined, we have chosen a subsystem in order to calculate a formal expansion of the solution. It will be shown later on that the expansion obtained using this subsystem lead to a formal solution of the full system of equations.\\
We prescribe $\phi$ and $S_0$ on $W_0$ as our datum and the initial conditions on $I$ for the $\partial_u$-equations are given in \ref{initialCond}. Following what has been said in \ref{hierarchy} we successively integrate the subsystems on G1 to determine all components of $x_1$ on $W_0$.\\
We give now an inductive argument involving G1 and the $\partial_w$-equations to show that $\partial_w^kx_2|_{W_0}$ can be determined for all $k$.\\
From our initial data and what has been said we know already $\partial_w^kx_2|_{W_0}$ for $k=0$. As inductive hypothesis we assume as known
\begin{equation*}
\partial^p_{w}x_2|_{W_{0}},\,\,\,0\leq p \leq k-1,\,\,\,k\geq 1.
\end{equation*}
Applying formally $\partial_w^{k-1}$ to the $\partial_w$-equations, and restricting them to $W_0$, we find $\partial^k_{w}\phi|_{W_{0}}$ and $\partial^k_{w}S_{0}|_{W_{0}}$ in terms of known functions. We apply formaly $\partial_w^{k}$ to G1. This is a system of PDE's where the unknowns are $\partial_w^kx_1$. Keeping the hierarchy and considering the functions that we already know on $W_0$, it agains becames a sequence of ODE's, which can be integrated on $W_0$ given the appropriate initial conditions on $I$.\\
The initial conditions for the frame coefficients and the connection coefficients are obtained from the gauge requirements \eqref{FrameGauge}, \eqref{GammaGauge} which imply
\begin{eqnarray*}
\partial^k_{w}\hat{e}^a\,_{A1}|_{I}=0,\,\,\,a=1,2,\,\,A=0,1,\\
\partial^k_{w}\hat{\Gamma}_{A1CD}|_{I}=0,\,\,\,A,C,D=0,1.
\end{eqnarray*}
From the spinorial behaviour as discussed in Section \ref{ExpansionType} we obtain the following set of initial conditions.
\begin{eqnarray*}
\partial^k_{w}\phi_{01}|_{I}=\frac{1}{2}\partial_{u}\partial_{v}\partial^k_{w}\phi|_{I},\\
\partial^k_{w}\phi_{11}|_{I}=\frac{1}{2}\partial_{u}\partial^2_{v}\partial^k_{w}\phi|_{I},\\
\partial^k_{w}S_{1}|_{I}=\frac{1}{4}\partial_{v}\partial^k_{w}S_{0}|_{I},\\
\partial^k_{w}S_{2}|_{I}=\frac{1}{12}\partial^2_{v}\partial^k_{w}S_{0}|_{I}.
\end{eqnarray*}
By restricting the equations $\Sigma_{11}=0$, $\Sigma_{11CD}=0$ and $\Pi_{11}=0$ to $U_0$ and using that $\Omega|_{I}=0$, $\Omega_{AB}|_{I}=0$ we get
\begin{eqnarray*}
\partial^k_{w}\Omega|_{I}=0,\\
\partial^k_{w}\Omega_{A1}|_{I}=0,\,\,\,A=0,1,
\end{eqnarray*}
\begin{equation}\label{dOmegaU0}
\partial_{w}\Omega_{00}|_{U_0}=\left.\left[-\frac{1}{3}R+\frac{8}{3(1-\phi^2)}(\phi\Omega_{00}\phi_{11}-2\phi_{00}\phi_{11}+2\phi^2_{01})\right]\right|_{U_0},
\end{equation}
\begin{eqnarray}\label{dRU0}
\partial_{w}R|_{U_0} & = & \bigg[3\Omega_{00}S_{4}+\frac{8}{3(1-\phi^2)}\bigg\{\frac{}{}(\phi\Omega_{00}-2\phi_{00})\partial_w\phi_{11}\\
\nonumber&& +4\phi_{01}\partial_{w}\phi_{01}-2\phi_{11}\partial_{w}\phi_{00}-\frac{1}{3}\phi R \phi_{11}\bigg\}\\
\nonumber&& +\frac{8}{3(1-\phi^2)^2}\phi_{11}\left\{(3+11\phi^2)\Omega_{00}\phi_{11}-28\phi(\phi_{00}\phi_{11}-\phi^2_{01})\right\}\bigg]\bigg|_{U_0}.
\end{eqnarray}
Applying $\partial_w^{k-1}$ to \eqref{dOmegaU0}, \eqref{dRU0} and evaluating them at $I$ by using the known functions from the inductive hypothesis and the previously stated initial conditions we get $\partial^k_{w}\Omega_{00}|_I$ and $\partial^k_{w}R|_I$.\\
Now we have all the needed initial conditions, thus we know
\begin{equation*}
\partial^k_{w}x_2|_{W_{0}}
\end{equation*}
and the induction step is completed.\\
The procedure with G2 is quite similar. Once we know $\partial^k_{w}x_2|_{W_{0}}$ for all $k$, G2.1 reduces to an ODE, which can be integrated on $U_0$ given the corresponding initial condition. Once we know the solution of G2.1, G2.2 also reduces to an ODE, and finally also G2.3 reduces to an ODE. The initial conditions for G2 are given in Section \ref{initialCond}.\\
The inductive step is very similar to the inductive step for $x_2$. We assume
\begin{equation*}
\partial^p_{w}x_3|_{W_{0}},\,\,\,0\leq p \leq k-1,\,\,\,k\geq 1,
\end{equation*}
to be known. We apply formaly $\partial_w^{k}$ to the equations in G2. If we stick to the hierarchy this system again reduces in the prescribed order to a system of ODE's for $\partial_w^kx_3$, which can be integrated given the corresponding initial conditions. Those are
\begin{eqnarray*}
&& \partial^k_{w}\hat{\Gamma}_{1111}|_{I}=0,\\
&& \partial^k_{w}S_{3}|_{I}=\frac{1}{24}\partial^3_{v}\partial^k_{w}S_{0}|_{I},\\
&& \partial^k_{w}S_{4}|_{I}=\frac{1}{24}\partial^4_{v}\partial^k_{w}S_{0}|_{I},
\end{eqnarray*}
obtained from \eqref{GammaGauge} and Section \ref{ExpansionType}.\\
Now we know
\begin{equation*}
\partial^k_{w}x_3|_{W_{0}}
\end{equation*}
and the induction step is complete.\\
If we now call $X$ any of the quantities included in $x_2$ and $x_3$, that is, $X$ comprises all the unknown quantities that we are solving for, the procedure just stated shows that we know $\partial^k_{w}X|_{W_{0}}$ for all $k$. Expanding these functions around $i=\{u=0,v=0,w=0\}$ gives
\begin{equation*}
\partial_u^m\partial_v^n\partial_w^pX|_i\,\,\,\forall\,\,m,n,p,
\end{equation*}
and the procedure gives a unique sequence of expansion coefficients for all the functions in $X$.
\begin{lem}\label{expansionCoefficients}
The procedure described above determines at the point $O=(u=0,v=0,w=0)$ from the data $\phi$, $S_0$, given on $W_0$ according to \eqref{expansionPhi}, \eqref{expansionS}, a unique sequence of expansion coefficients
\begin{equation*}
\partial_u^m\partial_v^n\partial_w^pf(O),\,\,\,m,n,p=0,1,2,...,
\end{equation*}
where $f$ stands for any of the functions $\hat{e}^a\,_{AB}$,$\hat{\Gamma}_{ABCD}$,$\phi$,$\phi_{AB}$,$\Omega$,$\Omega_{AB}$,$R$,$S_k$.\\
If the corresponding Taylor series are absolutely convergent in some neighbourhood $P$ of $O$, they define a solution to the equation $A_{00}=0$, to the $\partial_u$-equations and to the $\partial_w$-equations on $P$ which satisfies on $P\cap U_0$ equations \eqref{spinRelations} and $\Sigma_{11}=0$, $\Sigma_{11CD}=0$, $\Pi_{11}=0$.
\end{lem}
By Lemma \ref{expansionType} we know that all spinor-valued functions should have a specific $v$-finite expansion type. The following lemma, whose proof is quite similiar to the proof in \cite{Friedrich07}, will be important for the convergence proof.
\begin{lem}\label{expansionTypeFields}
If the data $\phi$, $S_0$ are given on $W_0$ as in \eqref{expansionPhi}, \eqref{expansionS}, the formal expansions of the fields obtained in Lemma \ref{expansionCoefficients} correspond to ones of functions of $v$-finite expansion types given by
\begin{eqnarray*}
&& k_{\hat{e}^1\,_{AB}}=-A-B,\,\,\,k_{\hat{e}^2\,_{AB}}=3-A-B,\,\,\,AB=01,11,\\
&& k_{\hat{\Gamma}_{01AB}}=2-A-B,\,\,\,k_{\hat{\Gamma}_{11AB}}=1-A-B,\,\,\,A,B=0,1,\\
&& k_\phi=0,\,\,\,k_{\phi_{AB}}=2-A-B,\,A,B=0,1,\\
&& k_\Omega=0,\,\,\,k_{\Omega_{AB}}=2-A-B,\,A,B=0,1,\\
&& k_R=0,\\
&& k_{S_j}=4-j,\,j=0,1,2,3,4.
\end{eqnarray*}
\end{lem}

\section{Convergence of the formal expansion}\label{convergence}
In the previous section we have seen how to calculate a formal expansion for $\hat{e}^a\,_{AB}$, $\hat{\Gamma}_{ABCD}$, $\phi$, $\phi_{AB}$, $\Omega$, $\Omega_{AB}$, $R$, $S_k$ given $\phi|_{W_0}$ and $S_0|_{W_0}$, or, what is the same, given the null data. From Lemma \ref{estimatesNullData} we know which are the necessary conditions on the null data in order to have analytic solutions of the conformal field equations. In this section we show that those conditions, \eqref{estimates-psi} and \eqref{estimates-Psi}, are also sufficient for the formal expansion determined in the previous section to be absolutely convergent.\\
So we start considering the abstract null data as given by two sequencies
\begin{equation*}
\hat{\cal D}^\phi_n=\{\psi_{A_1B_1},\psi_{A_2B_2A_1B_1},\psi_{A_3B_3A_2B_2A_1B_1},...\},
\end{equation*}
\begin{equation*}
\hat{\cal D}^S_n=\{\Psi_{A_2B_2A_1B_1},\Psi_{A_3B_3A_2B_2A_1B_1},\Psi_{A_4B_4A_2A_3B_3A_2B_2A_1B_1},...\},
\end{equation*}
of totally symmetric spinors satisfying the reality condition \eqref{realityCond} and we construct $\phi|_{W_0}$ and $S_0|_{W_0}$, by setting in the expansions \eqref{expansionPhi},\eqref{expansionS}
\begin{equation*}
D^*_{(A_mB_m}...D^*_{A_1B_1)}\phi(i)=\psi_{A_mB_m...A_1B_1},\,\,\,m\geq 1,
\end{equation*}
\begin{equation*}
D^*_{(A_pB_p}...D^*_{A_1B_1}S^*_{ABCD)}(i)=\Psi_{A_mB_m...A_1B_1ABCD},\,\,\,m\geq 0.
\end{equation*}
Observing Lemma \ref{estimatesNullData}, one finds as a necessary condition for the functions $\phi,S_0$ on $W_0$ to determine an analytic solution to the conformal static vacuum field equations that its non-vanishing Taylor coefficients at the point $O$ satisfy estimates of the form
\begin{equation}\label{firstEstimatePhi}
|\partial_u^m\partial_v^n\phi(0)|\leq\binom{2m}{n}m!n!\frac{M}{r^m},\,\,\,m\geq 0,\,\,\,0\leq n\leq 2m,
\end{equation}
\begin{equation}\label{firstEstimateS}
|\partial_u^m\partial_v^nS_0(0)|\leq\binom{2m+4}{n}m!n!\frac{M}{r^m},\,\,\,m\geq 0,\,\,\,0\leq n\leq 2m+4.
\end{equation}
This conditions are also sufficient for $\phi(u,v)$ and $S_0(u,v)$ to be holomorphic functions on $W_0$. So the null data gives rise to two analytic functions, $\phi$ and $S_0$, on $W_0$.\\
From $A_{00}=0$ we have $\phi_{00}=\partial_u\phi$, so having $\phi|_{W_0}$ we have $\phi_{00}|_{W_0}$, which is also an analytic function on $W_0$.\\
Following Lemma 6.1 in \cite{Friedrich07}, we can derive from \eqref{firstEstimatePhi},\eqref{firstEstimateS}, slightly different type of estimates for $\phi(u,v)$, $S_0(u,v)$, which are more convenient in our case.
\begin{lem}
Let $e$ be the Euler number. For given $\rho_\phi$, $\rho_{S_{0}}$, both in $\mathbb{R}$, such that $0<\rho_\phi<e^2$, $0<\rho_{S_{0}}<e^2$, there exist positive constants $\tilde{c}_\phi$, $r_\phi$, $\tilde{c}_{S_{0}}$, $r_{S_{0}}$, so that \eqref{firstEstimatePhi},\eqref{firstEstimateS}, imply estimates of the form
\begin{equation}\label{EstimatePhi}
|\partial^m_u\partial^n_v\phi|\leq \tilde{c}_\phi \frac{r^{m-1}_\phi m! \rho^n_\phi n!}{(m+1)^2 (n+1)^2},\,\,\,m\geq 0,\,\,\,0\leq n\leq 2m,
\end{equation}
\begin{equation}\label{EstimateS}
|\partial^m_u\partial^n_vS_0|\leq \tilde{c}_{S_{0}} \frac{r^m_{S_{0}} m! \rho^n_{S_{0}} n!}{(m+1)^2 (n+1)^2},\,\,\,m\geq 0,\,\,\,0\leq n\leq 2m+4.
\end{equation}
\end{lem}
We can present our estimates.
\begin{lem}\label{mainEstimates}
Assume $\phi=\phi(u,v)$, $S_0=S_0(u,v)$ are holomorphic funtions defined on some open neighbourhood $U$ of $O=\{u=0,v=0,w=0\}$ in $W_0=\{w=0\}$ which have expansions of the form
\begin{equation*}
\phi(u,v)=\sum_{m=0}^\infty\sum_{n=0}^{2m}\psi_{m,n}u^mv^n,
\end{equation*}
\begin{equation*}
S_0(u,v)=\sum_{m=0}^\infty\sum_{n=0}^{2m+4}\Psi_{m,n}u^mv^n,
\end{equation*}
so that its Taylor coefficients at the point $O$ satisfy estimates of the type \eqref{EstimatePhi},\eqref{EstimateS} with some positive constants $\tilde{c}_\phi$, $r_\phi$, $\tilde{c}_{S_{0}}$, $r_{S_{0}}$, and $\rho_\phi<\tfrac{1}{3}$, $\rho_{S_{0}}<\tfrac{1}{3}$. Then there exist positive constants
\begin{equation*}
r,\,\rho,\,c_{\hat{e}^a\,_{AB}},\,c_{\hat{\Gamma}_{ABCD}},\,c_{\phi},\,c_{\phi_{AB}},\,c_{\Omega},\,c_{\Omega_{AB}},\,c_R,\,c_{S_k}
\end{equation*}
so that the expansion coefficients determined from $\phi$ and $S_0$ in Lemma \ref{expansionCoefficients} satisfy for $m,n,p=0,1,2,...$
\begin{equation}\label{convergenceEstimates}
|\partial_u^m\partial_v^n\partial_w^pf(O)|\leq c_f \frac{r^{m+p+q_f}(m+p)!\rho^nn!}{(m+1)^2(n+1)^2(p+1)^2},
\end{equation}
where $f$ stands for any of the functions
\begin{equation*}
\hat{e}^a\,_{AB},\,\hat{\Gamma}_{ABCD},\,\phi,\,\phi_{AB},\,\Omega,\,\Omega_{AB},\,R,\,S_k, 
\end{equation*}
and
\begin{equation*}
q_{\hat{e}^a\,_{AB}}=q_{\hat{\Gamma}_{ABCD}}=q_{\phi}=q_{\Omega}=q_{\Omega_{AB}}=-1,\,\,\,\,q_{\phi_{AB}}=q_R=q_{S_k}=0.
\end{equation*}
\end{lem}
\begin{remark}
Taking into account the $v$-finite expansion types of the functions $f$ (Lemma \ref{expansionTypeFields}), we can replace the right hand sides in the estimates above by zero if $n$ is large enough relative to $m$. This will not be pointed out at each step and for convenience the estimates will be written as above.
\end{remark}
We take the following four lemmas from \cite{Friedrich07}. The first states the necessary part of the estimates, and the other three are needed in order to manipulate the estimates in the proof of Lemma \ref{mainEstimates}.
\begin{lem}
If $f$ is holomorphic near $O$, there exist positive constants $c$, $r_0$, $\rho_0$ such that
\begin{equation*}
|\partial_u^m\partial_v^n\partial_w^pf(O)|\leq c \frac{r^{m+p}(m+p)!\rho^nn!}{(m+1)^2(n+1)^2(p+1)^2},\,\,\,m,n,p=0,1,2,...
\end{equation*}
for any $r\geq r_0$, $\rho\geq\rho_0$. If in addition $f(0,v,0)=0$, the constants can be chosen such that
\begin{equation*}
|\partial_u^m\partial_v^n\partial_w^pf(O)|\leq c \frac{r^{m+p-1}(m+p)!\rho^nn!}{(m+1)^2(n+1)^2(p+1)^2},\,\,\,m,n,p=0,1,2,...
\end{equation*}
for any $r\geq r_0$, $\rho\geq\rho_0$.
\end{lem}
\begin{lem}
For any non-negative integer $n$ there is a positive constant $C$, $C>1$, independent of $n$ so that
\begin{equation*}
\sum_{k=0}^n\frac{1}{(k+1)^2(n-k+1)^2}\leq C\frac{1}{(n+1)^2}.
\end{equation*}
\end{lem}
In the following $C$ will always denote the constant above.
\begin{lem}
For any integers $m$, $n$, $k$, $j$, with $0\leq k\leq m$, and $0\leq j\leq n$ resp. $0\leq j \leq n-1$ holds
\begin{equation*}
\binom{m}{k}\binom{n}{j}\leq\binom{m+n}{k+j}\mbox{   resp.   }\binom{m}{k}\binom{n-1}{j}\leq\binom{m+n}{k+j}.
\end{equation*}
\end{lem}
\begin{lem}\label{multiplEstimates}
Let $m$, $n$, $p$ be non-negative integers and $f_i$, $i=1,...,N$, be smooth complex valued functions of $u$, $v$, $w$ on some neighbourhood $U$ of $O$ whose derivatives satisfy on $U$ (resp. at a given point $p\in U$) estimates of the form
\begin{equation*}
|\partial_u^j\partial_v^k\partial_w^lf_i|\leq c_i \frac{r^{j+l+q_i}(j+l)!\rho^kk!}{(j+1)^2(k+1)^2(l+1)^2}
\end{equation*}
for $0\leq j\leq m$, $0\leq k\leq n$, $0\leq l\leq p$, with some positive constants $c_i$, $r$, $\rho$ and some fixed integers $q_i$ (independent of $j$, $k$, $l$). Then one has on $U$ (resp. at $p$) the estimates
\begin{equation}\label{estimateMultiplication}
|\partial_u^m\partial_v^n\partial_w^p(f_1\cdot ... \cdot f_N)|\leq C^{3(N-1)}c_1\cdot ... \cdot c_N \frac{r^{m+p+q_1+...q_N}(m+p)!\rho^nn!}{(m+1)^2(n+1)^2(p+1)^2}.
\end{equation}
\end{lem}
\begin{remark}
This lemma remains true if $m$, $n$, $p$ are replaced in \eqref{estimateMultiplication} by integers $m'$, $n'$, $p'$ with $0\leq m'\leq m$, $0\leq n'\leq n$, $0\leq p'\leq p$.\\
The factor $C^{3(N-1)}$ in \eqref{estimateMultiplication} can be replaced by $C^{(3-r)(N-1)}$ if $r$ of the integers $m$, $n$, $p$ vanish.
\end{remark}
\begin{proof}[Proof of Lemma \ref{mainEstimates}]
The proof is by induction, following the procedure which led to Lemma \ref{expansionCoefficients}. A general outline is as follows. We start leaving the choice of the constants $r$, $\rho$, $c_f$, open. We use the induction hypothesis and the equations that lead to Lemma \ref{expansionCoefficients} to derive estimates for the derivatives of the next order. These estimates are of the form
\begin{equation}\label{generalEstimateForm}
|\partial_u^m\partial_v^n\partial_w^pf(O)|\leq c_f \frac{r^{m+p+q_f}(m+p)!\rho^nn!}{(m+1)^2(n+1)^2(p+1)^2}A_f,
\end{equation}
with certain constans $A_f$ which depend on $m$, $n$, $p$ and the constants $c_f$, $r$ and $\rho$. Sometimes superscripts will indicate to which order of differentiability particular constants $A_f$ refer. In the way we will have to make assumptions on $r$ to proceed with the induction step. We shall collect these conditions and the constants $A_f$, or estimates for them, and at the end it will be shown that the constants $c_f$, $r$ and $\rho$ can be adjusted so that all conditions are satisfied and $A_f\leq 1$. This will complete the induction proof.\\
In order not to write long formulas that do not add to the understanding of the procedure, we state here some properties that are used to simplify the estimates:
\begin{itemize}
\item{As a corollary of Lemma \ref{multiplEstimates} we have:\\
If
\begin{equation*}
|\partial_u^j\partial_v^k\partial_w^lg|\leq c_g \frac{r^{j+l-1}(j+l)!\rho^kk!}{(j+1)^2(k+1)^2(l+1)^2}
\end{equation*}
for $0\leq j\leq m$, $0\leq k\leq n$, $0\leq l\leq p$, where $g$ is $\phi$ or $\Omega$, and if $r>\tfrac{C^3}{2}[c_{\Omega}+(c_{\Omega}^2+4c_{\phi}^2)^\frac{1}{2}]$, then
\begin{eqnarray*}
&& \left|\partial_u^m\partial_v^n\partial_w^p\left(\frac{1}{1+\Omega-\phi^2}\right)\right|\\
&& \leq \frac{1}{C^3}\frac{1}{1-\tfrac{C^3}{r}\left(c_{\Omega}+\tfrac{C^3}{r}c_{\phi}^2\right)} \frac{r^{m+p}(m+p)!\rho^nn!}{(m+1)^2(n+1)^2(p+1)^2}.
\end{eqnarray*}
}
\item{If $r\geq C^3[c_{\Omega}+(c_{\Omega}^2+2c_{\phi}^2)^\frac{1}{2}]$ then
\begin{equation}\label{estimateF}
\frac{1}{1-\tfrac{C^3}{r}\left(c_{\Omega}+\tfrac{C^3}{r}c_{\phi}^2\right)}\leq 2.
\end{equation}
}
\item{After calculating the estimates and using \eqref{estimateF} we find that all the $A$'s satisfy inequalities of the form
\begin{equation*}
A\leq\alpha+\sum_{i=1}^9 \frac{\alpha_i}{r^i},
\end{equation*}
where $\alpha,\alpha_i$ are constants that don't depend on $r$. If $\alpha_i=0$ then we have to show that we can make $\alpha\leq 1$. If the $\alpha_i$'s not zero we can take a constant $a$, $0<a<1$, and require that $\alpha\leq a$ and then choose $r$ large enough such that $\sum_{i=1}^8 \frac{\alpha_i}{r^i}\leq 1-a$. In the estimates that follows, we shall not write the explicit expresions for the $\alpha_i$'s, as they do not play any role if we are able to make $r$ big enough at the end of the procedure.
}
\end{itemize}
From now on we consider that a function in a modulus sign is evaluated at the origin $O$.\\
From the analyticity of $\phi_{00}(u,v)$ we also get that, for given $\rho_{\phi_{00}}\in\mathbb{R}$, $0<\rho_{\phi_{00}}<\tfrac{1}{3}$, there exist positive constants $c_{\phi_{00}}$, $r_{\phi_{00}}$, such that
\begin{equation*}
|\partial^m_u\partial^n_v\phi_{00}|\leq c_{\phi_{00}} \frac{r^m_{\phi_{00}} m! \rho^n_{\phi_{00}} n!}{(m+1)^2 (n+1)^2},\,\,\,m\geq 0,\,\,\,0\leq n\leq 2m+2.
\end{equation*}
As $\phi(0,v)=0$ the inequalities \eqref{EstimatePhi},\eqref{EstimateS} are mantained if we change the constants for bigger constants. We choose
\begin{eqnarray}
\label{conditionsC1}& c_\phi=max\{\tilde{c}_{\phi},c_{\phi_{00}}\},\\
\label{conditionsC2}& c_{S_{0}}=max\left\{\tilde{c}_{S_{0}},\tfrac{64}{3}C^3c_{\phi_{00}}^2\right\}.
\end{eqnarray}
Also we require the constants $r,\rho$ to saisfy
\begin{eqnarray}\label{conditionsRrho}
\nonumber & r\geq max\{r_\phi,r_{S_{0}},r_{\phi_{00}}\}, \\
& \rho\geq max\{\rho_\phi,\rho_{S_{0}},\rho_{\phi_{00}}\},
\end{eqnarray}
but we leave the choice of the precise value open.
So we have
\begin{equation*}
|\partial^m_u\partial^n_v\partial_w^0\phi|\leq c_\phi \frac{r^{m-1} m! \rho^n n!}{(m+1)^2 (n+1)^2},\,\,\,m\geq 0,\,\,\,0\leq n\leq 2m,
\end{equation*}
\begin{equation*}
|\partial^m_u\partial^n_v\partial_w^0S_0|\leq c_{S_{0}} \frac{r^m m! \rho^n n!}{(m+1)^2 (n+1)^2},\,\,\,m\geq 0,\,\,\,0\leq n\leq 2m+4,
\end{equation*}
\begin{equation*}
|\partial^m_u\partial^n_v\partial_w^0\phi_{00}|\leq c_{\phi_{00}} \frac{r^m m! \rho^n n!}{(m+1)^2 (n+1)^2},\,\,\,m\geq 0,\,\,\,0\leq n\leq 2m+2.
\end{equation*}
From the frame properties $\hat{e}^a\,_{AB}|_{U_{0}}=0$, $\hat{\Gamma}_{ABCD}|_{U_{0}}=0$ follows
\begin{equation*}
|\partial_u^0\partial_v^n\partial_w^p\hat{\Gamma}_{ABCD}|=0,\,\,\,|\partial_u^0\partial_v^n\partial_w^p\hat{e}^a\,_{AB}|=0.
\end{equation*}
The conditions on the conformal factor, $\Omega|_I=0$, $\Omega_{AB}|_I=0$, give
\begin{equation*}
|\partial_u^0\partial_v^n\partial_w^0\Omega|=0,\,\,\,|\partial_u^0\partial_v^n\partial_w^0\Omega_{AB}|=0.
\end{equation*}
Using Lemma \ref{expansionType} we get the relations:
\begin{equation}\label{phiA1enU0}
\phi_{A1}|_{U_{0}}=\frac{1}{2}\partial^{1+A}_v\phi_{00}|_{U_{0}},\,\,\,A=0,1,
\end{equation}
\begin{equation}\label{SkenU0}
S_k|_{U_{0}}=\frac{(4-k)!}{4!}\partial^k_v S_0|_{U_{0}},\,\,\,k=1,2,3,4,
\end{equation}
which imply
\begin{eqnarray*}
|\partial^0_u\partial^n_v\partial^0_w\phi_{A1}| & \leq & \left\{ 
\begin{array}{c c}
  \frac{1}{2}c_{\phi_{00}}\frac{\rho^{n+1+A}(n+1+A)!}{(n+2+A)^2}, & n\leq 1-A\\
  0, & n>1-A\\ \end{array} \right\} \\
& = & c_{\phi_{A1}}\frac{\rho^n n!}{(n+1)^2}A^{m=0,p=0}_{\phi_{A1}},
\end{eqnarray*}
\begin{eqnarray}\label{AphiA1}
&& A^{m=0,p=0}_{\phi_{A1}}=\frac{1}{2}\frac{c_{\phi_{00}}}{c_{\phi_{A1}}}\rho^{1+A}h_{A,n}\leq \frac{1}{2}\frac{c_{\phi_{00}}}{c_{\phi_{A1}}}\rho^{1+A},\\
\nonumber&& h_{A,n}=\left\{
\begin{array}{c c}
\frac{(n+1+A)!}{n!}\frac{(n+1)^2}{(n+2+A)^2}, & 0\leq n\leq 1-A \\
0, & n>1-A
\end{array}
\right\}\leq 1,
\end{eqnarray}
and similarly
\begin{equation*}
|\partial^0_u\partial^n_v\partial^0_wS_k| \leq c_{S_{k}}\frac{\rho^n n!}{(n+1)^2}A^{m=0,p=0}_{S_k},
\end{equation*}
\begin{equation}\label{ASk}
A^{m=0,p=0}_{S_k}\leq \frac{c_{S_{0}}}{c_{S_{k}}}\rho^k.
\end{equation}
Taking into account that $R$ is a scalar and the initial condition $R|_{i}=-6-16\left(\phi_{00}\phi_{11}-\phi_{01}^2\right)|_i$, we get
\begin{eqnarray*}
|\partial^0_u\partial^n_v\partial^0_wR| & \leq & \left\{
\begin{array}{c c}
6+\frac{73}{36}\rho^2 c_{\phi_{00}}^2, & n=0 \\
0, & n>0
\end{array} \right\} \\
& = & c_{R}\frac{\rho^n n!}{(n+1)^2}A^{m=0,p=0}_{R},
\end{eqnarray*}
\begin{equation*}
A^{m=0,p=0}_{R}\leq \frac{1}{c_{R}}\left(6+\frac{73}{36}\rho^2 c_{\phi_{00}}^2\right).
\end{equation*}
We have obtained so far the estimates for $m=0$, $p=0$ and general $n$. Now we should consider the equations in G1 to get in an inductive form estimates for the quantities in $x_1$, that means, estimates for $|\partial_u^m\partial_v^n\partial_w^0x_1|$, considering as known estimates of this type for $|\partial_u^l\partial_v^n\partial_w^0x_1|$ with $0\leq l <m$. And once we have this estimates we should do the same procedure with G2 to get estimates for $|\partial_u^m\partial_v^n\partial_w^0x_3|$. These estimates, i.e. estimates for $p=0$, can be obtained from the estimates for general $p$ that appears later replacing $C^3$ by $C^2$ and $p$ by $0$. The estimates for general $p$ are also more restrictive, so we are not going to enumerate the estimates for $p=0$ here.\\
We continue with the induction procedure by considering that the estimates are satisfied for $|\partial_u^m\partial_v^n\partial_w^lX|$ for $0\leq l <p$, and try to determine conditions for performing the induction step.\\
We start by formally applying $\partial_u^m\partial_v^n\partial_w^{p-1}$ to the equation $A_{11}=0$ and taking the modulus at the origin. We get
\begin{equation*}
|\partial_u^m\partial_v^n\partial_w^p\phi|\leq |\partial_u^m\partial_v^n\partial_w^{p-1}\phi_{11}|+|\partial_u^m\partial_v^n\partial_w^{p-1}(\hat{e}^1\,_{11}\partial_u\phi)|+|\partial_u^m\partial_v^n\partial_w^{p-1}(\hat{e}^2\,_{11}\partial_v\phi)|.
\end{equation*}
To estimate the terms in the r.h.s. of this inequality we have, using the induction hypothesis,
\begin{equation*}
|\partial_u^m\partial_v^n\partial_w^{p-1}\phi_{11}|\leq c_{\phi_{11}}\frac{r^{m+p-1}(m+p-1)!\rho^nn!}{(m+1)^2(n+1)^2p^2},
\end{equation*}
\begin{eqnarray*}
&& |\partial_u^m\partial_v^n\partial_w^{p-1}(\hat{e}^1\,_{11}\partial_u\phi)|\\
&& \leq\sum_{j=0}^m\sum_{k=0}^n\sum_{l=0}^{p-1}\binom{m}{j}\binom{n}{k}\binom{p-1}{l}|\partial_u^j\partial_v^k\partial_w^l\hat{e}^1\,_{11}||\partial_u^{m-j+1}\partial_v^{n-k}\partial_w^{p-l-1}\phi|\\
&& \leq\sum_{j=0}^m\sum_{k=0}^n\sum_{l=0}^{p-1}\tfrac{\binom{m}{j}\binom{p-1}{l}}{\binom{m+p}{j+l}}\\
&& \times\frac{c_{\hat{e}^1\,_{11}}c_\phi r^{m+p-2}(m+p)!\rho^nn!}{(j+1)^2(k+1)^2(l+1)^2(m-j+2)^2(n-k+1)^2(p-l)^2}\\
&& \leq C^3 c_{\hat{e}^1\,_{11}}c_\phi \frac{r^{m+p-2}(m+p)!\rho^nn!}{(m+2)^2(n+1)^2p^2},
\end{eqnarray*}
and similarly
\begin{equation*}
|\partial_u^m\partial_v^n\partial_w^{p-1}(\hat{e}^2\,_{11}\partial_v\phi)|\leq C^3 c_{\hat{e}^2\,_{11}}c_\phi \frac{r^{m+p-3}(m+p-1)!\rho^n(n+1)!}{(m+1)^2(n+2)^2p^2}.
\end{equation*}
Using these inequalities and writting $|\partial_u^m\partial_v^n\partial_w^p\phi|$ in the form \eqref{generalEstimateForm}, we obtain
\begin{eqnarray*}
A_{\phi}^{p\geq 1} & = & \frac{1}{c_{\phi}}\bigg[\frac{(p+1)^2}{p^2(m+p)}c_{\phi_{11}}+\frac{(p+1)^2(m+1)^2}{p^2(m+2)^2}\frac{C^3}{r}c_{\hat{e}^1\,_{11}}c_{\phi}\\
&& \frac{(p+1)^2(n+1)^3}{p^2(m+p)(n+2)^2}\frac{C^3}{r^2}\rho c_{\hat{e}^2\,_{11}}c_{\phi}\bigg].
\end{eqnarray*}
Taking into account the $v$-finite expansion types of the terms involved, we see that $A_{\phi}^{p\geq 1}=0$ if $n>2m$, and thus
\begin{equation*}
A_{\phi}^{p\geq 1}\leq\frac{4}{c_{\phi}}\left(c_{\phi_{11}}+\frac{C^3}{r}c_{\hat{e}^1\,_{11}}c_{\phi}+2\frac{C^3}{r^2}\rho c_{\hat{e}^2\,_{11}}c_{\phi}\right)=4\frac{c_{\phi_{11}}}{c_{\phi}}+\sum_{i=1}^9 \frac{(\alpha_{\phi}^{p\geq 1})_i}{r^i}.
\end{equation*}
The procedure with the rest of the equations is similar to the one presented for the equation $A_{11}=0$, the only difference being that if an equation is singular with $u^{-1}$ terms we have first to multiply it by $u$, formally apply $\partial_u^{m+1}\partial_v^n\partial_w^{p-1}$, and then estimate the modulus. Therefore we shall not repeat the details that led from the equations to the estimates, as we shall not state the $v$-finite expansion type at each step. What we will state is which equation is used for deriving that particular estimate.\\
Applying formally $\partial_u^m\partial_v^n\partial_w^{p-1}$ to the equation $D_{11}\phi_{00}=D_{00}\phi_{11}$, which follows from $A_{00}=0$ and $A_{11}=0$, we obtain
\begin{equation*}
A_{\phi_{00}}^{p\geq 1}\leq 4\frac{c_{\phi_{11}}}{c_{\phi_{00}}}+\sum_{i=1}^9 \frac{(\alpha_{\phi_{00}}^{p\geq 1})_i}{r^i}.
\end{equation*}
Multiplying $H_{1(ABC)_{0}}+H_{0(ABC)_{1}}=0$ by $u$ and formaly aplying $\partial_u^{m+1}\partial_v^n\partial_w^{p-1}$ we get
\begin{eqnarray*}
A_{S_{0}}^{p\geq 1} & \leq & \frac{4}{c_{S_{0}}}\left[c_{S_2}+\tfrac{16}{3}C^3\left(c_{\phi_{00}}c_{\phi_{11}}+c_{\phi_{01}}c_{\phi_{01}}\right)\right]+\sum_{i=1}^9 \frac{(\alpha_{S_{0}}^{p\geq 1})_i}{r^i}.
\end{eqnarray*}
In the same way as we used \eqref{phiA1enU0}, \eqref{SkenU0} to obtain \eqref{AphiA1}, \eqref{ASk} we get
\begin{equation*}
A_{\phi_{A1}}^{m=0,p\geq 1}\leq \frac{1}{2}\frac{c_{\phi_{00}}}{c_{\phi_{A1}}}\rho^{1+A},
\end{equation*}
\begin{equation*}
A_{S_{k}}^{m=0,p\geq 1}\leq \frac{c_{S_{0}}}{c_{S_{k}}}\rho^k.
\end{equation*}
Restricting $\Sigma_{11}=0$ and $\Sigma_{11CD}=0$ to $U_{0}$ we find that on $U_0$
\begin{eqnarray*}
& \Omega=0,\,\,\Omega_{01}=0,\,\,\Omega_{11}=0,\\
& \partial_{w}\Omega_{00}=-\frac{1}{3}R+\frac{8}{3(1-\phi^2)}(\phi\Omega_{00}\phi_{11}-2\phi_{00}\phi_{11}+2\phi^2_{01}).
\end{eqnarray*}
Taking formal derivatives of these equations we get
\begin{eqnarray*}
A_{\Omega}^{m=0}=0, & A_{\Omega_{01}}^{m=0}=0, & A_{\Omega_{11}}^{m=0}=0,
\end{eqnarray*}
and
\begin{equation*}
A_{\Omega_{00}}^{m=0,p\geq 1} \leq \frac{4}{3}\frac{1}{c_{\Omega_{00}}}\left[c_{R}+32C^2\left(c_{\phi_{00}}c_{\phi_{11}}+c_{\phi_{01}}^2\right)\right]+\sum_{i=1}^9 \frac{(\alpha_{\Omega_{00}}^{m=0,p\geq 1})_i}{r^i}.
\end{equation*}
Restricting $\Pi_{11}=0$ to $U_{0}$ gives
\begin{eqnarray*}
\partial_{w}R = 3\Omega_{00}S_{4}+\frac{8}{1-\phi^2}\bigg[\frac{}{}-2\phi_{11}\partial_{w}\phi_{00}+4\phi_{01}\partial_{w}\phi_{01}+(\phi\Omega_{00}-2\phi_{00})\partial_w\phi_{11}\\
 -\frac{1}{3}\phi R \phi_{11}\bigg]+\frac{8}{3(1-\phi^2)^2}\phi_{11}\left[(3+11\phi^2)\Omega_{00}\phi_{11}-28\phi(\phi_{00}\phi_{11}-\phi^2_{01})\right],
\end{eqnarray*}
so that
\begin{equation*}
A_{R}^{m=0,p\geq 1} \leq \frac{64C^2}{c_{R}}(c_{\phi_{00}}c_{\phi_{11}}+c_{\phi_{01}}^2)+\sum_{i=1}^9 \frac{(\alpha_{R}^{m=0,p\geq 1})_i}{r^i}.
\end{equation*}
We complete the calculation of the $A$'s by using the $\partial_{u}$-equations. We have to calculate the estimates in the order given by the hierarchy presented in Section \ref{hierarchy} but for simplicity we present the estimates in the order the $\partial_u$-equations were stated in Section \ref{dUequations}.\\
\\
$t_{AB}\,^{EF}\,_{00}e^a\,_{EF}=0$:
\begin{eqnarray*}
A^{m\geq 1}_{\hat{e}^1\,_{01}} & \leq & \sum_{i=1}^9 \frac{(\alpha_{\hat{e}^1\,_{01}}^{m\geq 1})_i}{r^i}, \\
A^{m\geq 1}_{\hat{e}^2\,_{01}} & \leq & \frac{1}{2}\frac{c_{\hat{\Gamma}_{0100}}}{c_{\hat{e}^2\,_{01}}}+\sum_{i=1}^9 \frac{(\alpha_{\hat{e}^2\,_{01}}^{m\geq 1})_i}{r^i}, \\
A^{m\geq 1}_{\hat{e}^1\,_{11}} & \leq & \sum_{i=1}^9 \frac{(\alpha_{\hat{e}^1\,_{11}}^{m\geq 1})_i}{r^i}, \\
A^{m\geq 1}_{\hat{e}^2\,_{11}} & \leq & \frac{c_{\hat{\Gamma}_{1100}}}{c_{\hat{e}^2\,_{11}}}+\sum_{i=1}^9 \frac{(\alpha_{\hat{e}^2\,_{11}}^{m\geq 1})_i}{r^i}.
\end{eqnarray*}
$R_{AB00EF}=0$:
\begin{eqnarray*}
A^{m\geq 1}_{\hat{\Gamma}_{0100}} & \leq & \frac{2}{3}\frac{c_{S_{0}}}{c_{\hat{\Gamma}_{0100}}}+\sum_{i=1}^9 \frac{(\alpha_{\hat{\Gamma}_{0100}}^{m\geq 1})_i}{r^i}, \\
A^{m\geq 1}_{\hat{\Gamma}_{0101}} & \leq & \frac{c_{S_{1}}}{c_{\hat{\Gamma}_{0101}}}+\sum_{i=1}^9 \frac{(\alpha_{\hat{\Gamma}_{0101}}^{m\geq 1})_i}{r^i}, \\
A^{m\geq 1}_{\hat{\Gamma}_{0111}} & \leq & \frac{c_{S_{2}}}{c_{\hat{\Gamma}_{0111}}}+\frac{c_{R}}{6c_{\hat{\Gamma}_{0111}}}+\sum_{i=1}^9 \frac{(\alpha_{\hat{\Gamma}_{0111}}^{m\geq 1})_i}{r^i}, \\
A^{m\geq 1}_{\hat{\Gamma}_{1100}} & \leq & 2\frac{c_{S_{1}}}{c_{\hat{\Gamma}_{1100}}}+\sum_{i=1}^9 \frac{(\alpha_{\hat{\Gamma}_{1100}}^{m\geq 1})_i}{r^i}, \\
A^{m\geq 1}_{\hat{\Gamma}_{1101}} & \leq & \frac{4}{c_{\hat{\Gamma}_{1101}}}\left(c_{S_{2}}+\frac{1}{12}c_{R}\right)+\sum_{i=1}^9 \frac{(\alpha_{\hat{\Gamma}_{1101}}^{m\geq 1})_i}{r^i}, \\
A^{m\geq 1}_{\hat{\Gamma}_{1111}} & \leq & \frac{4c_{S_{3}}}{c_{\hat{\Gamma}_{1111}}}+\sum_{i=1}^9 \frac{(\alpha_{\hat{\Gamma}_{1111}}^{m\geq 1})_i}{r^i}.
\end{eqnarray*}
$\Sigma_{00}=0$:
\begin{equation*}
A^{m\geq 1}_{\Omega} \leq \sum_{i=1}^9 \frac{(\alpha_{\Omega}^{m\geq 1})_i}{r^i}.
\end{equation*}
$\Phi_{00}=0$:
\begin{equation*}
A^{m\geq 1}_{\phi_{01}} \leq \frac{c_{\phi_{00}}}{c_{\phi_{01}}}\rho+\sum_{i=1}^9 \frac{(\alpha_{\phi_{01}}^{m\geq 1})_i}{r^i}.
\end{equation*}
$\Phi_{10}=0$:
\begin{equation*}
A^{m \geq 1}_{\phi_{11}} \leq \frac{c_{\phi_{01}}}{c_{\phi_{11}}}\rho+\sum_{i=1}^9 \frac{(\alpha_{\phi_{11}}^{m\geq 1})_i}{r^i}.
\end{equation*}
$\Sigma_{00CD}=0$:
\begin{eqnarray*}
A^{m \geq 1}_{\Omega_{00}} & \leq & \sum_{i=1}^9 \frac{(\alpha_{\Omega_{00}}^{m\geq 1})_i}{r^i}, \\
A^{m \geq 1}_{\Omega_{01}} & \leq & \sum_{i=1}^9 \frac{(\alpha_{\Omega_{01}}^{m\geq 1})_i}{r^i}, \\
A^{m \geq 1}_{\Omega_{11}} & \leq & \frac{4}{3c_{\Omega_{11}}}\left[c_{R}+32C^3\left(c_{\phi_{00}}c_{\phi_{11}}+c_{\phi_{01}}^2\right)\right]+\sum_{i=1}^9 \frac{(\alpha_{\Omega_{11}}^{m\geq 1})_i}{r^i}.
\end{eqnarray*}
$\Pi_{00}=0$:
\begin{equation*}
A^{m\geq 1}_{R} \leq \frac{64C^3}{c_{R}}(c_{\phi_{00}}c_{\phi_{11}}+c_{\phi_{01}}^2)+\sum_{i=1}^9 \frac{(\alpha_{R}^{m\geq 1})_i}{r^i}.
\end{equation*}
$H_{0(ABC)_{k}}=0$:
\begin{eqnarray*}
&& A^{m\geq 1}_{S_{1}} \leq \frac{c_{S_{0}}}{c_{S_{1}}}\rho+\sum_{i=1}^9 \frac{(\alpha_{S_1}^{m\geq 1})_i}{r^i}, \\
&& A^{m\geq 1}_{S_{2}} \leq \frac{1}{c_{S_2}}\left[\rho c_{S_{1}}+\tfrac{8}{3}C^3\left(c_{\phi_{00}}c_{\phi_{11}}+c_{\phi_{01}}c_{\phi_{01}}\right)\right]+\sum_{i=1}^9 \frac{(\alpha_{S_2}^{m\geq 1})_i}{r^i}, \\
&& A^{m\geq 1}_{S_{3}} \leq \frac{1}{c_{S_{3}}}\left\{\rho c_{S_{2}}+\tfrac{8}{3}C^3\left[\rho\left(c_{\phi_{00}}c_{\phi_{11}}+2c_{\phi_{01}}^2\right)+\tfrac{3}{2}c_{\phi_{01}}c_{\phi_{11}}
\right]\right\}+\sum_{i=1}^9 \frac{(\alpha_{S_3}^{m\geq 1})_i}{r^i}, \\
&& A^{m\geq 1}_{S_{4}} \leq \frac{1}{c_{S_{4}}}\left[\rho c_{S_{3}}+4C^3c_{\phi_{11}}\left(2\rho c_{\phi_{01}}+\tfrac{8}{3}c_{\phi_{00}}+c_{\phi_{11}}\right)\right]+\sum_{i=1}^9 \frac{(\alpha_{S_4}^{m\geq 1})_i}{r^i}.
\end{eqnarray*}
We now have to show that all the constants can be chosen in a way that makes all the $A$'s less or equal than $1$. So, introducing a constant $a$, $0<a<1$, the following inequalities need to be satisfied:
\begin{eqnarray}
&& \label{in1}\frac{1}{2}\frac{c_{\phi_{00}}}{c_{\phi_{01}}}\rho\leq 1,\\
&& \label{in2}\frac{1}{2}\frac{c_{\phi_{00}}}{c_{\phi_{11}}}\rho^2\leq 1,\\
&& \label{in3}\frac{c_{S_{0}}}{c_{S_{1}}}\rho\leq 1,\,\,\,\,\, \frac{c_{S_{0}}}{c_{S_{2}}}\rho^2\leq 1,\,\,\,\,\, \frac{c_{S_{0}}}{c_{S_{3}}}\rho^3\leq 1,\,\,\,\,\, \frac{c_{S_{0}}}{c_{S_{4}}}\rho^4\leq 1,\\
&& \label{in4}\frac{1}{c_R}\left(6+\tfrac{73}{36}\rho^2 c_{\phi_{00}}^2\right)\leq 1,\\
&& \label{in5}4\frac{c_{\phi_{11}}}{c_\phi}\leq a,\,\,\,\,\, 4\frac{c_{\phi_{11}}}{c_{\phi_{00}}}\leq a,\,\,\,\,\, \frac{4}{c_{S_{0}}}\left[c_{S_{2}}+\tfrac{16}{3}C^3\left(c_{\phi_{00}}c_{\phi_{11}}+c_{\phi_{01}}^2\right)\right]\leq a,\\
&& \label{in6}\frac{4}{3}\frac{1}{c_{\Omega_{00}}}\left[c_R+32C^2\left(c_{\phi_{00}}c_{\phi_{11}}+c_{\phi_{01}}^2\right)\right]\leq a,\\
&& \label{in7}\frac{1}{c_R}64C^2\left(c_{\phi_{00}}c_{\phi_{11}}+c_{\phi_{01}}^2\right)\leq a,\\
&& \label{in8}\frac{1}{2}\frac{c_{\hat{\Gamma}_{0100}}}{c_{\hat{e}^2\,_{01}}}\leq a,\,\,\,\,\, \frac{c_{\hat{\Gamma}_{1100}}}{c_{\hat{e}^2\,_{11}}}\leq a,\,\,\,\,\, \frac{2}{3}\frac{c_{S_0}}{c_{\hat{\Gamma}_{0100}}}\leq a,\,\,\,\,\, \frac{c_{S_1}}{c_{\hat{\Gamma}_{0101}}}\leq a,
\end{eqnarray}
\begin{eqnarray}
&& \label{in9}\frac{1}{c_{\hat{\Gamma}_{0111}}}\left(c_{S_2}+\tfrac{1}{6}c_R\right)\leq a,\,\,\,\,\, 2\frac{c_{S_1}}{c_{\hat{\Gamma}_{1100}}}\leq a,\\
&& \label{in10}\frac{4}{c_{\hat{\Gamma}_{1101}}}\left(c_{S_2}+\tfrac{1}{12}c_R\right)\leq a,\,\,\,\,\, 4\frac{c_{S_3}}{c_{\hat{\Gamma}_{1111}}}\leq a,\\
&& \label{in11}\frac{c_{\phi_{00}}}{{c_{\phi_{01}}}}\rho\leq a,\\
&& \label{in12}\frac{c_{\phi_{01}}}{{c_{\phi_{11}}}}\rho\leq a,\\
&& \label{in13}\frac{4}{3}\frac{1}{c_{\Omega_{11}}}\left[c_R+32C^3\left(c_{\phi_{00}}c_{\phi_{11}}+c_{\phi_{01}}^2\right)\right]\leq a,\\
&& \label{in14}\frac{1}{c_R}64C^3\left(c_{\phi_{00}}c_{\phi_{11}}+c_{\phi_{01}}^2\right)\leq a,\\
&& \label{in15}\frac{c_{S_0}}{c_{S_1}}\rho\leq a,\,\,\,\,\, \frac{1}{c_{S_2}}\left[c_{S_1}\rho+\tfrac{8}{3}C^3\left(c_{\phi_{00}}c_{\phi_{11}}+c_{\phi_{01}}^2\right)\right]\leq a,\\
&& \label{in16}\frac{1}{c_{S_3}}\left[c_{S_2}\rho+\tfrac{4}{3}C^3\left(2\rho c_{\phi_{00}}c_{\phi_{11}}+4\rho c_{\phi_{01}}^2+3c_{\phi_{01}}c_{\phi_{11}}\right)\right]\leq a,\\
&& \label{in17}\frac{1}{c_{S_4}}\left[c_{S_3}\rho+\tfrac{4}{3}C^3c_{\phi_{11}}\left(8c_{\phi_{00}}+6\rho c_{\phi_{01}}+3c_{\phi_{11}}\right)\right]\leq a.
\end{eqnarray}
Now we have to show that we can choose the constants such that these inequalities will be satisfied.\\
We start by setting
\begin{equation*}
c_{\phi_{01}}\equiv\frac{\rho}{a}c_{\phi_{00}},
\end{equation*}
with which we satisfy \eqref{in1} and \eqref{in11}. Next we set
\begin{equation*}
c_{\phi_{11}}\equiv\frac{\rho^2}{a^2}c_{\phi_{00}},
\end{equation*}
so that \eqref{in2} and \eqref{in12} are satisfied.\\
We continue by setting
\begin{eqnarray*}
&& c_{S_1}\equiv\frac{\rho}{a}c_{S_0},\,\,\,\,\,c_{S_2}\equiv\frac{\rho^2}{a^2}\left(c_{S_0}+\frac{16}{3}\frac{C^3}{a}c_{\phi_{00}}^2\right),\\
&& c_{S_3}\equiv\frac{\rho^3}{a^3}\left[c_{S_0}+\frac{8}{3}\left(3+\frac{7}{2a}\right)C^3c_{\phi_{00}}^2\right],\\
&& c_{S_4}\equiv\frac{\rho^4}{a^4}\left[c_{S_0}+\frac{8}{3}\left(6+\frac{5}{a}+4\frac{a}{\rho^2}\right)C^3c_{\phi_{00}}^2\right].
\end{eqnarray*}
With this we satisfy \eqref{in3}, \eqref{in15}, \eqref{in16} and \eqref{in17}.\\
Inequalities \eqref{in4}, \eqref{in7} and \eqref{in14} are satisfied with
\begin{equation*}
c_R\equiv max\left\{128\frac{\rho^2}{a^3}C^3c_{\phi_{00}}^2,6+\frac{73}{36}\rho^2c_{\phi_{00}}^2\right\}.
\end{equation*}
With this definition for $c_R$ we set
\begin{equation*}
c_{\Omega_{00}}\equiv\frac{4}{3a}\left(c_R+64C^2\frac{\rho^2}{a^2}c_{\phi_{00}}^2\right),\,\,\,\,\,c_{\Omega_{11}}\equiv\frac{4}{3a}\left(c_R+64C^3\frac{\rho^2}{a^2}c_{\phi_{00}}^2\right),
\end{equation*}
so \eqref{in6} and \eqref{in13} are respectively satisfied.\\
Using the previous definitions we set also
\begin{eqnarray*}
&& c_{\hat{\Gamma}_{0100}}\equiv\frac{2}{3a}c_{S_0},\,\,\,\,\,c_{\hat{\Gamma}_{0101}}\equiv\frac{1}{a}c_{S_1},\\
&& c_{\hat{\Gamma}_{0111}}\equiv\frac{1}{a}\left(c_{S_2}+\frac{1}{6}c_R\right),\,\,\,\,\,c_{\hat{\Gamma}_{1100}}\equiv\frac{2}{a}c_{S_1},\\
&& c_{\hat{\Gamma}_{1101}}\equiv\frac{4}{a}\left(c_{S_2}+\frac{1}{12}c_R\right),\,\,\,\,\,c_{\hat{\Gamma}_{1111}}\equiv\frac{4}{a}c_{S_3},\\
&& c_{\hat{e}^2\,_{01}}\equiv\frac{1}{3a^2}c_{S_0},\,\,\,\,\,c_{\hat{e}^2\,_{11}}\equiv\frac{2}{a^2}c_{S_1},
\end{eqnarray*}
and \eqref{in8}, \eqref{in9} and \eqref{in10} are satisfied.\\
There are three inequalities that we have not yet considered, \eqref{in5}. These are now reduced to
\begin{eqnarray*}
&& 4\rho^2\frac{c_{\phi_{00}}}{c_\phi}\leq a^3,\\
&& 4\rho^2\leq a^3,\\
&& 4\rho^2\left[1+\frac{1}{c_{S_0}}\frac{16}{3}C^3c_{\phi_{00}}^2\left(2+\frac{1}{a}\right)\right]\leq a^3.
\end{eqnarray*}
Taking into consideration now \eqref{conditionsC1}, \eqref{conditionsC2}, \eqref{conditionsRrho} we see that these inequalities can be satisfied if we define
\begin{eqnarray*}
&& \rho\equiv max\{\rho_\phi,\rho_{S_0},\rho_{\phi_{00}}\}<\frac{1}{3},\\
&& a\equiv max\left\{\frac{1}{2},(8\rho^2)^\frac{1}{3}\right\}<1.
\end{eqnarray*}
Now we choose some positive constants $c_\Omega,c_{\Omega_{01}},c_{\hat{e}^1\,_{01}},c_{\hat{e}^1\,_{11}}$, that are not restricted by the procedure.\\
Finally we choose r so large that
\begin{equation*}
r>max\left\{r_\phi,r_{S_{0}},r_{\phi_{00}},C^3\left[c_\Omega+\left(c_\Omega^2+2c_\phi^2\right)^\frac{1}{2}\right]\right\}
\end{equation*}
and that all the $A$'s are less or equal than $1$. The induction proof is completed.
\end{proof}
The following lemma states the convergent result. The proof follows as the one given in \cite{Friedrich07}.
\begin{lem}\label{holomorphicSolutions}
The estimates \eqref{convergenceEstimates} for the derivatives of the functions $f$ and the expansion types given in Lemma \ref{expansionTypeFields} imply that the associated Taylor series are absolutely convergent in the domain $|v|<\tfrac{1}{\alpha\rho}$, $|u|+|w|<\tfrac{\alpha^2}{r}$, for any real number $\alpha$, $0<\alpha\leq 1$. It follows that the formal expansion determined in Lemma \ref{expansionCoefficients} defines indeed a (unique) holomorphic solution to the conformal static vacuum field equations which induces the data $\phi$, $S_0$ on $W_0$.
\end{lem}

\section{The complete set of equations on $\hat{N}$}\label{completeSet}
We have seen in Section \ref{sectionConformalEquations} how to calculate a formal expansion for our fields using a subset of the conformal stationary vacuum field equations. In the previous section we have shown that these formal expansions are convergent in a neighbourhood of infinity. In this section we shall show that these fields satisfy the complete system of conformal stationary vacuum field equations. First, we prove that the conformal stationary vacuum field equations are satisfied in the limit as $u\rightarrow 0$. Second, we derive a subsidiary system of equations, for which the first result provides the initial conditions, and which allows us to prove that the complete system is satisfied.
\begin{lem}\label{limitCompleteSystem}
The functions $\hat{e}^a\,_{AB}$, $\hat{\Gamma}_{ABCD}$, $\phi$, $\Omega$, $R$, $S_{ABCD}$, whose expansion coefficients are determined by Lemma \ref{expansionCoefficients}, with expansions that converge on an open neighbourhood of the point $0$, neighbourhood that we assume to coincide with $\hat{N}$, satisfy the complete set of conformal field equations on the set $U_0$ in the sense that the fields $t_{AB}\,^{CD}\,_{EF}$, $R_{ABCDEF}$, $A_{AB}$, $\Sigma_{AB}$, $\Phi_{AB}$, $\Pi_{AB}$, $\Sigma_{ABCD}$, $H_{ABCD}$ calculated from these functions on $\hat{S}\backslash U_0$ have vanishing limit as $u\rightarrow 0$.
\end{lem}
\begin{proof}
Taking into account which equations have already been used to determine the formal expansions, and the symmetries of the equations, it is left to show that
$t_{01}\,^{EF}\,_{11},R_{AB0111},A_{01},\Sigma_{01},\Pi_{01},\Sigma_{01CD},H_{1(BCD)_{k=1,2,3}},$
have vanishing limit on $\hat{N}\backslash U_0$ as $u\rightarrow 0$, and that in the same limit $\Phi_{AB}=-\Phi_{BA}$.\\
Because $\langle\sigma^{AB},e_{EF}\rangle=h^{AB}\,_{EF}$ then
\begin{equation*}
t_{01}\,^{EF}\,_{11}=2\Gamma_{01}\,^{(E}\,_{1}\epsilon_{1}\,^{F)}-2\Gamma_{11}\,^{(E}\,_{(0}\epsilon_{1)}\,^{F)}-\sigma^{EF}\,_{a}\left(e^a\,_{11,b}e^b\,_{01}-e^a\,_{01,b}e^b\,_{11}\right),
\end{equation*}
and using the way in which the coordinates and the frame field were constructed, we see that
\begin{equation*}
t_{01}\,^{EF}\,_{11}=O(u),\,\,\,\mbox{as}\,\,u\rightarrow 0.
\end{equation*}
We now consider
\begin{eqnarray*}
R_{AB0111} & = & -\frac{1}{2}\left(S_{AB11}-\frac{1}{6}R\epsilon_{A1}\epsilon_{B1}\right)+\frac{1}{2u}\partial_{v}\hat{\Gamma}_{11AB}-\frac{1}{u}\hat{\Gamma}_{111(A}\epsilon_{B)}\,^{0}\\
&& +\epsilon_{A}\,^{0}\epsilon_{B}\,^{0}\left(-\frac{1}{2u^2}\hat{e}^1\,_{11}+\frac{1}{u}\hat{\Gamma}_{0111}-\frac{1}{u}t_{01}\,^{01}\,_{11}\right)+O(u).
\end{eqnarray*}
Using that
\begin{equation*}
t_{01}\,^{01}\,_{11}=\hat{\Gamma}_{0111}-\frac{1}{2}\partial_v\hat{e}^2\,_{11}-\frac{1}{2u}\hat{e}^1\,_{11}+O(u^2)
\end{equation*}
we get
\begin{eqnarray*}
R_{AB0111} & = & \frac{1}{2u}\bigg[\partial_v\hat{\Gamma}_{11AB}-2\hat{\Gamma}_{111(A}\epsilon_{B)}\,^0+\\
&& \epsilon_A\,^0\epsilon_B\,^0\left(-\frac{2}{u}\hat{e}^1\,_{11}-\partial_v\hat{e}^2_{11}+4\hat{\Gamma}_{0111}\right)\bigg]\\
&& -\frac{1}{2}\left(S_{AB11}-\frac{1}{6}R\epsilon_{A1}\epsilon_{B1}\right)+O(u),
\end{eqnarray*}
so that
\begin{eqnarray*}
&& \lim_{u\rightarrow 0}R_{AB0111} = \frac{1}{2}\bigg[\partial_u\partial_v\hat{\Gamma}_{11AB}-2\partial_u\hat{\Gamma}_{111(A}\epsilon_{B)}\,^0\\
&& +\epsilon_A\,^0\epsilon_B\,^0\left(-\partial_u^2\hat{e}^1\,_{11}-\partial_u\partial_v\hat{e}^2\,_{11}+4\partial_u\hat{\Gamma}_{0111}\right)-S_{AB11}+\frac{1}{6}R\epsilon_{A1}\epsilon_{B1}\bigg]\bigg|_{u=0}.
\end{eqnarray*}
For the case $A=B=0$ we get from the $\partial_u$-equations that
\begin{equation*}
\partial_u^2\hat{e}^1\,_{11}=-2\partial_u\hat{\Gamma}_{1101},\,\,\,\partial_u\partial_v\hat{e}^2\,_{11}=\partial_u\partial_v\hat{\Gamma}_{1100},\,\,\,\partial_u\hat{\Gamma}_{0111}=\frac{1}{4}\left(S_2-\frac{1}{6}R\right),
\end{equation*}
on $U_0$, and so $\lim_{u\rightarrow 0}R_{000111}=0$.\\
Using the $\partial_u$-equations and that $\partial_vS_2=2S_3$ on $U_0$,
\begin{equation*}
\partial_u\hat{\Gamma}_{1111}=S_3,\,\,\,\partial_u\partial_v\hat{\Gamma}_{1101}=2S_3,
\end{equation*}
on $U_0$, and so $\lim_{u\rightarrow 0}R_{010111}=0$. As $\partial_vS_3=S_4$ on $U_0$, $\lim_{u\rightarrow 0}R_{110111}=0$.\\
We take now the limit of $A_{01}$ as $u$ goes to $0$,
\begin{equation*}
\lim_{u\rightarrow 0}A_{01}=\left.\left[\frac{1}{2}\partial_{u}\partial_{v}\phi-\phi_{01}\right]\right|_{u=0}.
\end{equation*}
Using that $\phi_{01}=\frac{1}{2}\partial_v\phi_{00}$ on $U_0$ and that we have $A_{00}=0$ as part of the $\partial_u$-equations we get $\lim_{u\rightarrow 0}A_{01}=0$. With the same procedure we get $\lim_{u\rightarrow 0}\Sigma_{01}=0$.\\
Now we consider $\Phi_{AB}$. As $A_{AB}=0$ on $U_0$ then
\begin{equation*}
D^P\,_B\phi_{AP}|_{U_0}=-D^P\,_A\phi_{BP}|_{U_0},
\end{equation*}
so $\Phi_{AB}|_{U_0}=-\Phi_{BA}|_{U_0}$ and as we already have $\Phi_{10}=0$ then $\Phi_{AB}=0$ on $U_0$.\\
We now take the limit as $u$ goes to $0$ of the combination $\Pi_{01}-\frac{1}{2}\partial_v\Pi_{00}$. For the limits of the derivatives involved we have at $\{u=0\}$
\begin{eqnarray*}
&& D_{00}\phi_{AB}=\partial_u\phi_{AB},\\
&& D_{01}\phi_{01}=\frac{1}{2}\left(\partial_u\partial_v\phi_{01}-\partial_u\phi_{11}\right),\\
&& D_{01}\phi_{11}=\frac{1}{2}\partial_u\partial_v\phi_{11},\\
&& D_{11}\phi_{11}=\partial_w\phi_{11},\\
&& D_{00}R=\partial_u R,\\
&& D_{01}R=\frac{1}{2}\partial_u\partial_v R,\\
&& D_{11}R=\partial_wR.
\end{eqnarray*}
We also use that on $U_0$
\begin{eqnarray*}
\partial_v\phi_k=(2-k)\phi_{k+1},\\
\partial_v S_k=(4-k)S_{k+1},
\end{eqnarray*}
We have already used the equations $\Sigma_{11}=0,\Sigma_{11CD}=0$ restricted to $U_0$, finding that $\Omega,\Omega_{A1}$ are zero on $U_0$.\\
Furthermore we use $\Phi_{00}=0$, that says that on $U_0$
\begin{equation*}
\partial_u\partial_v\phi_{00}=4\partial_u\phi_{01}.
\end{equation*}
So we get for the limit
\begin{equation*}
\lim_{u\rightarrow 0}\left(\Pi_{01}-\frac{1}{2}\partial_v\Pi_{00}\right)=0.
\end{equation*}
Considering that from the $\partial_u$-equations we already have $\Pi_{00}=0$ we get
\begin{equation*}
\lim_{u\rightarrow 0}\Pi_{01}=0.
\end{equation*}
We apply a similar procedure to the $\Sigma_{01AB}$ equations. We take the limit as $u$ goes to zero of the combinations
\begin{eqnarray*}
& \Sigma_{0100}-\Sigma_{0001},\\
& 2\Sigma_{0101}-\partial_v\Sigma_{0001}+\Sigma_{0011},\\
& 2\Sigma_{0111}-\partial_v\Sigma_{0011}.
\end{eqnarray*}
Using what has already been said together with the following limits at $\{u=0\}$
\begin{eqnarray*}
&& D_{00}\Omega_{AB}=\partial_u\Omega_{AB},\\
&& D_{01}\Omega_{01}=\frac{1}{2}\left(\partial_u\partial_v\Omega_{01}-\partial_u\Omega_{11}\right),\\
&& D_{01}\Omega_{11}=\frac{1}{2}\partial_u\partial_v\Omega_{11},\\
&& D_{11}\Omega_{11}=\partial_w\Omega_{11},
\end{eqnarray*}
we see that the limits vanishes, which imply
\begin{equation*}
\lim_{u\rightarrow 0}\Sigma_{01AB}=0.
\end{equation*}
Finally we consider the limit as $u$ goes to zero of the combinations
\begin{eqnarray*}
& 4H_{1(ABC)_1}-\partial_vH_{1(ABC)_0}-\partial_vH_{0(ABC)_1}+2H_{0(ABC)_2},\\
& 12H_{1(ABC)_2}-\partial_v^2H_{1(ABC)_0}-\partial_v^2H_{0(ABC)_1}-2\partial_vH_{0(ABC)_2}+4H_{0(ABC)_3},\\
& 24H_{1(ABC)_3}-\partial_v^3H_{1(ABC)_0}-\partial_v^3H_{0(ABC)_1}-2\partial_v^2H_{0(ABC)_2}-8\partial_vH_{0(ABC)_3},
\end{eqnarray*}
and using what has been said together with:\\
the limits
\begin{eqnarray*}
&& D_{00}S_k=\partial_uS_k,\\
&& D_{01}S_k=\frac{1}{2}\left[\partial_u\partial_vS_k-(4-k)\partial_uS_{k+1}\right],\\
&& D_{11}S_k=\partial_wS_k,
\end{eqnarray*}
the equality on $U_0$
\begin{equation*}
\partial_vS_k=(4-k)S_{k+1},
\end{equation*}
and the equations
\begin{eqnarray*}
\Phi_{A0}=0,
\end{eqnarray*}
we find that those limits are all zero, and considering the equations that we have used to calculate the unknowns we get
\begin{equation*}
\lim_{u\rightarrow 0}H_{1(ABC)_k}=0,\,\,\,k=1,2,3.
\end{equation*}
This completes the proof that the complete system of conformal field equations are satisfied in the limit as $u\rightarrow 0$.
\end{proof}
\begin{lem}
The functions $\hat{e}^a\,_{AB},\hat{\Gamma}_{ABCD},\phi,\Omega,R,S_{ABCD}$ corresponding to the expansions determined in Lemma \ref{expansionCoefficients} satisfy the complete set of conformal vacuum field equations on the set $\hat{N}$.
\end{lem}
\begin{proof}
We have to show that on $\hat{N}$ the quantities $t_{01}\,^{EF}\,_{11}$ ,$R_{AB0111}$, $A_{01}$, $\Sigma_{A1}$, $\Pi_{A1}$ , $\Sigma_{A1CD}$, $H_{1(BCD)_{k=1,2,3}}$ vanish, and that $\Phi_{AB}=-\Phi_{BA}$. For this we derive a system of subsidiary equations for these fields. The values of the fields at $U_0$, given by Lemma \ref{limitCompleteSystem}, are the initial conditions for the subsidiary system of equations, and they are used throughout the proof.\\
Using the definitions of $A_{AB}$ and $\Phi_{AB}$:
\begin{equation*}
D_{AB}A_{CD}-D_{CD}A_{AB}=-t_{AB}\,^{EF}\,_{CD}D_{EF}\phi+\epsilon_{AD}\Phi_{BC}+\epsilon_{BC}\Phi_{DA},
\end{equation*}
and in particular
\begin{equation*}
\left(\partial_u +\frac{1}{u}\right)A_{01}=2\hat{\Gamma}_{0100}A_{01},
\end{equation*}
which implies $A_{01}=0$, and from that $A_{AB}=0$. This also shows that $\Phi_{AB}=-\Phi_{AB}$, and as we already know that $\Phi_{10}=0$ then $\Phi_{AB}=0$.\\
Following the proof of Lemma 5.5 in \cite{Friedrich07} we find that
\begin{equation}\label{subsidiary1}
\left(\partial_u +\frac{1}{u}\right)t_{01}\,^{AB}\,_{11}=2\hat{\Gamma}_{0100}t_{01}\,^{AB}\,_{11}+2R^{(A}\,_{00111}\epsilon_0\,^{B)},
\end{equation}
which directly shows that $t_{01}\,^{11}\,_{11}=0$.\\
Also following the proof of Lemma 5.5 in \cite{Friedrich07} and taking into account that $S_{ABCD}$ and $R$ satisfy the the contracted Bianchi identity then
\begin{equation}\label{subsidiary2}
\left(\partial_u +\frac{1}{u}\right)R_{AB0111}=2\hat{\Gamma}_{0100}R_{AB0111}-\frac{1}{2}\left(H_{1AB0}-\frac{1}{6}\Pi_{AB}\right),
\end{equation}
from which $R_{000111}=0$, wich also gives $t_{01}\,^{01}\,_{11}=0$.\\
It is still left to show that
\begin{equation}\label{zeroRest}
t_{01}\,^{00}\,_{11},\,\,\,R_{A10111},\,\,\,\Sigma_{A1},\,\,\,\Pi_{A1},\,\,\,\Sigma_{A1CD},\,\,\,H_{1(BCD)_{k=1,2,3}}
\end{equation}
vanish on $\hat{N}$.\\
Using the definitions of $\Sigma_{AB}$ and $\Sigma_{ABCD}$,
\begin{equation*}
D_{AB}\Sigma_{CD}-D_{CD}\Sigma_{AB}=-t_{AB}\,^{EF}\,_{CD}D_{EF}\Omega-\Sigma_{ABCD}+\Sigma_{CDAB},
\end{equation*}
and from that
\begin{equation}\label{subsidiary3}
\partial_u\Sigma_{A1} +\frac{1}{u}\Sigma_{01}\epsilon_A\,^0=2\hat{\Gamma}_{A100}\Sigma_{01}+\Sigma_{A100}.
\end{equation}
At this point the expressions became to long to be treated by hand, so we resort to a computer progam for tensor manipulations.\\
For $\Sigma_{ABCD}$ we obtain
\begin{eqnarray*}
D_{EF}\Sigma_{CDAB}-D_{CD}\Sigma_{EFAB}=t_{CD}\,^{PQ}\,_{EF}D_{PQ}\Omega_{AB}-2\Omega^P\,_{(A}R_{B)PCDEF}\\
+\Omega\left(\epsilon_{DE}H_{FABC}+\epsilon_{CF}H_{DABE}\right)+S_{ABCD}\Sigma_{EF}-S_{ABEF}\Sigma_{CD}\\
+\frac{R}{3}\left(h_{ABCD}\Sigma_{EF}-h_{ABEF}\Sigma_{CD}\right)+\frac{1}{3}\left(1+\Omega\right)\left(h_{ABCD}\Pi_{EF}-h_{ABEF}\Pi_{CD}\right)\\
+\frac{1}{6\left(1+\Omega-\phi^2\right)}\bigg\{3\left[1+\left(1+\Omega\right)\phi^2\right]\big[\Sigma_{ABCD}\Omega_{EF}-\Sigma_{ABEF}\Omega_{CD}\\
-\Omega_{AB}\left(\Sigma_{CDEF}-\Sigma_{EFCD}\right)\big]\\
-4\left(2+3\Omega\right)\phi^2\left(h_{ABCD}\Sigma_{EFPQ}\Omega^{PQ}-h_{ABEF}\Sigma_{CDPQ}\Omega^{PQ}\right)\\
-8\left(1+\Omega\right)\left(2+3\Omega\right)\phi\left(h_{ABCD}\Sigma_{EFPQ}\phi^{PQ}-h_{ABEF}\Sigma_{CDPQ}\phi^{PQ}\right)\bigg\}\\
+\frac{1}{\left(1+\Omega-\phi^2\right)^2}\bigg\{\frac{1}{2}\left(1-\phi^2\right)^2\Omega_{AB}\left(\Omega_{CD}\Sigma_{EF}-\Omega_{EF}\Sigma_{CD}\right)\\
+\Omega\phi\left(2+\Omega-2\phi^2\right)\big[\phi_{AB}\left(\Omega_{CD}\Sigma_{EF}-\Omega_{EF}\Sigma_{CD}\right)\\
+\Omega_{AB}\left(\phi_{CD}\Sigma_{EF}-\phi_{EF}\Sigma_{CD}\right)\big]\\
-2\Omega\left[2\left(1+\Omega^2\right)^2-\left(2+3\Omega\right)\phi^2\right]\phi_{AB}\left(\phi_{CD}\Sigma_{EF}-\phi_{EF}\Sigma_{CD}\right)\\
-\frac{1}{3}\bigg[\frac{}{}\phi^2\left(-1+3\phi^2\right)\Omega^{PQ}\Omega_{PQ}+4\phi\left[3\left(1+\Omega\right)^2-\left(5+6\Omega\right)\phi^2\right]\Omega^{PQ}\phi_{PQ}\\
-4\left(1+\Omega\right)\left[\left(1+\Omega\right)\left(5+6\Omega\right)-\left(7+9\Omega\right)\phi^2\right]\phi^{PQ}\phi_{PQ}\bigg]\\
\left(h_{ABCD}\Sigma_{EF}-h_{ABEF}\Sigma_{CD}\right)\bigg\},
\end{eqnarray*}
which implies
\begin{eqnarray}\label{subsidiary4}
\partial_u\Sigma_{C1AB}+\frac{1}{u}\epsilon_C\,^0\Sigma_{01AB}=2\hat{\Gamma}_{C100}\Sigma_{01AB}-\Omega\epsilon_{0C}H_{10AB}-S_{00AB}\Sigma_{1C}\\
\nonumber-\frac{R}{3}h_{00AB}\Sigma_{1C}-\frac{1}{3}\left(1+\Omega\right)h_{00AB}\Pi_{1C}\\
\nonumber+\frac{1}{6\left(1+\Omega-\phi^2\right)}\bigg\{\frac{}{}3\left[1+\left(1+\Omega\right)\phi^2\right]\left(\Sigma_{1CAB}\Omega_{00}+\Omega_{AB}\Sigma_{1C00}\right)\\
\nonumber+4\left(2+3\Omega\right)\phi^2h_{00AB}\Sigma_{C1PQ}\Omega^{PQ}+8\left(1+\Omega\right)\left(2+3\Omega\right)\phi h_{00AB}\Sigma_{C1PQ}\phi^{PQ}\bigg\}\\
\nonumber+\frac{1}{\left(1+\Omega-\phi^2\right)^2}\bigg\{-\frac{1}{2}\left(1-\phi^2\right)^2\Omega_{AB}\Omega_{00}\Sigma_{1C}+\Omega_{AB}\phi_{00}\Sigma_{1C}\\
\nonumber-\Omega\phi\left(2+\Omega-2\phi^2\right)\phi_{AB}\Omega_{00}\Sigma_{1C}\\
\nonumber+2\Omega\left[2\left(1+\Omega^2\right)^2-\left(2+3\Omega\right)\phi^2\right]\phi_{AB}\phi_{00}\Sigma_{1C}\\
\nonumber+\frac{1}{3}\bigg[\frac{}{}\phi^2\left(-1+3\phi^2\right)\Omega^{PQ}\Omega_{PQ}+4\phi\left[3\left(1+\Omega\right)^2-\left(5+6\Omega\right)\phi^2\right]\Omega^{PQ}\phi_{PQ}\\
\nonumber-4\left(1+\Omega\right)\left[\left(1+\Omega\right)\left(5+6\Omega\right)-\left(7+9\Omega\right)\phi^2\right]\phi^{PQ}\phi_{PQ}\bigg]h_{00AB}\Sigma_{1C}\bigg\}.
\end{eqnarray}
Now with $\Pi_{AB}$
\begin{eqnarray*}
&& D_{CD}\Pi_{AB}-D_{AB}\Pi_{CD}=t_{AB}\,^{EF}\,_{CD}D_{EF}R\\
&& +\frac{1}{1+\Omega-\phi^2}\bigg\{\frac{}{}-2\left(4+7\Omega\right)\left[\phi\Omega^{GH}-2(1+\Omega)\phi^{GH}\right]D_{EF}\phi_{GH}t_{AB}\,^{EF}\,_{CD}\\
&& -4\left(4+7\Omega\right)\left[\phi\Omega^{GH}-2(1+\Omega)\phi^{GH}\right]\phi_{GE}R^{E}\,_{HABCD}\\
&& -\left[\left(3-3\phi^2+7\Omega\phi^2\right)\Omega^{EF}-2\Omega\phi\left(4+7\Omega\right)\phi^{EF}\right]\left(\epsilon_{BC}H_{DAEF}+\epsilon_{AD}H_{BCEF}\right)\\
&& +\frac{1}{3}\phi\left(4+7\Omega\right)\left[\phi\left(\Omega_{CD}\Pi_{AB}-\Omega_{AB}\Pi_{CD}\right)-2(1+\Omega)\left(\phi_{CD}\Pi_{AB}-\phi_{AB}\Pi_{CD}\right)\right]\\
&& -\frac{1}{3}\phi^2\left(4+7\Omega\right)R\left(\Sigma_{ABCD}-\Sigma_{CDAB}\right)+2\left(4+7\Omega\right)\phi\\
&& \left(D_{CD}\phi^{EF}\Sigma_{ABEF}-D_{AB}\phi^{EF}\Sigma_{CDEF}\right)\\
&& +\left(3-3\phi^2+7\Omega\phi^2\right)\left(S_{CD}\,^{EF}\Sigma_{ABEF}-S_{AB}\,^{EF}\Sigma_{CDEF}\right)\bigg\}\\
&& +\frac{1}{\left(1+\Omega-\phi^2\right)^2}\bigg\{-\frac{1}{6}\phi^2\left(-12+40\phi^2+21\Omega\phi^2\right)\Omega^{EF}\Omega_{EF}\\
&&\left(\Sigma_{ABCD}-\Sigma_{CDAB}\right)
\end{eqnarray*}
\begin{eqnarray*}
&& +\frac{2}{3}\phi\left[-18(1+\Omega)+\left(46+61\Omega+21\Omega^2\right)\phi^2\right]\Omega^{EF}\phi_{EF}\left(\Sigma_{ABCD}-\Sigma_{CDAB}\right)\\
&& -\frac{2}{3}(1+\Omega)\left[-24(1+\Omega)+\left(52+61\Omega+21\Omega^2\right)\phi^2\right]\phi^{EF}\phi_{EF}\\
&& \left(\Sigma_{ABCD}-\Sigma_{CDAB}\right)\\
&& +\frac{1}{3}\phi^2\left(-12+40\phi^2+21\Omega\phi^2\right)\Omega^{EF}\left(\Omega_{CD}\Sigma_{ABEF}-\Omega_{AB}\Sigma_{CDEF}\right)\\
&& -\frac{2}{3}\phi\left[18(1+\Omega)+\left(46+61\Omega+21\Omega^2\right)\phi^2\right]\phi^{EF}\left(\Omega_{CD}\Sigma_{ABEF}-\Omega_{AB}\Sigma_{CDEF}\right)\\
&& -\frac{2}{3}\phi\left[12(1+\Omega)+\left(16+61\Omega+21\Omega^2\right)\phi^2\right]\Omega^{EF}\left(\phi_{CD}\Sigma_{ABEF}-\phi_{AB}\Sigma_{CDEF}\right)\\
&& +\frac{4}{3}(1+\Omega)\left[6(1+\Omega)+\left(22+61\Omega+21\Omega^2\right)\phi^2\right]\phi^{EF}\\
&& \left(\phi_{CD}\Sigma_{ABEF}-\phi_{AB}\Sigma_{CDEF}\right)\\
&& +\frac{1}{3}\phi^2\left(-3+7\phi^2\right)R\left(\Omega_{AB}\Sigma_{CD}-\Omega_{CD}\Sigma_{AB}\right)\\
&& -\frac{2}{3}\phi\left[-7(1+\Omega)^2+(11+14\Omega)\phi^2\right]R\left(\phi_{AB}\Sigma_{CD}-\phi_{CD}\Sigma_{AB}\right)\\
&& +2\phi\left(-3+7\phi^2\right)\Omega^{EF}\left(D_{AB}\phi_{EF}\Sigma_{CD}-D_{CD}\phi_{EF}\Sigma_{AB}\right)\\
&& -4\left[-7(1+\Omega)^2+(11+14\Omega)\phi^2\right]\phi^{EF}\left(D_{AB}\phi_{EF}\Sigma_{CD}-D_{CD}\phi_{EF}\Sigma_{AB}\right)\\
&& +\left(-1+\phi^2\right)\left(-3+7\phi^2\right)\Omega^{EF}\left(S_{EFAB}\Sigma_{CD}-S_{EFCD}\Sigma_{AB}\right)\\
&& -2\phi\left[-4-14\Omega-7\Omega^2+2(2+7\Omega)\phi^2\right]\phi^{EF}\left(S_{EFAB}\Sigma_{CD}-S_{EFCD}\Sigma_{AB}\right)\bigg\}\\
&& +\frac{1}{3\left(1+\Omega-\phi^2\right)^3}\bigg\{\frac{1}{2}\phi^2\left[-24+(59+21\Omega)\phi^2+21\phi^4\right]\Omega^{EF}\Omega_{EF}\\
&& \left(\Omega_{AB}\Sigma_{CD}-\Omega_{CD}\Sigma_{AB}\right)\\
&& -2\phi\left[-18(1+\Omega)+(13+19\Omega)\phi^2+(61+42\Omega)\phi^4\right]\Omega^{EF}\phi_{EF}\\
&& \left(\Omega_{AB}\Sigma_{CD}-\Omega_{CD}\Sigma_{AB}\right)\\
&& +2\phi^2\left[-3(1+\Omega)\left(19+14\Omega+7\Omega^2\right)+\left(113+164\Omega+63\Omega^2\right)\phi^2\right]\phi^{EF}\phi_{EF}\\
&& \left(\Omega_{AB}\Sigma_{CD}-\Omega_{CD}\Sigma_{AB}\right)\\
&& -\phi\left[12(1+\Omega)+(-17+19\Omega)\phi^2+(61+42\Omega)\phi^4\right]\Omega^{EF}\Omega_{EF}\\
&& \left(\phi_{AB}\Sigma_{CD}-\phi_{CD}\Sigma_{AB}\right)\\
&& +4\phi^2\left[-3(1+\Omega)\left(9+14\Omega+7\Omega^2\right)+\left(83+164\Omega+63\Omega^2\right)\phi^2\right]\Omega^{EF}\phi_{EF}\\
&& \left(\phi_{AB}\Sigma_{CD}-\phi_{CD}\Sigma_{AB}\right)\\
&& +4\phi(1+\Omega)\left[(1+\Omega)^2(61+42\Omega)-3\left(39+75\Omega+28\Omega^2\right)\phi^2\right]\phi^{EF}\phi_{EF}\\
&& \left(\phi_{AB}\Sigma_{CD}-\phi_{CD}\Sigma_{AB}\right)\bigg\},
\end{eqnarray*}
and we get
\begin{eqnarray}\label{subsidiary5}
&& \partial_u\Pi_{A1}+\frac{1}{u}\epsilon_{A}\,^{0}\Pi_{01}=2\hat{\Gamma}_{A101}\\
\nonumber&& +\frac{1}{1+\Omega-\phi^2}\bigg\{\left[\left(3-3\phi^2+7\Omega\phi^2\right)\Omega^{EF}-2\Omega\phi\left(4+7\Omega\right)\phi^{EF}\right]\epsilon_{0A}H_{10EF}\\
\nonumber&& +\frac{1}{3}\phi\left(4+7\Omega\right)\left[\phi\Omega_{00}-2(1+\Omega)\phi_{00}\right]\Pi_{A1}-\frac{1}{3}\phi^2\left(4+7\Omega\right)R\Sigma_{A100}\\
\nonumber&& \left.\frac{}{}+2\left(4+7\Omega\right)\phi \partial_u\phi^{EF}\Sigma_{A1EF}+\left(3-3\phi^2+7\Omega\phi^2\right)S_{00}\,^{EF}\Sigma_{A1EF}\right\}\\
\nonumber&& +\frac{1}{\left(1+\Omega-\phi^2\right)^2}\bigg\{-\frac{1}{6}\phi^2\left(-12+40\phi^2+21\Omega\phi^2\right)\Omega^{EF}\Omega_{EF}\Sigma_{A100}\\
\nonumber&& +\frac{2}{3}\phi\left[-18(1+\Omega)+\left(46+61\Omega+21\Omega^2\right)\phi^2\right]\Omega^{EF}\phi_{EF}\Sigma_{A100}\\
\nonumber&& -\frac{2}{3}(1+\Omega)\left[-24(1+\Omega)+\left(52+61\Omega+21\Omega^2\right)\phi^2\right]\phi^{EF}\phi_{EF}\Sigma_{A100}\\
\nonumber&& +\frac{1}{3}\phi^2\left(-12+40\phi^2+21\Omega\phi^2\right)\Omega^{EF}\Omega_{00}\Sigma_{A1EF}\\
\nonumber&& -\frac{2}{3}\phi\left[18(1+\Omega)+\left(46+61\Omega+21\Omega^2\right)\phi^2\right]\phi^{EF}\Omega_{00}\Sigma_{A1EF}\\
\nonumber&& -\frac{2}{3}\phi\left[12(1+\Omega)+\left(16+61\Omega+21\Omega^2\right)\phi^2\right]\Omega^{EF}\phi_{00}\Sigma_{A1EF}\\
\nonumber&& +\frac{4}{3}(1+\Omega)\left[6(1+\Omega)+\left(22+61\Omega+21\Omega^2\right)\phi^2\right]\phi^{EF}\phi_{00}\Sigma_{A1EF}\\
\nonumber&& -\frac{1}{3}\phi^2\left(-3+7\phi^2\right)R\Omega_{00}\Sigma_{A1}+\frac{2}{3}\phi\left[-7(1+\Omega)^2+(11+14\Omega)\phi^2\right]R\phi_{00}\Sigma_{A1}\\
\nonumber&& -2\phi\left(-3+7\phi^2\right)\Omega^{EF}\partial_u\phi_{EF}\Sigma_{A1}\\
\nonumber&& +4\left[-7(1+\Omega)^2+(11+14\Omega)\phi^2\right]\phi^{EF}\partial_u\phi_{EF}\Sigma_{A1}\\
\nonumber&& -\left(-1+\phi^2\right)\left(-3+7\phi^2\right)\Omega^{EF}S_{EF00}\Sigma_{A1}\\
\nonumber&& +2\phi\left[-4-14\Omega-7\Omega^2+2(2+7\Omega)\phi^2\right]\phi^{EF}S_{EF00}\Sigma_{A1}\bigg\}\\
\nonumber&& -\frac{1}{3\left(1+\Omega-\phi^2\right)^3}\bigg\{\frac{1}{2}\phi^2\left[-24+(59+21\Omega)\phi^2+21\phi^4\right]\Omega^{EF}\Omega_{EF}\Omega_{00}\Sigma_{A1}\\
\nonumber&& -2\phi\left[-18(1+\Omega)+(13+19\Omega)\phi^2+(61+42\Omega)\phi^4\right]\Omega^{EF}\phi_{EF}\Omega_{00}\Sigma_{A1}\\
\nonumber&& +2\phi^2\left[-3(1+\Omega)\left(19+14\Omega+7\Omega^2\right)+\left(113+164\Omega+63\Omega^2\right)\phi^2\right]\\
\nonumber&& \phi^{EF}\phi_{EF}\Omega_{00}\Sigma_{A1}\\
\nonumber&& -\phi\left[12(1+\Omega)+(-17+19\Omega)\phi^2+(61+42\Omega)\phi^4\right]\Omega^{EF}\Omega_{EF}\phi_{00}\Sigma_{A1}\\
\nonumber&& +4\phi^2\left[-3(1+\Omega)\left(9+14\Omega+7\Omega^2\right)+\left(83+164\Omega+63\Omega^2\right)\phi^2\right]\\
\nonumber&& \Omega^{EF}\phi_{EF}\phi_{00}\Sigma_{A1}
\end{eqnarray}
\begin{eqnarray}
\nonumber&& +4\phi(1+\Omega)\left[(1+\Omega)^2(61+42\Omega)-3\left(39+75\Omega+28\Omega^2\right)\phi^2\right]\\
\nonumber&& \phi^{EF}\phi_{EF}\phi_{00}\Sigma_{A1}\bigg\}.
\end{eqnarray}
Finally, with $H_{ABCD}$,
\begin{eqnarray}\label{subsidiary6}
&& D^{EF}H_{EFCD}=-\frac{1}{2}t^{EFHI}\,_{E}\,^{G}D_{HI}S_{CDFG}-2S_{E(FGC}R^{E}\,_{D)}\,^{HF}\,_{H}\,^{G}\\
\nonumber && +\frac{1}{1+\Omega-\phi^2}\bigg\{\left[\phi\Omega^{EF}-2(1+\Omega)\phi^{EF}\right]\bigg[-\frac{1}{3}(1+\Omega)D_{GH}\phi_{EF}t^{I}\,_{C}\,^{GH}\,_{DI}\\
\nonumber&& -\Omega D_{GH}\phi_{FI}t_{E}\,^{IGH}\,_{CD}-\frac{2}{3}(1+\Omega)\phi_{FG}R^{G}\,_{E}\,^{H}\,_{CDH}+\Omega\phi_{FG}R^{GH}\,_{EHCD}\\
\nonumber&& -\Omega\phi^{GH}R_{GFEHCD}\bigg]+\frac{1}{3}\phi\Omega\left[\phi\Omega^{EF}-2(1+\Omega)\phi^{EF}\right]H_{(CD)EF}\\
\nonumber&&+\frac{1}{2}\left[\left(1-\phi^2+\Omega\phi^2\right)\Omega^{EF}-2\Omega^2\phi\phi^{EF}\right]H_{EFCD}\\
\nonumber&& +\frac{1}{18}(-2+\Omega)\phi\left[\phi\Omega^{E}\,_{(C}\Pi_{D)E}-2(1+\Omega)\phi^{E}\,_{(C}\Pi_{D)E}\right]\\
\nonumber&& +\frac{1}{18}(-2+\Omega)\phi^2 R\Sigma_{E(CD)}\,^{E}+\Omega\phi D_{CD}\phi^{EF}\Sigma_{GEF}\,^{G}\\
\nonumber&& -\frac{2}{3}(1+\Omega)\phi D^{E}\,_{(C}\phi^{FG}\Sigma_{D)EFG}+\frac{1}{2}\left(1-\phi^2+\Omega\phi^2\right)S_{CD}\,^{EF}\Sigma_{GEF}\,^{G}\\
\nonumber&& -\frac{1}{3}\Omega\phi^2 S^{EFG}\,_{(C}\Sigma_{D)EFG}-2\left[\phi\Omega_{CD}-2(1+\Omega)\phi_{CD}\right]\phi^{EF}\Sigma_{GEF}\,^{G}\\
\nonumber&& \left.\frac{}{}-4\phi\Omega^{EF}\phi_{CD}\Sigma_{GEF}\,^{G}+2\phi\Omega^{EF}\phi_{F}\,^{G}\Sigma_{EGCD}\right\}\\
\nonumber&& +\frac{1}{\left(1+\Omega-\phi^2\right)^2}\bigg\{\frac{1}{18}\phi^2\left(3-10\phi^2+6\Omega\phi^2\right)\Omega^{EF}\\
\nonumber&& \left(\Omega_{EF}\Sigma_{G(CD)}\,^{G}+2\Omega^{G}\,_{(C}\Sigma_{D)GEF}\right)\\
\nonumber&& +\frac{2}{9}\phi^3\left(7+4\Omega-6\Omega^2\right)\phi^{EF}\left(\Omega_{EF}\Sigma_{G(CD)}\,^{G}+\Omega^{G}\,_{(C}\Sigma_{D)GEF}\right)\\
\nonumber&& -\frac{2}{9}\phi\left[6(1+\Omega)-\left(13+4\Omega-6\Omega^2\right)\phi^2\right]\Omega^{EF}\phi^{G}\,_{(C}\Sigma_{D)GEF}\\
\nonumber&& -\frac{2}{9}(1+\Omega)\left[3(1+\Omega)+2\left(2+2\Omega-3\Omega^2\right)\phi^2\right]\phi^{EF}\phi_{EF}\Sigma_{G(CD)}\,^{G}\\
\nonumber&& +\frac{4}{9}(1+\Omega)^2\left(3-10\phi^2+6\Omega\phi^2\right)\phi^{EF}\phi^{G}\,_{(C}\Sigma_{D)GEF}\\
\nonumber&& -\frac{1}{18}\phi^2\left(-3+\phi^2\right)R\Omega^{E}\,_{(C}\Sigma_{D)E}-\frac{1}{9}\phi\left[(1+\Omega)^2+(1-2\Omega)\phi^2\right]R\phi^{E}\,_{(C}\Sigma_{D)E}\\
\nonumber&& +\frac{2}{3}\phi^3\Omega_{EF}D^{G}\,_{(C}\phi^{EF}\Sigma_{D)G}-\phi\left(1-\phi^2\right)\Omega^{EF}D_{CD}\phi_{F}\,^{G}\Sigma_{EG}
\end{eqnarray}
\begin{eqnarray}
\nonumber&& +\frac{4}{3}(1+\Omega)\left(1+\Omega-2\phi^2\right)\phi_{EF}D^{G}\,_{(C}\phi^{EF}\Sigma_{D)G}+2\left[(1+\Omega)^2-(1+2\Omega)\phi^2\right]\\
\nonumber&& \phi^{EF}D_{CD}\phi_{F}\,^{G}\Sigma_{EG}-\frac{1}{3}\left(1-\phi^2\right)\phi^2\Omega^{EF}S^{G}\,_{EF(C}\Sigma_{D)G}\\
\nonumber&& +\frac{1}{2}\left(1-\phi^2\right)^2\Omega^{EF}S^{G}\,_{FCD}\Sigma_{EG}\\
\nonumber&& -\frac{2}{3}\phi\left[-(1+\Omega)^2+(1+2\Omega)\phi^2\right]\phi^{EF}S^{G}\,_{EF(C}\Sigma_{D)G}\phi^{EF}S^{G}\,_{FCD}\Sigma_{EG}\\
\nonumber&& +\Omega\phi\left(2+\Omega-2\phi^2\right)+2\phi\left(\Omega_{CD}-2\phi\phi_{CD}\right)\Omega_{E}\,^{F}\phi^{EG}\Sigma_{FG}\bigg\}\\
\nonumber&& +\frac{1}{9\left(1+\Omega-\phi^2\right)^3}\bigg\{\frac{}{}-\phi^2\left(3-13\phi^3+3\Omega\phi^2+3\phi^4\right)\Omega^{EF}\Omega_{EF}\Omega^{G}\,_{(C}\Sigma_{D)G}\\
\nonumber&& -4\phi^3\left(5+8\Omega+2\phi^2-6\Omega\phi^2\right)\Omega^{EF}\phi_{EF}\Omega^{G}\,_{(C}\Sigma_{D)G}\\
\nonumber&& +4\phi^2\left[3(1+\Omega)^3+\left(4-2\Omega-9\Omega^2\right)\phi^2\right]\phi^{EF}\phi_{EF}\Omega^{G}\,_{(C}\Sigma_{D)G}\\
\nonumber&& +2\phi\left[(1+\Omega)\left(3-8\phi^2\right)-2(1-3\Omega)\phi^4\right]\Omega^{EF}\Omega_{EF}\phi^{G}\,_{(C}\Sigma_{D)G}\\
\nonumber&& +8(1+\Omega)\phi^2\left[3\Omega(2+\Omega)+(7-9\Omega)\phi^2\right]\Omega^{EF}\phi_{EF}\phi^{G}\,_{(C}\Sigma_{D)G}\\
\nonumber&& -8(1+\Omega)\phi\left[2(1+\Omega)^2(-1+3\Omega)+3\left(3-4\Omega^2\right)\phi^2\right]\phi^{EF}\phi_{EF}\phi^{G}\,_{(C}\Sigma_{D)G}\bigg\}
\end{eqnarray}
where the l.h.s. is
\begin{eqnarray*}
&& D^{EF}H_{EFCD}=\partial_u H_{11CD}+\frac{1}{u}\left(H_{11CD}+H_{110(C}\epsilon_{D)}\,^{0}\right)\\
&& -\left(\frac{1}{2u}\partial_v +\hat{e}^{a}\,_{01}\partial_a\right)H_{10CD}-2\hat{\Gamma}_{0100}H_{11CD}-\hat{\Gamma}_{010C}H_{110D}\\
&& -\hat{\Gamma}_{010D}H_{110C}+\hat{\Gamma}_{011C}H_{100D}+\hat{\Gamma}_{011D}H_{100C}+\hat{\Gamma}_{1100}H_{10CD}.
\end{eqnarray*}
Equations \eqref{subsidiary1}, \eqref{subsidiary2}, \eqref{subsidiary3}, \eqref{subsidiary4}, \eqref{subsidiary5}, \eqref{subsidiary6} are the system of subsidiary equations for the quantities \eqref{zeroRest}. The expressions on the right hand sides of these equations are homogeneous functions of the quantities \eqref{zeroRest}. Together with Lemma \ref{limitCompleteSystem} this implies that all the expansion coefficients of the quantities \eqref{zeroRest} vanish on $U_0$. As the functions \eqref{zeroRest} are necessarily holomorphic, this implies that they vanish on $\hat{N}$.
\end{proof}

\section{Analyticity at space-like infinity}\label{analyticity}
Our gauge is singular and thus the holomorphic solution of Lemma \ref{holomorphicSolutions} does not cover a full neighbourhood of the point $i$. To show that we can indeed get a holomorphic solution in a hole neighbourhood of $i$ we go to a normal frame field based on the frame $c_{AB}$ at $i$ and the corresponding normal coordinates $x^a$. The argument follows with some modifications the line of the corresponding argument in \cite{Friedrich07}.\\
The geodesic equation for $z^a(u(s),v(s),w(s))$, $D_{\dot{z}}\dot{z}=0$, can be written in the form
\begin{eqnarray*}
&& \dot{z}^a=m^{AB}e^a\,_{AB},\\
&& \dot{m}^{AB}=-2m^{CD}\Gamma_{CD}\,^{(A}\,_Bm^{B)E}.
\end{eqnarray*}
The initial conditions for the geodesics to start at $i$ are
\begin{equation*}
u|_{s=0}=0,\,\,\,\,\,w|_{s=0}=0,
\end{equation*}
and we have to prescribe
\begin{equation*}
v_0=v|_{s=0}=v_0,
\end{equation*}
in order to determine the $\partial_u$-$\partial_w$-plane where the tangent vector is.\\
The components of the tangent vector to the geodesic at $i$ are given by $m^{AB}|_{s=0}=m^{AB}_0$, and by regularity and the geodesic equations we have
\begin{equation*}
m^{00}_0=\dot{u}|_{s=0}\equiv \dot{u}_0,\,\,\,\,\,m^{01}_0=0,\,\,\,\,\,m^{11}_0=\dot{w}|_{s=0}\equiv \dot{w}_0.
\end{equation*}
We can identify the frame $e_{AB}$ with its projection into $T_iN_c$, then $m^{AB}_0e_{AB}=m^{*AB}c_{AB}=x^ac_{\bf a}$, where as defined $c_{AB}=\alpha^a\,_{AB}c_{\bf a}$, and we get
\begin{equation*}
x^1=\tfrac{1}{\sqrt{2}}\left(\dot{w}_0+(v_0^2-1)\dot{u}_0\right),\,\,\,x^2=\tfrac{i}{\sqrt{2}}\left(\dot{w}_0+(v_0^2+1)\dot{u}_0\right),\,\,\,x^3=\sqrt{2}v_0\dot{u}_0,
\end{equation*}
or, inverting the relations
\begin{equation*}
\dot{u}_0(x^a)=-\frac{x^1+ix^2}{\sqrt{2}},\,\,\,v_0(x^a)=-\frac{x^3}{x^1+ix^2},\,\,\,\dot{w}_0=\frac{\delta_{ab}x^ax^b}{\sqrt{2}(x^1+ix^2)}.
\end{equation*}
Here we see that in order to have a well defined vector we need $x^1+ix^2\neq 0$, or, what is the same, $\dot{u}_0\neq 0$. This correspond to the singular generator of ${\cal N}_i$ in the $c_{AB}$-gauge. The vectors $x^ac_{\bf a}$ cover all directions at $i$ except those tangent to the complex null hyperplane $(c_{\bf 1}+ic_{\bf 2})^\perp =\{a(c_{\bf 1}+ic_{\bf 2})+bc_{\bf 3}|a,b\in\mathbb{C}\}$.\\
As we have used a frame formalism, we need also to determine the normal frame centered at $i$ and based on the frame $c_{AB}$ at $i$. As we already have the frame fields $e_{AB}$, we write the equation for the normal frame $c_{AB}$, $D_{\dot{x}}c_{AB}=0$ as an equation for the transformacion $t^A\,_B\in SL(2,\mathbb{C})$ that relates the frames $e_{AB}$ and $c_{AB}$, $c_{AB}=t^C\,_At^D\,_Be_{CD}$. The equation can be written as
\begin{equation}\label{eqTransformation}
\dot{t}^A\,_{B}=-m^{DE}\Gamma_{DE}\,^{A}\,_{C}t^{C}\,_{B},
\end{equation}
and the initial condition cames from having to take $e_{AB}|_i(v)$ to $c_{AB}|_i$,
\begin{equation}\label{initialCondEqTransformation}
t^A\,_B|_{s=0}=s^A\,_B(-v_0).
\end{equation}
Following the proofs of Lemma 7.1, Lemma 7.2 and Lemma 7.3 in \cite{Friedrich07} we arrive at the following two lemmas.
\begin{lem}
For any given initial data $\dot{u}_0$, $v_0$, $\dot{w}_0$, with $\dot{u}_0\neq 0$, there exist a number $t=t(\dot{u}_0,v_0,\dot{w}_0)$ and unique holomorphic solutions $z^j(s)=z^j(s,\dot{u}_0,v_0,\dot{w}_0)$ of the initial value problem for the geodesic equations with initial conditions as described above which is defined for $|s|<1/t$. The functions $z^j(s,\dot{u}_0,v_0,\dot{w}_0)$ are in fact holomorphic functions of all four variables $(s,\dot{u}_0,v_0,\dot{w}_0)$  in a certain $P_{1/t}(0)\times U$, where $U$ is a compactly embedded subset of $(\mathbb{C}\backslash\{0\})\times\mathbb{C}\times\mathbb{C}$.
\end{lem}
\begin{lem}
Along te geodesic corresponding to $s\rightarrow z^j(s,\dot{u}_0,v_0,\dot{w}_0)$ equations \eqref{eqTransformation} have a unique holomorphic solution $t^A\,_B(s)=t^A\,_B(s,\dot{u}_0,v_0,\dot{w}_0)$ satisfying the initial conditions \eqref{initialCondEqTransformation}. The functions $t^A\,_B(s)=t^A\,_B(s,\dot{u}_0,v_0,\dot{w}_0)$ are holomorphic in all four variables in the domain where the $z^j(s,\dot{u}_0,v_0,\dot{w}_0)$ are holomorphic.
\end{lem}
Following the discussion in \cite{Friedrich07} it can be seen that, as $|x|\equiv \sqrt{\delta_{ab}\bar{x}^ax^b}\rightarrow 0$, $x^1+ix^2\neq 0$,
\begin{eqnarray*}
&& u(x^c)=-\frac{x^1+ix^2}{\sqrt{2}}+{\cal O}(|x|^3),\\
&& v(x^c)=-\frac{x^3}{x^1+ix^2}+{\cal O}(|x|^2),\\
&& w(x^c)=\frac{\delta_{ab}x^ax^b}{\sqrt{2}(x^1+ix^2)}+{\cal O}(|x|^3).
\end{eqnarray*}
This gives for the forms $\chi^{AB}=\chi^{AB}\,_cdx^c$ dual to the normal frame $c_{AB}$
\begin{equation*}
\chi^{AB}(x^c)=\left(\alpha^{AB}\,_a+\hat{\chi}^{AB}\,_a\right)dx^a,
\end{equation*}
with holomorphic functions $\hat{\chi}^{AB}\,_a(x^c)$ which satisfy $\hat{\chi}^{AB}\,_a={\cal O}(|x|^2)$ as $|x|\rightarrow 0$. Also the coefficients $c^a\,_{AB}=\langle dx^a,c_{AB}\rangle $ of the normal frame in the normal coordinates satisfy
\begin{equation*}
c^a\,_{AB}(x^c)=\alpha^a\,_{AB}+\hat{c}^a\,_{AB},
\end{equation*}
with holomorphic  functions $\hat{c}^a\,_{AB}(x^c)$ which satisfy $\hat{c}^a\,_{AB}={\cal O}(|x|^2)$ as $|x|\rightarrow 0$.\\
The three 1-forms $\alpha^a\,_{AB}dx^a$ are linearly independent and thus for small $|x^c|$ the coordinate transformation $x^a\rightarrow z^a(x^c)$, where defined, is nondegenarate. This means that all the tensor fields entering the conformal stationary vacuum field equations can be expressed in term of the normal coordinates $x^c$ and the normal frame field $c_{AB}$.\\
Now we can derive our main result.
\begin{proof}[Proof of Theorem \ref{mainThm}]
The coordinates $x^a$ cover a domain $U$ in $\mathbb{C}^3$ on which the frame vector fields $c_{AB}=c^a\,_{AB}\partial/\partial_{x^a}$ exist, are linearly independent and holomorphic. Also in $U$ the other tensor fields expressed in terms of the $x^a$ and $c_{AB}$ are holomorphic. However $U$ does not contain the hypersurface $x^1+ix^2=0$ but the boundary of $U$ becames tangent to this hypersurface at $x^a=0$.\\
We want to see that the solution indeed cover a domain containing an open neighbourhood of the origin.\\
We still have the gauge freedom to perform with some $t^A\,_B\in SU(2)$ a rotation $\delta^*\rightarrow\delta^*\cdot t$ of the spin frame. Whit this rotation is associated the rotation
\begin{equation*}
c_{AB}\rightarrow c^t_{AB}=t^C\,_At^D\,_Bc_{CD}
\end{equation*}
of the frame $c_{AB}$ at $i$. The construction of the submanifold $\hat{N}$ was done based on the frame $c_{AB}$, starting now with $c^t_{AB}$ all the previous constructions and derivations can be repeated as far as the estimates for the null data in the $c_{AB}$-gauge can be translated to the same type of estimates for the null data in the $c^t_{AB}$-gauge.\\
We will denote $u'$, $v'$, $w'$ and $e^t_{AB}$ the analogues in the new gauge of the coordinates $u$, $v$, $w$ and the frame $e_{AB}$. The set ${\cal N}_i$ is invariant under this rotation. The sets $\{w=0\}$ and $\{w'=0\}$ are both lifts of the set ${\cal N}_i$ to the bundle of spin frames. The coordinates $u$ and $u'$ are both affine parameters on the null generators of ${\cal N}_i$, which vanish at $i$. The coordinats $v$, $v'$ both label the null generators of ${\cal N}_i$. The frame vectors $e_{00}$ and $e^t_{00}$ are auto-parallel vector fields tangent to the null generators.\\
If $v$ and $v'$ label the same generator $\eta$ of ${\cal N}_i$, then $e^t_{00}(v')=f^2e_{00}(v)$ at  $i$, with some $f\neq 0$. Furthermore, as $e_{00}$ and $e^t_{00}$ are auto-parallel, then $e^t_{00}=f^2e_{00}$ must hold along $\eta$, with $f$ constant along the geodesic. This means that at $i$
\begin{equation*}
s^C\,_0(v')s^D\,_0(v')t^E\,_Ct^F\,_Dc_{EF}=f^2s^C\,_0(v)s^D\,_0(v)c_{CD},
\end{equation*}
and absorbing the undetermined sign in $f$,
\begin{equation}\label{relf}
t^E\,_Cs^C\,_0(v')=fs^E\,_0(v).
\end{equation}
We can write $t^A\,_B\in SU(2)$ as
\begin{equation}\label{tTansformation}
(t^A\,_B)=\left(\begin{array}{cc}a&-\bar{c}\\c&\bar{a}\end{array}\right),\,\,\,\,\,a,c\in\mathbb{C},\,\,\,|a|^2+|c|^2=1.
\end{equation}
This gives with \eqref{relf}
\begin{equation}\label{relNullGenerators}
v'=\frac{-c+av}{\bar{a}+\bar{c}v},\,\,\,f=\frac{1}{\bar{a}+\bar{c}v},\,\,\,\text{resp.}\,\,\,\,\,v=\frac{c+\bar{a}v'}{a-\bar{c}v'},\,\,\,f=a-\bar{c}v'.
\end{equation}
As $\langle du,e_{00}\rangle=1=\langle du',e^t_{00}\rangle$ we have for the affine parameter along $\eta$
\begin{equation}\label{relAffineParameter}
u=f^2u'.
\end{equation}
With \eqref{relNullGenerators}, \eqref{relAffineParameter} holds $\eta(u',v')=\eta(u,v)$.\\
If $c\neq 0$ then $v\rightarrow\infty$ as $v'\rightarrow a/\bar{c}$. So the null generator in the $c_{AB}$-gauge, where we need information, is contained, excepting the origin, in the regular domain of the $c^t_{AB}$-gauge.\\
Let us consider now the abstract null data given in the $c_{AB}$-gauge $\hat{{\cal D}}^{\phi}_n$, $\hat{{\cal D}}^S_n$ satisfying estimates of the form \eqref{estimates-psi}, \eqref{estimates-Psi}. In the $c^t_{AB}$-gauge we have $\hat{{\cal D}}^{\phi t}_n$, $\hat{{\cal D}}^{St}_n$, with terms given by
\begin{equation*}
\psi^t_{A_mB_m...A_1B_1}=t^{G_m}\!_{A_m}t^{H_m}\!_{B_m}...t^{G_1}\!_{A_1}t^{H_1}\!_{B_1}\psi_{G_mH_m...G_1H_1},
\end{equation*}
\begin{equation*}
\Psi^t_{A_mB_m...A_1B_1CDEF}=t^{G_m}\!_{A_m}t^{H_m}\!_{B_m}...t^{G_1}\!_{A_1}t^{H_1}\!_{B_1}t^{I}\!_{C}t^{J}\!_{D}t^{K}\!_{E}t^{K}\!_{L}\Psi_{G_mH_m...G_1H_1IJKL}.
\end{equation*}
Using the essential components of $\psi$ and $\psi^t$
\begin{eqnarray*}
\psi^t_{(A_mB_m...A_1B_1)_n} & = & \sum_{j=0}^{2m}\binom{2m}{j}t^{(G_m}\!_{(A_m}t^{H_m}\!_{B_m}...t^{G_1}\!_{A_1}t^{H_1)_j}\!_{B_1)_n}\psi_{(G_mH_m...G_1H_1)_j}\\
& = & \binom{2m}{n}^{-\frac{1}{2}}\sum_{j=0}^{2m}\binom{2m}{j}^\frac{1}{2}T_{2m}\,^j\,_n(t)\psi_{(G_mH_m...G_1H_1)_j}.
\end{eqnarray*}
The numbers
\begin{equation*}
T_{2m}\,^j\,_n(t)=\binom{2m}{n}^\frac{1}{2}\binom{2m}{j}^\frac{1}{2}t^{(G_m}\!_{(A_m}t^{H_m}\!_{B_m}...t^{G_1}\!_{A_1}t^{H_1)_j}\!_{B_1)_n}
\end{equation*}
satisfy
\begin{equation*}
|T_{2m}\,^j\,_n(t)|\leq 1,\,\,\,\,\,m=0,1,2,...,\,\,\,\,\,0\leq j\leq 2m,\,\,\,\,\,0\leq n\leq 2m,
\end{equation*}
as they represent the matrix elements of a unitary representation of $SU(2)$.
So we get
\begin{equation*}
|\psi^t_{A_mB_m...A_1B_1}|\leq \frac{m!M}{r'^m},\,\,m=1,2,3,...,
\end{equation*}
where $r'=r/4$.\\
In the same way we get
\begin{equation*}
|\Psi^t_{A_mB_m...A_1B_1CDEF}|\leq \frac{m!M'}{r'^m},\,\,m=0,1,2,3,...,
\end{equation*}
where $M'=16M$.\\
So the estimates for the null data on the $c_{AB}$-gauge translate into the same type of estimates for the null data on the $c^t_{AB}$-gauge.\\
Assuming now $c\neq 0$ in \eqref{tTansformation}, we have two possibilities for getting the solution in the $c^t_{AB}$-gauge:
\begin{enumerate}
\item{Using the solution in the $c_{AB}$-gauge we can determine, where possible, the coordinate and frame transformation to the $c^t_{AB}$-gauge. In particular, the singular generator of ${\cal N}_i$ in the $c^t_{AB}$-gauge will coincide with the regular generator of ${\cal N}_i$ in the $c_{AB}$ gauge on which $v=-\bar{a}/\bar{c}$. We are thus able to determine near the singular generator in the $c^t_{AB}$-gauge the expansion of the solution in terms of the coordinates $u'$, $v'$, $w'$ and the frame field $e^t_{AB}$.}\label{firstMethod}
\item{Using the null data $\hat{{\cal D}}^{\phi t}_n$, $\hat{{\cal D}}^{St}_n$ in the $c^t_{AB}$-gauge, one can repeat all the steps of the previous sections to show the existence of a solution to the conformal stationary vacuum field equations in the coordinates $u'$, $v'$, $w'$ of the $c^t_{AB}$-gauge. All the statements made about the solution in the $c_{AB}$-gauge apply also to this solution, in particular statements about domains of convergence.}\label{secondMethod}
\end{enumerate}
The formal expansions of the fields in terms of $u'$, $v'$, $w'$ are uniquely determined by the data $\hat{{\cal D}}^{\phi t}_n$, $\hat{{\cal D}}^{St}_n$, thus the solutions obtained by the two methods are holomorphically related to each other on certain domains by the gauge transformation obtained in (\ref{firstMethod}). As done with the solution in the $c_{AB}$-gauge, the solution in the $c^t_{AB}$-gauge can be expressed in terms of the normal coordinates $x_t^a$ and the normal frame field $c^t_{AB}$. The $x^a_t$ cover a certain domain $U_t\in\mathbb{C}^3$ and the frame field $c^t_{AB}$ is non-degenerate. All the tensor fields expressed in terms of $x^a_t$ and $c^t_{AB}$ are holomorphic on $U_t$. Then the solution in the $c_{AB}$-gauge and the solution in the $c^t_{AB}$-gauge are related on certain domains by the transformation
\begin{equation*}
x_t^a=t^{-1a}\,_bx^b,\,\,\,\,\,c^t_{AB}=t^C\,_At^D\,_Bc_{CD},
\end{equation*}
which gives the transformation corresponding to the rotation of normal coordinates. We can extend this as a coordinate and frame transformation to the solution obtained in (\ref{secondMethod}) to express all fields in terms of $x^a$ and $c_{AB}$. With this extension all fields are defined and holomorphic on $t^{-1}U_t$. Then the solution obtained in the $c_{AB}$-gauge and the solution in the $c^t_{AB}$-gauge are genuine holomorphic extensions of each other, as one covers the singular generator of the other one away from the origin in a regular way.\\
Let now $x_*^a\neq 0$ be an arbitraty point in $\mathbb{C}^3$. We want to show that the solution extends in the coordinates $x^a$ to a domain which covers the set $sx_*^a$ for $0<s<\epsilon$ for some $\epsilon>0$. That is the case in the $c_{AB}$-gauge as far as $x_*^a\neq (\alpha,i\alpha,\beta)$, $\alpha,\beta\in\mathbb{C}$. We need to see what happens if $x_*^a=(\alpha,i\alpha,\beta)$, with $\alpha\neq 0$ or $\beta\neq 0$.\\
If $x_*^a=(\alpha,i\alpha,\beta)$ and $\alpha\neq 0$, we consider the $c^{t'}_{AB}$-gauge, where $t'_{AB}$ is given by \eqref{tTansformation} with $a=0$, $c=1$. The normal coordinates in the two gauges are related by
\begin{equation*}
x^1_{t'}=-x^1,\,\,\,\,\,x^2_{t'}=x^2,\,\,\,\,\,x^3_{t'}=-x^3.
\end{equation*}
The holomorphic transformation $(x^1_{t'},x^2_{t'},x^3_{t'})\rightarrow(-x^1,x^2,-x^3)$ maps $U_{t'}$ onto a subset of $\mathbb{C}^3$, denoted by $t'^{-1}U_{t'}$, which has nonempty intersection with $U$. After the transformation the two solutions coincide on $t'^{-1}U_{t'}\cap U_t$.\\
Under this transformation, the singular set $\{x^1+ix^2=0\}$ in the $c_{AB}$-gauge correspond to the set $\{x^1_{t'}-ix^2_{t'}=0\}$, which is covered in a regular way in a neighbourhood of $i$ in the $c^{t'}_{AB}$-gauge. So the set $t'^{-1}U_{t'}\cup U_t$ admits a holomorphic extension of our solution in the coordinates $x^a$ and the frame $c_{AB}$. In this extension there exist $\epsilon$ such that $sx_*^a$, $x_*^a=(\alpha,i\alpha,\beta)$ with $\alpha\neq 0$, is covered by the solution for $0<s<\epsilon$.\\
We need also to consider the case $\alpha=0$, that is, $x_*^a=(0,0,\beta)$, $b\neq 0$. In this case we use the $c^{t''}_{AB}$-gauge, where $t''_{AB}$ is given by \eqref{tTansformation} with $a=\tfrac{1}{\sqrt{2}}$, $c=\tfrac{i}{\sqrt{2}}$. The normal coordinates are related by
\begin{equation*}
x^1_{t''}=x^1,\,\,\,\,\,x^2_{t''}=-x^3,\,\,\,\,\,x^3_{t''}=x^2.
\end{equation*}
The argument follows the same lines as for the $a\neq 0$ case.\\
Thus the set $U$ can be extended so that the points $sx_*^a$ with $0<s<\epsilon$ are covered by $U$ and all fields are holomorphic on $U$ in the coordinates $x^a$. Then it can be assumed $U$ to contain a punctured neighbourhood of the origin in which the solution is holomorphic in the normal coordinates $x^a$ and the normal frame $c_{AB}$. Then the solution is in fact holomorphic on a full neighbourhood of the origin $x^a=0$, which represents the point $i$, as holomorphic functions in more than one dimension cannot have isolated singularities.\\
By Lemma \ref{formalExpansion} we have from null data satisfying the reality conditions a formal expansion of the solution with expansion coefficients satisfying the reality conditions. By the various uniqueness statemets obtained in the lemmas, this expansion must coincide with the expansion in normal coordinates of the solution obtained above. This implies the existence of a 3-dimensional real slice on which the tensor fields satisfy the reality conditions. It is obtained by requiring the coordinates $x^a$ to assume values in $\mathbb{R}^3$.
\end{proof}

\section{Conclusions}
We have seen how to determine a formal expansion of an asymptotically flat stationary vacuum solution to Einstein's field equations using a minimal set of freely specifyable data, the null data. This data are given by two sequences of symmetric trace free tensors at space-like infinity. We have obtained necessary and sufficient conditions on the null data for the formal expansion to be absolutely convergent, hence showing that the null data characterize all asymptotically flat stationary vacuum solutions to the field equations.\\
This work contains the static case as a particular case, and is ageneralization of Friedrich's work \cite{Friedrich07} from the static to the stationary case.\\
In relation with the works of Corvino and Schoen \cite{CorvinoSchoen06} and Chru\'siel and Delay \cite{ChruscielDelay03}, where they are able to deform given vacuum initial data in an annulus that encompasses the asymptotic end in order to glue that data to an asymptotically flat vacuum stationary solution of the field equations, our result shows that the null data provides a complete survey of all the asymptotics that can be attained. In particular, for performing the gluing they need families of solutions, it would be interesting to see what are the restriction imposed on the null data in order to form one of these families.\\
It is a long standing conjecture that Hansen`s multipoles \cite{Hansen74}, which are relevant because they have nice geometrical transformation properties under change of conformal factor, do characterize an asymptotically flat stationary vacuum solutions to the field equations in the way we have shown the null data do. This have been shown in the axisymmetric case \cite{Backdahl07} and some steps have been achieved in the general case, like showing that the multipoles determine a formal expansion of a solution \cite{BeigSimon81} \cite{Kundu81}, or necessary bounds on the multipoles if the solution exist \cite{BackdahlHerberthson06}, but general conditions on the multipoles for the expansion to be convergent has not been found yet. As there is a bijective correspondence between the null data and Hansen's multipoles, although the relation is highly non linear, it would be nice if this correspondence could be exploited to get necessary and sufficient conditions on the multipoles to determine a convergent expansion.

\section*{Acknowledgement}
I would like to thank Helmut Friedrich for presenting me this problem and for guidance during this work. The author is supported by a PhD scholarship from the International Max Planck Research School.


\end{document}